\documentclass[10pt,letterpaper]{article}

\usepackage[letterpaper, top=.95in, bottom=1in, left=1.2in, right=1.2in, marginparwidth=2.4cm, marginparsep=1mm, includefoot]{geometry}

\usepackage{latexsym,
amsthm,
color,
amssymb,
url,
bm,
cite,
amssymb,
microtype,
amsmath,
amsfonts,
mathtools}
\usepackage{verbatim,listings}
\usepackage{todonotes}
\usepackage{caption}
\usepackage{setspace}
\usepackage{subcaption}
\usepackage{setspace}
\usepackage[pdftex, plainpages = false, pdfpagelabels, 
                 bookmarks = true,
                 bookmarksopen = true,
                 bookmarksnumbered = true,
                 breaklinks = true,
                 linktocpage,
                 pagebackref,
                 colorlinks = true,  
                 linkcolor = blue,
                 urlcolor  = blue,
                 citecolor = red,
                 anchorcolor = green,
                 hyperindex = true,
                 hyperfigures
                 ]{hyperref}                   
\usepackage[inline]{enumitem}
\usepackage[mathscr]{euscript}
\usepackage{amsthm}
\usepackage{nicefrac}
\usepackage{dsfont}
\usepackage{tikz}
\usepackage[ruled,linesnumbered]{algorithm2e}
\usepackage{nicefrac}
\usepackage[mathscr]{euscript}
\graphicspath{{./figures/}}
\RequirePackage[T1]{fontenc}
\RequirePackage[utf8]{inputenc}
\usepackage[english]{babel}
\usepackage{alphabeta}
\usetikzlibrary{math}
\makeatletter
\input{colordvi}
\usepackage{aliascnt}
\hypersetup{
    colorlinks,
    linkcolor={red!75!black},
    citecolor={green!75!black},
    urlcolor={blue!75!black},
}
\definecolor{MidnightBlack}{rgb}{0.1,0.1,.34}
\definecolor{MidnightBlue}{rgb}{0.1,0.1,0.43}
\definecolor{Black}{rgb}{0,0, 0}
\definecolor{Blue}{rgb}{0, 0 ,1}
\definecolor{Red}{rgb}{1, 0 ,0}
\definecolor{White}{rgb}{1, 1, 1}
\definecolor{grey}{rgb}{.6, .6, .6}
\definecolor{Mygreen}{rgb}{.0, .7, .0}
\definecolor{Yellow}{rgb}{.55,.55,0}
\definecolor{Mustard}{rgb}{1.0, 0.86, 0.35}
\definecolor{applegreen}{rgb}{0.55, 0.71, 0.0}
\definecolor{darkturquoise}{rgb}{0.0, 0.81, 0.82}
\definecolor{celestialblue}{rgb}{0.29, 0.59, 0.82}
\definecolor{green_yellow}{rgb}{0.68, 1.0, 0.18}
\definecolor{crimsonglory}{rgb}{0.75, 0.0, 0.2}
\definecolor{darkmagenta}{rgb}{0.30, 0.0, 0.30}
\definecolor{magenta}{rgb}{0.50, 0.0, 0.50}
\definecolor{internationalorange}{rgb}{1.0, 0.31, 0.0}
\definecolor{darkorange}{rgb}{1.0, 0.55, 0.0}
\definecolor{ao}{rgb}{0.0, 0.5, 0.0}
\definecolor{awesome}{rgb}{1.0, 0.13, 0.32}
\definecolor{darkcyan}{rgb}{0.0, 0.50, 0.50}
\definecolor{violet}{rgb}{0.93, 0.51, 0.93}
\definecolor{brown}{rgb}{0.65, 0.16, 0.16}
\definecolor{orange}{rgb}{1.0, 0.65, 0.0}
\definecolor{cornflowerblue}{rgb}{0.39, 0.58, 0.93}

\newcommand{\cref}[1]{\autoref{#1}}

\newcommand{\remove}[1]{}

\newcounter{func}
\newcommand{\newfun}[1]{f_{\refstepcounter{func}\llabel{#1}\thefunc}} 
\newcommand{\funref}[1]{\hyperref[#1]{f_{\ref*{#1}}}} 

\tikzset{black node/.style={draw, circle, fill = black, minimum size = 5pt, inner sep = 0pt}}
\tikzset{white node/.style={draw, circlternary_treese, fill = white, minimum size = 5pt, inner sep = 0pt}}
\tikzset{normal/.style = {draw=none, fill = none}}
\tikzset{lean/.style = {draw=none, rectangle, fill = none, minimum size = 0pt, inner sep = 0pt}}
\usetikzlibrary{decorations.pathreplacing}
\usetikzlibrary{arrows.meta}
\usetikzlibrary{calc}

\usetikzlibrary{shapes}
\tikzset{diam/.style={draw, diamond, fill = black, minimum size = 7pt, inner sep = 0pt}}

\tikzset{
	position/.style args={#1:#2 from #3}{
		at=($(#3)+(#1:#2)$)
	}
}

\tikzset{
  v:main/.style = {draw, circle, scale=0.8, thick,fill=black,inner sep=0.7mm},
  v:ghost/.style = {inner sep=0pt,scale=1},
  v:marked/.style = {circle, scale=1.3, fill=DarkGoldenrod,opacity=0.4},
  >={latex},
  e:main/.style = {line width=1pt}
}

\usepackage{ 
boxedminipage,
etoolbox,
framed,
ifthen,
tabularx,
tcolorbox,
thm-restate,
tikz,
xifthen,
cite
}
\usetikzlibrary{calc}

\usepackage{thmtools}
\usepackage{thm-restate}

\newlength{\RoundedBoxWidth}
\newsavebox{\GrayRoundedBox}
\newenvironment{GrayBox}[1]%
   {\setlength{\RoundedBoxWidth}{.93\textwidth}
    \def\boxheading{#1}
    \begin{lrbox}{\GrayRoundedBox}
       \begin{minipage}{\RoundedBoxWidth}}%
   {   \end{minipage}
    \end{lrbox}
    \begin{center}
    \begin{tikzpicture}%
       \node(Text)[draw=black!20,fill=white,rounded corners,%
             inner sep=2ex,text width=\RoundedBoxWidth]%
             {\usebox{\GrayRoundedBox}};
        \coordinate(x) at (current bounding box.north west);
        \node [draw=white,rectangle,inner sep=3pt,anchor=north west,fill=white]
        at ($(x)+(6pt,.75em)$) {\boxheading};
    \end{tikzpicture}
    \end{center}}
\newenvironment{defproblemx}[2][]{\noindent\ignorespaces%
                                \FrameSep=6pt%
                                \parindent=0pt%
                \vspace*{-1.5em}
                \ifthenelse{\isempty{#1}}{%
                  \begin{GrayBox}{\textsc{#2}}%
                }{%
                  \begin{GrayBox}{\textsc{#2}  parameterized by~{#1}}%
                }
                \begin{tabular*}{\textwidth}{@{\hspace{.1em}} >{\itshape} p{1.8cm} p{0.8\textwidth} @{}}%
            }{
                \end{tabular*}%
                \end{GrayBox}%
                \ignorespacesafterend
            }
\newcommand{\defproblem}[3]{
  \begin{defproblemx}{#1}
    Input:  & #2 \\
    Question: & #3
  \end{defproblemx}
}%

\newcommand{\bN}{\mathbb{N}}

\newcommand{\bR}{\mathbb{R}}

\newcommand{\Acal}{\mathcal{A}}
\newcommand{\Bcal}{\mathcal{B}}

\newcommand{\Fcal}{\mathcal{F}}

\newcommand{\Hcal}{\mathcal{H}}
\newcommand{\Ical}{\mathcal{I}}

\newcommand{\Lcal}{\mathcal{L}}
\newcommand{\Mcal}{\mathcal{M}}

\newcommand{\Ocal}{\mathcal{O}}
\newcommand{\Pcal}{\mathcal{P}}
\newcommand{\Qcal}{\mathcal{Q}}
\newcommand\Rcal{\mathcal{R}}
\newcommand{\Scal}{\mathcal{S}}
\newcommand{\Tcal}{\mathcal{T}}

\newcommand{\Wcal}{\mathcal{W}}

\RequirePackage{stmaryrd}
\usepackage{textcomp}
\DeclareUnicodeCharacter{2286}{\subseteq}
\DeclareUnicodeCharacter{2192}{\ifmmode\to\else\textrightarrow\fi}
\DeclareUnicodeCharacter{2203}{\ensuremath\exists}
\DeclareUnicodeCharacter{183}{\cdot}
\DeclareUnicodeCharacter{2200}{\forall}
\DeclareUnicodeCharacter{2264}{\leq}
\DeclareUnicodeCharacter{2265}{\geq}
\DeclareUnicodeCharacter{8614}{\mathbin{\mapsto}}
\DeclareUnicodeCharacter{8656}{\Leftarrow}
\DeclareUnicodeCharacter{8657}{\Uparrow}
\DeclareUnicodeCharacter{8658}{\Rightarrow}
\DeclareUnicodeCharacter{8659}{\Downarrow}
\DeclareUnicodeCharacter{8669}{\rightsquigarrow}
\newcommand{\eqdef}{\stackrel{{\scriptsize\rm def}}{=}}
\DeclareUnicodeCharacter{8797}{\eqdef}
\DeclareUnicodeCharacter{8870}{\vdash}
\DeclareUnicodeCharacter{8873}{\Vdash}
\DeclareUnicodeCharacter{22A7}{\models}
\DeclareUnicodeCharacter{9121}{\lceil}
\DeclareUnicodeCharacter{9123}{\lfloor}
\DeclareUnicodeCharacter{9124}{\rceil}
\DeclareUnicodeCharacter{2208}{\in}
\DeclareUnicodeCharacter{9126}{\rfloor}
\DeclareUnicodeCharacter{9655}{\triangleright}
\DeclareUnicodeCharacter{9665}{\triangleleft}
\DeclareUnicodeCharacter{9671}{\diamond}
\DeclareUnicodeCharacter{9675}{\circ}
\DeclareUnicodeCharacter{10178}{\bot}
\DeclareUnicodeCharacter{10214}{} 
\DeclareUnicodeCharacter{10215}{} 
\DeclareUnicodeCharacter{10229}{\longleftarrow}
\DeclareUnicodeCharacter{10230}{\longrightarrow}
\DeclareUnicodeCharacter{10231}{\longleftrightarrow}
\DeclareUnicodeCharacter{10232}{\Longleftarrow}
\DeclareUnicodeCharacter{10233}{\Longrightarrow}
\DeclareUnicodeCharacter{10234}{\Longleftrightarrow}
\DeclareUnicodeCharacter{10236}{\longmapsto}
\DeclareUnicodeCharacter{10238}{\Longmapsto} 
\DeclareUnicodeCharacter{10503}{\Mapsto}    
\DeclareUnicodeCharacter{10971}{\mathrel{\not\hspace{-0.2em}\cap}}
\DeclareUnicodeCharacter{65294}{\ldotp}
\DeclareUnicodeCharacter{65372}{\mid}

\definecolor{Red}{rgb}{1, 0 ,0}
\definecolor{Blue}{rgb}{0, 0 ,1}

\newtheorem{theorem}{Theorem}[section]
\newtheorem{observation}{Observation}[section]
\newtheorem{proposition}{Proposition}[section]
\newtheorem{claim}{Claim}[section]
\newtheorem{lemma}{Lemma}[section]

\newcommand{\hh}{\end{document}}

\newcommand{\obs}{{\sf obs}}
\newcommand{\exc}{{\sf exc}}

\newcommand{\tw}{{\sf tw}\xspace}

\newcommand{\cupall}{{\pmb{\bigcup}}}

\newenvironment{cproof}{\proof[Proof of claim]}{\endproof}

\newcommand{\poly}{\text{$\mathsf{poly}$}\xspace}

\newcommand{\FPT}{\textsf{FPT}\xspace}

\newcommand{\no}{{\sf no}\xspace}

\newcommand{\bd}{{\sf bd}\xspace}
\newcommand{\NP}{{\sf NP}\xspace}
\newcommand{\odd}{{\sf odd}}
\newcommand{\flaps}{{\sf Flaps}}
\newcommand{\compass}{{\sf Compass}}
\newcommand{\influence}{{\sf influence}}
\newcommand{\apb}{{\sc $\Lcal$-AR-$\exc(\Fcal)$}\xspace}
\newcommand{\apbpl}{{\sc $\Lcal$-AR-$\Pcal$}\xspace}

\newcommand{\yes}{{\sf yes}\xspace}

\newcommand{\numen}[1]{\ifthenelse{\not\equal{#1}{1}}{#1}{}}

\usepackage{color}

\definecolor{vagelisColour}{RGB}{0, 65, 130}

\newcommand{\llabel}[1]{\label{#1}}

\title{Graph modification of bounded size to minor-closed classes as fast as vertex deletion\thanks{The first and the third author were supported by  the French-German Collaboration ANR/DFG Project UTMA (ANR-20-CE92-0027) and the {Franco}-Norwegian project PHC AURORA 2024-25 (Projet n° 51260WL). The first and the second author were supported by the French ANR project ELIT (ANR-20-CE48-0008).
The third author was also supported by the ANR project GODASse ANR-24-CE48-4377. An extended abstract of this article appeared in the {\em Proceedings of the 33rd Annual European Symposium on Algorithms (\textbf{ESA}), volume 351 of LIPIcs, pages 7:1--7:18, \textbf{2025}}.}}

\author{Laure Morelle\thanks{LIRMM, Université de Montpellier, CNRS, Montpellier, France.}~$^{,}$\thanks{now in Department of Informatics, University of Bergen, Norway.\\ Emails:  \texttt{laure.morelle@uib.no}, \texttt{ignasi.sau@lirmm.fr},
\texttt{sedthilk@thilikos.info}\,.}\and Ignasi Sau\footnotemark[2]\and Dimitrios M. Thilikos\footnotemark[2]}

\date{\empty}

\begin{document}

\maketitle

\begin{abstract}
\noindent A {\em replacement action} is a function $\Lcal$ that maps each graph $H$ to a collection of graphs of size at most $|V(H)|$.
Given a graph class $\Hcal$, we consider a general family of graph modification problems, called {\sc $\Lcal$-Replacement to $\Hcal$}, where the input is a graph $G$ and the question is whether it is possible to replace some induced subgraph $H_1$  of $G$ on at most $k$ vertices by a graph $H_2$ in $\Lcal(H_1)$ so that the resulting graph belongs to $\Hcal$.
{\sc $\Lcal$-Replacement to $\Hcal$} can simulate many graph modification problems including vertex deletion, edge deletion/addition/edition/contraction, vertex identification, subgraph complementation, independent set deletion, (induced) matching deletion/contraction, etc.
We present two algorithms.
The first one solves {\sc $\Lcal$-Replacement to $\Hcal$} in time $2^{\poly(k)}\cdot |V(G)|^2$ for every minor-closed graph class $\Hcal$, where \poly is a polynomial whose degree depends on $\Hcal$, under a mild technical condition on $\Lcal$.
This generalizes the results of Morelle, Sau, Stamoulis, and Thilikos [ACM Trans. Algorithms 2022, TheoretiCS 2024]  for the particular case of {\sc Vertex Deletion to~$\Hcal$} within the same running time.
Our second algorithm is an improvement of the first one when $\Hcal$ is the class of graphs embeddable in a surface of Euler genus at most $g$ and runs in time $2^{\Ocal(k^{9})}\cdot |V(G)|^2$, where the $\Ocal(\cdot)$ notation depends on $g$.
To the best of our knowledge, these are the first parameterized algorithms with a reasonable parametric dependence for such a general family of graph modification problems to minor-closed classes.
\end{abstract}

\medskip\medskip\medskip



\section{Introduction}
\label{sec:intro}

A graph modification problem is typically determined by a target graph class $\Hcal$ and a prescribed set of allowed local modifications $\Mcal$, such as vertex/edge removal or edge addition/contraction or combinations of them, and the question is, given a graph $G$ and an integer $k$, whether it is possible to transform $G$ to a graph in $\Hcal$ by applying at most $k$ modification operations from $\Mcal$.
Graph modification problems are fundamental in algorithmic graph theory, as can be seen from
the span of applications in domains as diverse as computational biology, computer vision, machine learning, networking, or sociology; see~\cite{FominSM15grap} and the references therein.
Unfortunately, most of these problems are \NP-complete~\cite{LewisY80then,Yannakakis78}, and this
 justifies, among other approaches, to study them from the parameterized complexity viewpoint (see the monographs~\cite{CyganFKLMPPS15para,DowneyF13fund,FlumG06para,Niedermeier06invi} for an introduction to the field), where the number~$k$ of allowed modifications is taken as the parameter.

In the recent years, there has been a very active line of research about algorithmic meta-theorems  for graph modification problems where the target class $\Hcal$ is minor-closed\footnote{A \emph{minor-closed} graph class is a class of graphs that is closed under vertex deletion, edge deletion, and edge contraction.}.
By Robertson and Seymour's seminal result~\cite{RobertsonS04XX},
a minor-closed graph class $\Hcal$ has a finite number of minor-obstructions, that is, graphs that are not in $\Hcal$ but whose all proper minors are.
Combined with a minor containment algorithm~\cite{KorhonenPS24mino} (see also \cite{RobertsonS95XIII,KawarabayashiKR12thed}) running in almost-linear time, this implies that checking membership in a minor-closed graph class can be done in almost-linear time.
For some modification problems where the target class $\Hcal$ is minor-closed, such as vertex deletion, edge deletion, or vertex identification, the graphs $G$ such that $(G,k)$ is a \yes-instance of the problem for a fixed $k$ form a minor-closed graph class, which immediately implies an \FPT-algorithm in almost-linear time for these problems.
However, not all graph modification problems define a minor-closed graph class.
For instance, edge contraction to planar graphs does not. Indeed, consider the graph $K_{3,4}^+$ obtained from  $K_{3,4}$ by adding one edge $e$ on the side with three vertices.
Contracting $e$ gives the planar graph $K_{2,4}$, but $K_{3,4}$ cannot be made planar by contracting one edge.
Hence, some other algorithmic meta-theorems for graph modification problems to minor-closed graph classes were later introduced.
Some of them are ad-hoc meta-theorems such as the one of Fomin, Golovach, and Thilikos~\cite{FominGT19modi} that gives quadratic \FPT-algorithms for  graph modification problems where $\Hcal$ is the class of planar graphs and the modification is any combination of edge addition and edge deletion.
Much more generally, there has been a recent line of research on model-checking on minor-closed graph classes~\cite{FominGSST25comp,SauST25para}, which in particular implies quadratic \FPT-algorithms for an extremely wide family of graph modification problems where the target class $\Hcal$ is minor-closed.
Unfortunately, all these algorithmic meta-theorems have a major drawback: the parametric dependence on the ``amount of modification'' is humongous; in fact, even a rough upper bound is not known.

On the other hand, another line of research has focused on optimizing the parametric dependence for some particular graph modification problems when the parameter is the solution size.
When the target class $\Hcal$ is minor-closed, such study usually does not go much beyond $\Hcal$ being the class of forests and the class of union of paths~\cite{CrespelleDFG23asur,Tsur23fast,HeggernesHLLP14cont,MorelleST24vert,LiP20anim,ChenKX10impr,IwataK21impr}.
To  the best of our knowledge, only the case of  vertex deletion has been studied in a a series of papers~\cite{SauST23kapiI,SauST22kapiII,MorelleSST24fast,JansenLS14anea,MarxS12obta,Kawarabayashi09plan,KociumakaP19dele}, focusing on optimizing the running time (both the dependence on $k$ and $n$). 
In particular, when $\Hcal$ is the class of planar graphs, the currently fastest algorithm runs in time $2^{\Ocal(k\log k)}\cdot n$~\cite{JansenLS14anea}, when $\Hcal$ is the class of graphs embeddable in a surface of bounded genus, in time $2^{\Ocal(k^2\log k)}\cdot n^{\Ocal(1)}$~\cite{KociumakaP19dele}, and when $\Hcal$ is any minor-closed graph class, in time $2^{\poly(k)}\cdot n^2$~\cite{MorelleSST24fast}.

This article places itself in-between these two lines of research: we consider generic ``meta-modification'' operations (of course, much less generic than those of the currently most general algorithmic meta-theorem in~\cite{SauST25para}, but still quite versatile), and we manage to achieve the same (very reasonable) parametric dependence as the currently best one for vertex deletion~\cite{MorelleSST24fast} when the target~$\Hcal$ is any minor-closed graph class.
We hope that our work will trigger further research about {\sl efficient} algorithmic meta-theorems for graph modification problems to minor-closed graph classes.

\subparagraph{Our results.}
We define a graph modification problem, called {\sc $\Lcal$-Replacement to $\Hcal$} ({\sc $\Lcal$-R-$\Hcal$} for short), which, depending on the choice of the function $\Lcal$, called a \emph{replacement action}, can simulate vertex deletion, edge deletion, edge completion, edge edition, edge contraction, vertex identification, independent set deletion, matching deletion, matching contraction, star deletion, and subgraph complementation, to name a few (see \autoref{sec_ex} for an exposition of some problems encompassed by our result).
When $\Hcal$ is minor-closed, we solve {\sc $\Lcal$-R-$\Hcal$} in time $2^{\poly(k)}\cdot n^2$ (\autoref{th}), where the degree of the polynomial $\poly$ depends on the maximum size $s_\Hcal$ of the minor-obstructions of $\Hcal$.
This is the same running time as the one achieved by the currently best algorithm for vertex deletion~\cite{MorelleSST24fast} (the degree of $k$ in \poly is the same as in~\cite{MorelleSST24fast} up to an extra additive constant of one that is absolutely negligible compared to the total degree that depends (wildly) on $s_\Hcal$).
For the other graph modification problems encompassed within {\sc $\Lcal$-R-$\Hcal$}, to the authors' knowledge, the only minor-closed classes for which a good parametric dependence was previously known, if any, were the class of forests and the class of union of paths \cite{CrespelleDFG23asur,Tsur23fast,HeggernesHLLP14cont,MorelleST24vert,LiP20anim}.

As is usually the case concerning meta-theorems, the degree $d$ of the polynomial $\poly$ in \autoref{th} is unfortunately huge.
While we did not compute its exact value, we know that $d\ge 2^{2^{s_\Hcal^{24}}}$.
Nevertheless, $d$ can improved for some specific target classes $\Hcal$.
The \emph{Euler genus} of a surface $\Sigma$ that is obtained from the sphere by adding $h$ handles and $c$ crosscaps is defined to be $c+2h$.
In particular, when $\Hcal$ is the class of graphs embeddable in a surface of Euler genus at most $g$, we provide another algorithm solving {\sc $\Lcal$-R-$\Hcal$} in time $2^{\Ocal(k^{9})}\cdot n^2$ (\autoref{thpl}), where the $\Ocal(\cdot)$ notation depends on $g$. Note that, as opposed to \autoref{th}, in \autoref{thpl} the contribution of the genus (that is, of the target graph class $\Hcal$) does {\sl not} affect the degree of the parameter $k$ in the exponent.

\subparagraph{Our techniques.}
To handle several modification problems at once, we adapt  the vocabulary of Fomin, Golovach, and Thilikos~\cite{FominGT19modi}, who introduced the notion of replacement action.
A \emph{replacement action} is a function $\Lcal$ that maps a graph $H_1$ to a collection $\Lcal(H_1)$ of pairs $(H_2,\phi)$ where $H_2$ is a graph with at most $|V(H_1)|$ vertices and $\phi$ maps each vertex of $H_1$ to either a vertex of $H_2$ or the empty set.
Mapping a vertex of $H_1$ to the empty set corresponds to a deletion, while mapping several vertices to the same vertex of $H_2$ corresponds to an identification.
Replacement actions were originally defined in~\cite{FominGT19modi} to solve a collection of graph modification problems where only edges are modified and where the target class is the class of planar graphs.
Compared to~\cite{FominGT19modi}, however, the size of $H_2$ may here be smaller than the size of $H_1$, which happens when deleting or identifying vertices, while in~\cite{FominGT19modi} it is required that $|V(H_1)|=|V(H_2)|$.
Let us fix a replacement action $\Lcal$ and a target graph class $\Hcal$.
The {\sc $\Lcal$-Replacement to $\Hcal$} ({\sc $\Lcal$-R-$\Hcal$}) problem asks, given a graph $G$ and $k\in\bN$, whether there is an induced subgraph $H_1$ of size at most $k$ in $G$ and a pair $(H_2,\phi)\in\Lcal(H_1)$ such that $H_1$ can be replaced by $H_2$ such that the resulting graph $G'$ belongs to $\Hcal$ (for $u\in V(G)\setminus V(H_1)$ and $v\in V(H_2)$, $uv\in E(G')$ if and only if there is $v'\in\phi^{-1}(v)$ such that $uv'\in E(G)$).
The formal definition of the problem is given in \autoref{subsec_defpb}.
For our techniques to work (see the ``irrelevant vertex technique'' paragraph below for more precision), we require our function $\Lcal$ to be \emph{hereditary}, which essentially means if $H_2$ is in $\Lcal(H_1)$, then for any induced subgraph $H_1'$ of $H_1$, the corresponding induced subgraph of $H_2$ is in $\Lcal(H_1')$ (cf.~\autoref{subsec_defpb} for the formal definition and \autoref{fig_hered} for an illustration).
For instance, this implies that we can ask whether it is possible to do {\sl at most} $k$ edge editions to get a graph in $\Hcal$, but we cannot ask whether it is possible to do {\sl exactly} $k$ edge editions to get a graph in $\Hcal$.

The techniques that we employ for our first algorithm (that is, when $\Hcal$ is any minor-closed graph class) are
strongly inspired by those used by Morelle, Sau, Stamoulis,
Thilikos~\cite{MorelleSST24fast} for the particular case of vertex deletion (see also~\cite{SauST22kapiII}), namely {\sc Vertex Deletion to $\Hcal$}, achieving the same running time. Nevertheless, in order to deal with our ``meta-modification'' operations, we need several new technical insights compared to the approach of~\cite{SauST22kapiII}, which we sketch below.

There are many algorithmic results in the literature explaining why and when a certain technique is successful, just as Courcelle's theorem~\cite{Courcelle90them} explains why many problems can be solved by dynamic programming on graphs of bounded treewidth.
This is also the case of this paper.
Namely, while the proof of our result for the general case follows the path of~\cite{MorelleSST24fast}, we aim to characterize all graph modification problems to which these techniques can be applied, and to solve all of them at once.

In a nutshell, the algorithm of~\cite{MorelleSST24fast} employs a  win/win strategy that proceeds as follows:
\begin{itemize}
\item If the treewidth of the input graph is small (as a function of the parameter $k$), then solve the problem via a dynamic programming approach.
\item If the treewidth of the input graph is big, then either
\begin{itemize}
\item (\emph{irrelevant vertex}) find a vertex $v$ such that $(G,k)$ and $(G-v,k)$ are equivalent instances, or
\item (\emph{branching case}) find a set $A\subseteq V(G)$ of small size such that there exists $v\in A$ such that $(G,k)$ and $(G-v,k-1)$ are equivalent instances,
\end{itemize}
and recurse.
\end{itemize}
Hence, we require three ingredients: the first is to solve the problem parameterized by treewidth, the second is to find an irrelevant vertex, and the thrird is to find an ``obligatory set'' $A$, all with a ``reasonable'' parametric dependence on $k$.
Then, we need to construct an algorithm so that one of these three cases always applies and such that the overall running time is still within the desired bound, which is one of the most convoluted parts of the proof.

Let $S'$ be the set of vertices recursively guessed to be modified in the branching step.
An advantage when the modification consists in vertex deletion is that we can simply recurse on $(G-S',k-|S'|)$.
For the more general case of {\sc $\Lcal$-R-$\Hcal$}, we cannot simply delete $S'$, as the considered modification may be different from vertex deletion.
We need 1) to guess how $G[S']$ is modified, that is, to guess $(H_2',\phi')\in\Lcal(G[S'])$ and 2) to remember $S'$ and $(H_2',\phi')$ in order to check that we eventually find a set $S\supseteq S'$ and an allowed modification $(H_2,\phi)\in\Lcal(G[S])$ whose restriction to $S'$ is $(H_2',\phi')$ such that the modified graph is in $\Hcal$.
This is why we need to solve the {\sl annotated version of the problem}, denoted by {\sc $\Lcal$-AR-$\Hcal$}, where we add to the input a subset $S'$ of vertices of $G$ that are required to be part of~$H_1$, as well as the modification $(H_2',\phi')$ made on $S'$.

As for solving the problem when the graph has bounded treewidth, we cannot just use Courcelle's theorem~\cite{Courcelle90them}, since we require a nice parametric dependence on $k$.
Hence, we need to design our own dynamic programming algorithm to solve {\sc $\Lcal$-AR-$\Hcal$} parameterized by the treewidth and $k$ (\autoref{th_tw}, proved in \autoref{sec_tw}).
Essentially, the idea is to guess, in each bag $\beta(t)$ of the decomposition, the set $S_t$ of vertices that are modified as well as how they are modified,
and to reduce the size of the graph $G_t$ induced by the bag $t$ and its children using the representative-based technique of~\cite{BasteST23hittIV}.
This technique is essentially based on the property that (cf. \autoref{@iinelstaai} and \autoref{@encounters}), given a graph $G$ in a minor-closed graph class $\Hcal$ with a boundary $B$, there is a graph $R$ of {\sl bounded size} with same boundary $B$, called the \emph{representative of $G$}, such that, for any graph $H$ glued on $B$ to get $G\oplus H$ and $R\oplus H$, $G\oplus H\in\Hcal$ if and only if $R\oplus H\in\Hcal$.
$G_t$ does not belong to $\Hcal$, so we cannot find a representative of $G_t$, but we find instead a representative of the graph $G_t'\in\Hcal$ modified from $G_t$ according to the guessed modification on $S_t$ and the previously guessed modification on the children of $t$.
Given that we may need to identify together vertices that are far apart in the tree decomposition, we need to remember throughout the algorithm the vertices that are guessed to be part of the solution.
The fact that we keep information about these at most $k$ vertices explains the dependence on $k$ of the dynamic programming algorithm, which runs in time $2^{\Ocal(k^2+(k+\tw)\log(k+\tw))}\cdot n$.
This result parameterized by treewidth and $k$ may be of independent interest, given that it implies an algorithm with a good parametric dependence on the treewidth and $k$ for a number of graph modification problems.
Note that the question of whether {\sc $\Lcal$-R-$\Hcal$} is \FPT parameterized by only treewidth  is open.
Even Courcelle's theorem only implies a running time of $f(\tw,k)\cdot n$, given that the size of the CMSO formula expressing \yes-instances of {\sc $\Lcal$-R-$\Hcal$} depends on $k$. Note that, in~\cite{MorelleSST24fast}, the bounded treewidth part just consisted in a black-box application of the algorithm of Baste, Sau, and Thilikos~\cite{BasteST23hittIV}.

As expected, finding an irrelevant vertex (\autoref{th_irr}, proved in \autoref{sec_irr}) is done using the irrelevant vertex technique of Robertson and Seymour~\cite{RobertsonS95XIII}.
More specifically, we generalize the irrelevant vertex technique used in~\cite{MorelleSST24fast} (that is actually proved in~\cite{SauST23kapiI}).
The irrelevant vertex technique is based on the (intuitive but surprisingly hard to prove) fact that the central vertex of a specific structure called a ``flat wall'' (cf. \autoref{fig_flat}), that we require moreover to be ``homogeneous'', is always irrelevant.
While this technique is very powerful, it requires heavy machinery to be defined formally (\autoref{sec_flat}).
Also, while our irrelevant vertex technique for {\sc $\Lcal$-AR-$\Hcal$} takes inspiration from~\cite{SauST23kapiI}, it is far more involved due to the annotation and the fact that we allow a wide variety of modifications.
In particular, we need to redefine what it means to be homogeneous for a flat wall, so that it works in our new setting.
The previous definition was made to handle the case when we had to remove a small vertex set, called \emph{apex set}, to find a flat wall, and more specifically to handle the fact that some vertices are possibly deleted from the apex set.
Now, we also need to handle any other way the apex set may be modified, hence the new definition.
The fact that we ask the replacement action $\Lcal$ to be hereditary comes from the irrelevant vertex technique.
Indeed, in order to prove that the central vertex $v$ of a homogeneous flat wall $W$ is irrelevant, we essentially prove that, for any solution $(S,H_2,\phi)$, we can delete a small part $X$ of $W$ containing $v$, and that the restriction of $(S,H_2,\phi)$ to $G-X$ is still a solution.

The branching case (\autoref{lem_obl}, proved in \autoref{sec_obl}) is not much different from what is done in~\cite{MorelleSST24fast} (originally from~\cite{SauST23kapiI}):
essentially, if there is a big enough wall $W$ (cf. \autoref{@manuscript}) and a set $A$ of vertices having many disjoint paths to $W$ (cf. \autoref{fig_obl_vtx}), then some modification $(H_A,\phi_A)$ must happen in $A$ and we can branch.
Here, we however need to additionally prove that we must have $|\phi_A(A)\setminus\{\emptyset\}|<|A|$.
We stress that it is important here to guess some modification in $A$ that strictly decreases the size of $A$, so that,
after applying this partial modification to $G$
at the next step in the recursion, we will not find the exact same obligatory set $A$.
Hence, in the algorithm with input $(G,S',H_2',\phi',k)$, at each step, either we find an irrelevant vertex and strictly decrease the size of $G$, or we branch and strictly increase the size of $S'$.

Finally, in \autoref{sec_algo} we combine these three ingredients to find an algorithm for {\sc $\Lcal$-AR-$\Hcal$}.
It essentially proceeds as follows. Let $(G,S',H_2',\phi',k)$ be the instance we want to solve, and $G'$ be obtained by doing the modification $(H_2',\phi')$ of $S'$.
In the first steps, we either find that $G$ has small treewidth, where we can use our dynamic programming algorithm to conclude, or that $G'$ contains a wall $W$.
Given $W$, we first try to find a flat wall $W'$ inside, with all the necessary conditions to find an irrelevant vertex. If we manage to do so, we remove the irrelevant vertex and recurse.
Otherwise, through a greedy procedure, we try to find an obligatory vertex set $A$ with many disjoint paths to $W$ in $G'$.
If we find such a set, we branch and recurse.
If not, we manage to argue that we must have a \no-instance, and conclude.

The second algorithm, when $\Hcal$ is a class of graphs embeddable in a surface of bounded Euler genus, uses two additional ideas to get an improved running time.
The first one is that here, the obligatory set $A$ is a singleton. Indeed, the size of $A$ is the size of the minimum number of vertices one can remove from an obstruction of $\Hcal$ to make it planar.
It is well known that, when $\Hcal$ is such a class,  there is some integer $t$ depending on the Euler genus such that $K_{3,t}\notin\Hcal$, and thus, $|A|=1$.
In particular, this  implies that we do not need to branch on $A$, but that we instead immediately find an obligatory vertex.
The second idea is about homogeneous flat walls.
In the running time $2^{\poly(k)}\cdot n^2$ of the first algorithm, the degree of \poly essentially corresponds to the size of the required flat wall to find a big enough homogeneous flat wall, and hence an irrelevant vertex, inside of it.
In the case where $\Hcal$ is the class of graphs embeddable in a surface of Euler genus at most $g$, we prove that we can find a homogeneous flat wall inside a flat wall of smaller size, hence the improved running time (\autoref{th_irr_pl}, proved in \autoref{subsec_pl}). To do so, we prove that, after some preliminary processing, a flat wall that is furthermore embeddable in a disk with the perimeter on its boundary is already homogeneous (\autoref{@disreputablepl}).
Hence, our second algorithm (\autoref{subsec_planar}) proceeds similarly to the first one, but if we find a flat wall $W'$ in $G'$, we divide $W'$ into $k+1$ disjoint smaller flat walls and check whether they belong to $\Hcal$.
By the pigeonhole principle, one of them, $W_i$, does not contain a modified vertex and must thus be in $\Hcal$, otherwise we return a \no-instance. Then, we argue, using a result from~\cite{DemaineHT06theb} (\autoref{prop:nice-embedding}) to guarantee additional properties of the planar embedding that are needed for technical reasons, that we can find a smaller flat wall $W_i'$ in $W_i$ with a {\sl planar embedding} (even if the genus of the target graph class is strictly positive). Hence, we find an irrelevant vertex in $W_i'$ and conclude.

\subparagraph{Organization.} In \autoref{sec_problem} we give basic definitions and conventions, and we formally define the problem.
We also give a non-exhaustive list of graph modification problems that can be simulated by {\sc $\Lcal$-Replacement to $\Hcal$}.
In \autoref{sec_flat} we introduce walls, flat walls, and related notions, as well as some preliminary results concerning flat walls.
In \autoref{sec_algo} we present our two algorithms.
The proof of the three main ingredients is deferred to the later sections: we give our irrelevant vertex result in \autoref{sec_irr}, \autoref{sec_obl} is dedicated to the branching argument, and we finish in \autoref{sec_tw} with the dynamic programming algorithm. We present in \autoref{sec_conclusions}
 some directions for further research.

\vspace{-1mm}\section{Definition of the problem, results, and applications}
\label{sec_problem}

In this section, we formally define the {\sc $\Lcal$-R-$\Hcal$} problem and its annotated version in \autoref{subsec_defpb}, after having given basic definitions on graphs and some conventions in \autoref{subsec_basic}.
Then, we give in \autoref{sec_ex} a non-exhaustive list of graph modification problems that correspond to different instantiations of $\Lcal$.

\subsection{Basic definitions}\label{subsec_basic}

\subparagraph{Sets and integers.}
We denote by $\mathbb{N}$ the set of non-negative integers.
Given two integers $p, q,$ where $p \leq q,$ we denote by $[p, q]$ the set $\{p, \dots, q\}.$ For an integer $p \geq 1,$ we set $[p] = [1, p]$ and $\mathbb{N}_{\geq p} = \mathbb{N} \setminus [0, p - 1].$
For a set $S,$ we denote by $2^{S}$ the set of all subsets of $S$ and by $S \choose 2$ the set of all subsets of $S$ of size $2.$
If $\mathcal{S}$ is a collection of objects where the operation $\cup$ is defined, then we denote $\cupall \Scal = \bigcup_{X \in \mathcal{S}} X.$
Given $x\in\bN$, we denote by $\odd(x)$ the smallest odd $p\in\bN$ such that $p\ge x$.
Given a set $A$, we denote the identity function mapping each $a\in A$ to itself by ${\sf id}_A$.
Given two sets $A,B$, $v\in B$, $S\subseteq B$, and a function $f:A\to B$, $f^{-1}(v)$ is the set of elements $u\in A$ such that $f(u)=v$, and $f^{-1}(S)=\bigcup_{v\in S}f^{-1}(v)$.
Given $R\subseteq A$, $f(R)=\bigcup_{r\in R}\{f(r)\}$, and the restriction of $f$ to $R$ is denoted by $f|_{R}$.
Additionally, for some new vertex $u\notin A$, $f\cup(u\mapsto v):A\cup\{u\}\to B$ is the function that maps $u$ to $v$ and whose restriction to $A$ is $f$.

\subparagraph{Basic concepts on graphs.}
A graph $G$ is any pair $(V, E)$ where $V$ is a finite set and $E \subseteq {V \choose 2},$ i.e., all graphs in this paper are undirected, finite, and without loops or multiple edges.
We refer the reader to~\cite{Diestel17grap}  for any undefined terminology on graphs.
When we denote
an edge $\{x,y\}$, we use instead the simpler notation $xy$ (or $yx$).
We also define $V(G) = V$ and $E(G) = E.$
The \emph{detail} of $G$ is $\max\{|V(G)|,|E(G)|\}$.
For $S \subseteq V(G),$ we set $G[S] = (S, E \cap {S \choose 2})$ and use $G - S$ to denote $G[V(G) \setminus S].$ We say that $G[S]$ is an \emph{induced (by $S$) subgraph} of $G$.
A graph class $\Hcal$ is \emph{hereditary} if, for each graph $G$ and each induced subgraph $H$ of $G$, the fact that $G\in\Hcal$ implies that $H\in\Hcal$.
Given a vertex $v \in V(G),$ we denote by $N_{G}(v)$ the set of vertices of $G$ that are adjacent to $v$ in $G$ and we set $N_G[v]=N_G(v)\cup\{v\}$.
Given a set $S \subseteq V(G),$ we set $N_G[S]=\bigcup_{v\in S}N_G[v]$ and $N_{G}(S) = N_G[S] \setminus S$.
Given an edge $e = uv \in E(G),$ we define the \emph{subdivision} of $e$ to be the operation of deleting $e,$ adding a new vertex $w,$ and making it adjacent to $u$ and $v.$
Given two graphs $H$ and $G,$ we say that $H$ is a \emph{subdivision} of $G$ if $H$ can be obtained from $G$ by subdividing edges.
A \emph{separation} of a graph $G$ is a pair $(L,R)\subseteq V(G)^2$ such that $L\cup R=V(G)$ and there is no edge in $G$ between $L\setminus R$ and $R\setminus L$. The \emph{order} of $(L,R)$ is $|L\cap R|$.
A \emph{cut vertex} in a graph $G$ is a vertex whose removal increases the number of connected components of $G$.
The \emph{apex number} of a graph $G$ is the smallest integer $a$ for which there is a set $A\subseteq V(G)$ of size $a$ such that $G-A$ is planar.

\subparagraph{Tree decompositions.}
 A \emph{tree decomposition} of a graph $G$ is a pair $\mathcal{T}=(T,\beta)$ where $T$ is a tree and $\beta\colon V(T)\rightarrow 2^{V(G)}$ is a function, whose images are called the \emph{bags} of $\mathcal{T},$ such that
\begin{itemize}
\item $\bigcup_{t\in V(T)} \beta (t)= V (G),$
\item for every $e\in E(G)$, there exists $t\in V(T)$ with $e\subseteq \beta(t),$ and
\item for every $v\in V(G)$, the set $\{t\in V(T) \mid v\in \beta(t)\}$ induces a subtree of $T.$
\end{itemize}
The \emph{width} of $\mathcal{T}$ is the maximum size of a bag minus one
and the \emph{treewidth} of $G$, denoted by $\tw(G),$ is the minimum width of a tree decomposition of $G$.

\subparagraph{Minors.}
The \emph{contraction} of an edge $e = uv \in E(G)$ results in a graph $G'$ obtained from $G - \{ u, v \}$ by adding a new vertex $w$ adjacent to all vertices in the set $N_G(\{u,v\}).$
A graph $H$ is a \emph{minor} of a graph $G$ if $H$ can be obtained from a subgraph of $G$ after a series of edge contractions.
It is easy to verify that $H$ is a minor of $G$ if and only if there is a collection $\Scal=\{S_v\mid v\in V(H)\}$ of pairwise-disjoint connected subsets of $V(G)$ such that, for each edge $xy\in E(H)$, the set $S_x\cup S_y$ is connected in $V(G)$. $\Scal$~is called a \emph{model} of $H$ in $G$.

\subparagraph{Minor-closed graph classes.}
A graph class $\Hcal$ is \emph{minor-closed} if, for each graph $G$ and each minor $H$ of $G$, the fact that $G\in\Hcal$ implies that $H\in\Hcal$.
Given a collection of graphs $\Fcal$, we denote by \emph{$\exc(\Fcal)$} the class of graphs that do not contain a graph in $\Fcal$ as a minor.
Obviously, $\exc(\Fcal)$ is minor-closed.
A \emph{(minor-)obstruction} of a graph class $\Hcal$ is a graph $F$ that is not in $\Hcal$, but whose minors are all in $\Hcal$.
The set of all the obstructions of $\Hcal$ is denoted by \emph{$\obs(\Hcal)$}.
By the seminal work of Robertson and Seymour~\cite{RobertsonS04XX}, if $\Hcal$ is a minor-closed graph class, then $\obs(\Hcal)$ is finite.
Note that, if $\Fcal=\obs(\Hcal)$, then $\exc(\Fcal)=\Hcal$.

\subsection{Definition of the problem and main results}
\label{subsec_defpb}

\subparagraph{Ordered graphs.}
For the definitions of the next two paragraphs to be correct,
we actually need to
consider ordered graphs instead of graphs (see the ``Graph modifications'' paragraph).
An \emph{ordered graph} is  a graph $G$ equipped with a strict total order on $V(G)$, denoted by $<_G$.
In other words, there exists an indexation $v_1,\dots,v_n$ of the vertices of $V(G)$ such that  $v_1<_Gv_2<_G\dots<_Gv_n$.
A subgraph $H$ of an ordered graph $G$ naturally comes equipped with the strict order $<_H$ such that, for each distinct $u,v\in V(H)$,
$u<_Hv$ if and only if $u<_Gv$.

\subparagraph{Replacement actions.}
The \emph{any-replacement action} is the function $\Mcal$ that maps
each ordered graph $H_1$ to the collection $\Mcal(H_1)$ of all the pairs $(H_2,\phi)$, where $H_2$ is an ordered graph and $\phi:V(H_1)\to V(H_2)\cup\{\emptyset\}$ is a function such that:
\begin{itemize}
\item $|V(H_2)|\le|V(H_1)|$,
\item for each $v\in V(H_2)$, $\phi^{-1}(v)\ne\emptyset$, and
\item $<_{H_2}$ is the strict total order such that, for each distinct $v_1,v_2\in V(H_2)$, we have $v_1<_{H_2} v_2$ if and only if
$u_1<_{H_1} u_2$ where, for $i\in[2]$, $u_i$ is the smallest vertex (according to $<_{G}$) in $\phi^{-1}(v_i)$.
\end{itemize}
A \emph{replacement action} (abbreviated as \emph{R-action}) is any function $\Lcal$
that maps an ordered graph (called a \emph{pattern}) $H_1$ to a non-empty collection $\Lcal(H_1)\subseteq \Mcal(H_1)$ of its possible \emph{pattern transformations}.
See \autoref{fig_raction} for an illustration.
The vertices of $H_1$ mapped by $\phi$ to the empty set are said to be \emph{deleted}, and two vertices of $H_1$ mapped by $\phi$ to the same vertex of $H_2$ are
said to be \emph{identified}.
Given $S\subseteq V(H_1)$, we set $\phi^{+}(S)=\phi(S)\setminus \{\emptyset\}$.
Note that, if $\phi(S)=\{\emptyset\}$, then $\phi^+(S)=\{\emptyset\}\setminus\{\emptyset\}=\emptyset$.

\subparagraph{Graph modifications.}
Let $\Lcal$ be an R-action, let $G$ be an ordered graph, and $S\subseteq V(G)$.
Let $(H_2,\phi)\in\Lcal(G[S])$.
We denote by \emph{$G_{(H_2,\phi)}^S$} the graph obtained from the disjoint union of $G-S$ and $H_2$
by adding an edge $u\phi(v)$ for each $u\in V(G)\setminus S$ and each $v\in \phi^{-1}(V(H_2))$
such that $uv\in E(G)$.
We equip $G':=G_{(H_2,\phi)}^S$ with the strict total order $<_{G'}$ such that $v_1<_{G'} v_2$ if and only if $u_1<_Gu_2$ where,
for $i\in[2]$, $u_i:=v_i$ if $v_i\in V(G)\setminus S$, and $u_i$ is the smallest vertex in $\phi^{-1}(v_i)$ if $v_i\in V(H_2)$.
We also set $\Lcal_S(G)=\{G_{(H_2,\phi)}^S\mid (H_2,\phi)\in\Lcal(G[S])\}$. 
See \autoref{fig_raction} for an illustration.

Note that we consider ordered graphs merely so that the correspondence between the vertices in $S$ and the vertices in $V(H_2)$ is well-defined.
We actually omit the order from the statements, but it will be implicitly assumed that vertices have a label that allows us to keep track of them during the modification procedure.

\begin{figure}[h]
\center
\includegraphics{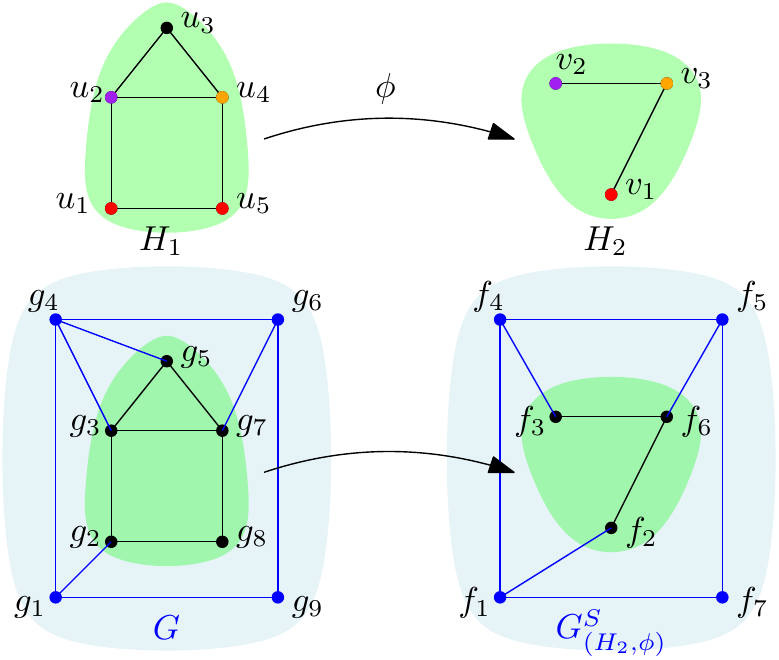}
\caption{Example of $(H_2,\phi)\in\Lcal(H_1)$ and of the graph modification $G_{(H,\phi)}^S$ where $S$ is the set of black vertices of $G$.
$\phi$ is represented by the colors, that is, $\phi(u_1)=\phi(u_5)=v_1$, $\phi(u_2)=\phi(v_2)$, $\phi(u_3)=\emptyset$, and $\phi(u_4)=v_3$. The order on the vertex sets of the depicted graphs is given by the corresponding labels.}
\label{fig_raction}
\end{figure}

\medskip
Let $\Lcal$ be an R-action and $\Hcal$ be a graph class.
We define the problem {\sc $\Lcal$-Replacement to $\Hcal$} as follows.

\defproblem{$\Lcal$-Replacement to $\Hcal$ ($\Lcal$-R-$\Hcal$)}
{A graph $G$ and $k\in\bN$.}
{Is there a set $S\subseteq V(G)$ of size at most $k$ such that $\Lcal_{S}(G)\cap\Hcal\ne\emptyset$?}
Such a set $S$ is called \emph{solution} of {\sc $\Lcal$-R-$\Hcal$} for the instance $(G,k)$.

\medskip
We will use the following observation, which implies that a \no-instance for {\sc Vertex Deletion to $\Hcal$} is also a \no-instance for {\sc $\Lcal$-R-$\Hcal$}.

\begin{observation}\label{obs_deltomodif}
Let $\Hcal$ be a hereditary graph class,
let $\Lcal$ be an R-action,
let $G$ be a graph, and let $S\subseteq V(G)$.
If $\Lcal_S(G)\cap\Hcal\ne\emptyset$, then $G-S\in\Hcal$.
\end{observation}

\begin{proof}
Indeed, suppose that there is $(H_2,\phi)\in\Lcal(G[S])$ such that $G_{(H_2,\phi)}^S\in\Hcal$.
Then, because~$\Hcal$ is hereditary, $G_{(H_2,\phi)}^S-\phi^+(S)=G-S\in\Hcal$.
\end{proof}

\begin{figure}[h]
\center
\includegraphics{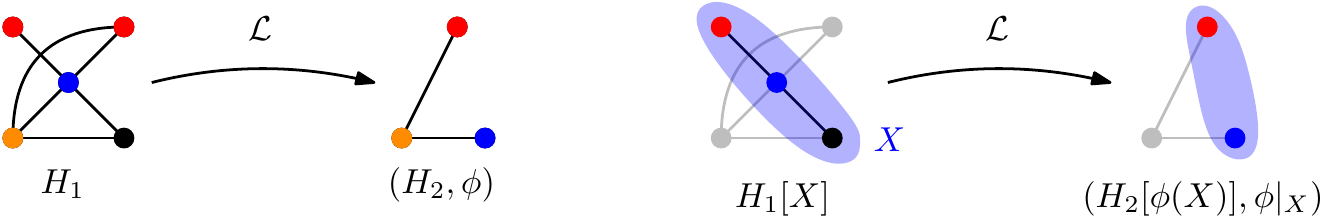}
\caption{If $\Lcal$ is hereditary, then a restriction of an allowed modification is also allowed.}
\label{fig_hered}
\end{figure}

\subparagraph{Hereditary R-actions.}
An R-action is said to be \emph{hereditary}
if, for each ordered graph $H_1$, for each non-empty $X\subseteq V(H_1)$, and for each $(H_2,\phi)\in\Lcal(H_1)$,
we have $(H_2[\phi^+(X)],\phi|_X)\in\Lcal(H_1[X])$.
We say that $(H_2[\phi^+(X)],\phi|_X)$ is the \emph{restriction} of $(H_2,\phi)$ to $X$.
See \autoref{fig_hered} for an illustration.

Informally, an R-action is hereditary if, when a modification is allowed, then modifying ``less'' is allowed as well.
For instance, if $\Lcal$ allows us to delete {\sl exactly} $k$ vertices, then $\Lcal$ also allows us to delete {\sl at most} $ k$ vertices.

\subparagraph{Some conventions.}
By convention, when there is no confusion, we set $n:=|V(G)|$ and $m:=|E(G)|$.
In the rest of the paper, instead of considering a minor-closed graph class $\Hcal$, we consider its obstruction set $\Fcal$, and thus the minor-closed graph class $\exc(\Fcal)$.
We define three constants depending on $\Fcal$ that will be used throughout the paper whenever we consider such a collection $\Fcal$.
We define $a_\Fcal$ as the minimum apex number of a graph in $\Fcal$,
we set $s_\Fcal:=\max\{|V(F)|\mid F\in\Fcal\}$, and we define $\ell_\Fcal$ to be the maximum detail of a graph in $\Fcal$.
Given a tuple ${\bf t}=(x_1,\dots,x_\ell)\in\bN^\ell$ and two functions $\chi,\psi:\bN\to\bN$,
we write $\chi(n)=\Ocal_{\bf t}(\psi(n))$ in order to denote that there exists a computable function $\phi:\bN^\ell\to\bN$ such that $\chi(n)=\Ocal(\phi({\bf t})\cdot\psi(n))$.
Notice that $s_\Fcal\le\ell_\Fcal\le s_\Fcal(s_\Fcal-1)/2$, and thus $\Ocal_{\ell_\Fcal}(\cdot)=\Ocal_{s_\Fcal}(\cdot)$.
Observe also that the \yes-instances of {\sc $\Lcal$-R-$\exc(\Fcal)$} exclude $K_{s_\Fcal+k}$ as a minor by \autoref{obs_deltomodif}.
Thus, due to~\cite{Thomason01thee}, we can always assume that the input graph $G$ has $\Ocal_{s_\Fcal}(k\sqrt{\log k}\cdot n)$ edges,
since otherwise we can directly conclude that $(G,k)$ is a \no-instance for {\sc $\Lcal$-R-$\exc(\Fcal)$}.

\medskip
Our main result is the following.

\begin{theorem}\label{th}
Let $\Fcal$ be a finite collection of graphs and let $\Lcal$ be a hereditary R-action.
There is an algorithm that, given a graph $G$ and $k\in\bN$, runs in time $2^{\poly_\Fcal(k)}\cdot n^2$ and either outputs a solution of {\sc $\Lcal$-R-$\exc(\Fcal)$} for the instance $(G,k)$ or reports a \no-instance.
Moreover, $\poly_\Fcal$ is a polynomial whose degree depends on the maximum detail of a graph in $\Fcal$.
\end{theorem}

As mentioned in the introduction, the main result in~\cite{SauST25para} already implies that {\sc $\Lcal$-R-$\Hcal$} is solvable in time $f(k)\cdot n^2$ when $\Hcal$ is minor-closed for some huge function $f$ that is not even estimated.
Our main contribution is an explicit and single-exponential dependence on $k$. 

The degree of $\poly_\Fcal(k)$ is quite big, but we can reduce it in some specific cases.

\begin{theorem}\label{thpl}
Let $\Lcal$ be a hereditary R-action and $\Hcal$ be the class of graphs embeddable in a surface $\Sigma$ of Euler genus at most $g$.
There is an algorithm that, given a graph $G$ and $k\in\bN$, runs in time $2^{\Ocal_g(k^{9})}\cdot n^2$ and either outputs a solution of {\sc $\Lcal$-R-$\Hcal$} for the instance $(G,k)$ or reports a \no-instance.
\end{theorem}

\medskip
More generally, we study the annotated version of $\Lcal$-R-$\Hcal$.
Let $\Lcal$ be a hereditary R-action and $\Hcal$ be a graph class.
We define the problem {\sc $\Lcal$-Annotated Replacement to $\Hcal$} as follows.

\defproblem{$\Lcal$-Annotated Replacement to $\Hcal$ ($\Lcal$-AR-$\Hcal$)}
{A graph $G$, a set of annotated vertices $S'\subseteq V(G)$, $(H_2',\phi')\in\Lcal(G[S'])$, and $k\in\bN$.}
{Is there a set $S\subseteq V(G)$ of size at most $k$ and $(H_2,\phi)\in\Lcal(G[S])$ such that $(H_2',\phi')$ is the restriction of $(H_2,\phi)$ to $S'$ and $G_{(H_2,\phi)}^S\in\Hcal$?}
Obviously, we must have $S'\subseteq S$.
Such a triple $(S,H_2,\phi)$ is called a \emph{solution} of {\sc $\Lcal$-AR-$\Hcal$} for the instance $(G,S',H_2',\phi',k)$.
An instance of {\sc $\Lcal$-AR-$\Hcal$} where $S'=\emptyset$ is an instance of {\sc $\Lcal$-R-$\Hcal$},
so {\sc $\Lcal$-AR-$\Hcal$} generalizes {\sc $\Lcal$-R-$\Hcal$}.
Two instances $\Ical_1$ and $\Ical_2$ are said to be \emph{equivalent} instances of {\sc $\Lcal$-AR-$\Hcal$}
if $\Ical_1$ is a \yes-instance of {\sc $\Lcal$-AR-$\Hcal$} if and only if $\Ical_2$ is a \yes-instance of {\sc $\Lcal$-AR-$\Hcal$}.

\medskip
In fact, the results that we actually prove are the following.

\begin{restatable}{theorem}{than}
\label{th_an}
Let $\Fcal$ be a finite collection of graphs and let $\Lcal$ be a hereditary R-action.
There is an algorithm that, given a graph $G$, $S'\subseteq V(G)$, $(H_2',\phi')\in\Lcal(G[S'])$, and $k\in\bN$, runs in time $2^{\poly_\Fcal(k)}\cdot n^2$ and either outputs a solution of {\sc $\Lcal$-AR-$\exc(\Fcal)$} for the instance $(G,S',H_2',\phi',k)$ or reports a \no-instance.
Moreover, $\poly_\Fcal$ is a polynomial whose degree depends on the maximum detail of a graph in $\Fcal$.
\end{restatable}

\begin{restatable}{theorem}{thplan}
\label{thpl_an}
Let $\Lcal$ be a hereditary R-action and $\Hcal$ be the class of graphs embeddable in a surface $\Sigma$ of Euler genus at most $g$.
There is an algorithm that, given a graph $G$, $S'\subseteq V(G)$, $(H_2',\phi')\in\Lcal(G[S'])$, and $k\in\bN$, runs in time $2^{\Ocal_g(k^{9})}\cdot n^2$ and either outputs a solution of {\sc $\Lcal$-AR-$\Hcal$} for the instance $(G,S',H_2',\phi',k)$ or reports a \no-instance.
\end{restatable}

\subsection{Problems generated by different instantiations of $\Lcal$}\label{sec_ex}

Many graph modification problems correspond to {\sc $\Lcal$-R-$\Hcal$} for a specific R-action $\Lcal$ and a specific target graph class $\Hcal$.
We give a few examples below.
Let $\Hcal$ be a minor-closed graph class.
For instance, $\Hcal$ could be the class of edgeless graphs, of forests, of graphs whose connected components have size at most $k$, of planar graphs, or of graphs embeddable in a surface $\Sigma$.
Note that we do not mention {\sc Edge Addition to $\Hcal$} (nor {\sc Edge Edition to $\Hcal$}) here, because when $\Hcal$ is a minor-closed graph class, 
adding edges is ``unnecessary'', in the sense that the edge deletion variant has the same expressive power, and we can solve it.
Note also that {\sc $\Lcal$-R-$\Hcal$}, and thus
in particular all problems of this section, was already known to be solvable in \FPT-time (when $\Hcal$ is minor-closed) by the result of~\cite{SauST25para}.
However, as mentioned before, the parametric dependence is huge and not even explicit in~\cite{SauST25para}.

\smallskip
\defproblem{Vertex Deletion to $\Hcal$}
{A graph $G$ and $k\in\bN$.}
{Is there a set $S\subseteq V(G)$ of size at most $k$ such that $G-S\in\Hcal$?}
{\sc Vertex Deletion to $\Hcal$} reduces to {\sc $\Lcal_{\sf vDel}$-R-$\Hcal$}, where
$\Lcal_{\sf vDel}$ is the function that maps any graph $H_1$ to the singleton containing the empty graph and the constant function $\phi:V(H_1)\to\{\emptyset\}$.

{\sc Vertex Deletion to $\Hcal$} is already known~\cite{MorelleSST24fast} to be solvable within the same running time as the one of \autoref{th}.
Hence, the result of \autoref{th} is not an improvement for this specific problem, but it shows that our result is tight compared to the currently best known result for {\sc Vertex Deletion to $\Hcal$}.

\smallskip
\defproblem{Edge Deletion to $\Hcal$}
{A graph $G$ and $k\in\bN$.}
{Is there a set $F\subseteq E(G)$ of size at most $k$ such that $G-F\in\Hcal$?}
$(G,k)$ is a \yes-instance of {\sc Edge Deletion to $\Hcal$} if and only if $(G,2k)$ is a \yes-instance of {\sc $\Lcal_{{\sf eDel},k}$-R-$\Hcal$}, where
 $\Lcal_{{\sf eDel},k}$ is the function that maps each graph $H_1$ to the set of pairs $(H_1-F,{\sf id}_{V(H_1)})$ over all $F\subseteq E(G)$ of size at most $k$.

Algorithms with a nice parametric dependence are only known for specific target classes $\Hcal$.
Namely, when $\Hcal$ is the class of forests, {\sc Edge Deletion to $\Hcal$} corresponds to {\sc Feedback Edge Set}, which can be solved in constant time given that the size of a minimum feedback edge set is $m-n+1$ (assume the graph is connected).
When $\Hcal$ is the class of graphs that are a union of paths, then there is a linear kernel for the problem~\cite{LiFCH17impr}, as well as a \FPT algorithm with parametric dependence on $k$ at most $2^k$~\cite{Tsur23fast}.
We refer the reader to the survey of~\cite{CrespelleDFG23asur}, as well as~\cite{DumasPRT22poly}, for other results with explicit dependence on $k$ when $\Hcal$ is not a minor-closed graph class.
 \medskip

Given a graph $G$ and a set of edges $F\subseteq E(G)$, we denote by $G/F$ the graph obtained from $G$ after contracting the edges in $F$.

\defproblem{Edge Contraction to $\Hcal$}
{A graph $G$ and $k\in\bN$.}
{Is there a set $F\subseteq E(G)$ of size at most $k$ such that $G/F\in\Hcal$?}
$(G,k)$ is a \yes-instance of {\sc Edge Contraction to $\Hcal$} if and only if $(G,2k)$ is a \yes-instance of {\sc $\Lcal_{{\sf Con},k}$-R-$\Hcal$}, where
 $\Lcal_{{\sf Con},k}$ is the function that maps each graph $H_1$ to the set of pairs $(H_1/F,\phi)$ over all $F\subseteq E(G)$ of size at most $k$, where $\phi$ maps $v\in V(H_1)$ to the corresponding vertex of $H_1/F$.

An explicit parametric dependence was given in~\cite{HeggernesHLLP14cont} when $\Hcal$ is a class of paths (running time $2^{k+o(k)}+n^{\Ocal(1)}$) or the class of trees (running time $4.98^k\cdot n^{\Ocal(1)}$). Though these classes are not minor-closed, we can easily extend these results to the case when $\Hcal$ is the class of unions of paths or the class of forests (up to a $2^{k}$ factor).
\FPT-algorithms with an explicit parametric dependence were also studied when $\Hcal$ is a collection of generalization and restriction of trees~\cite{AgrawalST19onth,AgrawalKST21path}, or when $\Hcal$ is the class of cactus graphs~\cite{KrithikaMT23asin}.
We refer the reader to~\cite{GolovachHP13obta} for more results when the target class is not minor-closed.

\smallskip
\defproblem{Vertex Identification to $\Hcal$}
{A graph $G$ and $k\in\bN$.}
{Is there a set $S\subseteq V(G)$ of size at most $k$ and a partition $(X_1,\dots,X_p)$ of $S$ such that the graph obtained after identifying the vertices in $X_i$ to a single vertex $x_i$, for $i\in[p]$, belongs to $\Hcal$?}
{\sc Vertex Identification to $\Hcal$} reduces to {\sc $\Lcal_{{\sf Id}}$-R-$\Hcal$}, where
 $\Lcal_{{\sf Id}}$ is the function that maps each graph $H_1$ to the set of pairs $(H_2,\phi)$,
where $H_2$ can be obtained from $H_1$ after identifying each $X_i$ of a partition $(X_1,\dots,X_p)$ of some set $S\subseteq V(H_1)$ to a single vertex $x_i$, and $\phi$ maps vertices of $X_i$ to $x_i$ and is the identity on $V(H_1)\setminus S$.

{\sc Vertex Identification to $\Hcal$} is known to admit a kernel of size $2k+1$ when $\Hcal$ is the class of forests~\cite{MorelleST24vert}.
To the authors' knowledge, this is the only known result for this problem.

\smallskip
\defproblem{Independent Set Deletion to $\Hcal$}
{A graph $G$ and $k\in\bN$.}
{Is there an independent set $I\subseteq V(G)$ of size at most $k$ such that $G-I\in\Hcal$?}
{\sc Independent Set Deletion to $\Hcal$} reduces to {\sc $\Lcal_{\sf ISDel}$-R-$\Hcal$}, where
$\Lcal_{\sf ISDel}$ is the function that maps any graph $H_1$ to the set of pairs $(H_1-I,\phi)$ over all independent sets $I\subseteq V(H_1)$, where $\phi$ maps vertices of $I$ to the empty set and is the identity on $V(H_1)\setminus I$.

When $\Hcal$ is the class of forests, the problem is known to be solvable in  time $3.62^k\cdot n^{\Ocal(1)}$~\cite{LiP20anim}.
Concerning other target classes that are not minor-closed, mainly bipartite graphs, let us mention~\cite{FurmanczykKR16onbi,BonamyDFJP19inde,AgrawalJKM018expl}.

\smallskip
To illustrate the versatility of {\sc $\Lcal$-R-$\Hcal$}, let us present some other problems that  can be defined by particular hereditary R-actions, though they do not seem to have been studied when parameterized by the solution size.

\defproblem{(Induced) Matching Deletion to $\Hcal$}
{A graph $G$ and $k\in\bN$.}
{Is there an (induced) matching $M\subseteq E(G)$ of size at most $k$ such that $G-M\in\Hcal$?}
$(G,k)$ is a \yes-instance of {\sc (Induced) Matching Deletion to $\Hcal$} if and only if $(G,2k)$ is a \yes-instance of {\sc $\Lcal_{{\sf mDel},k}$-R-$\Hcal$}, where
$\Lcal_{{\sf mDel},k}$ is defined similarly to $\Lcal_{{\sf eDel},k}$ above, but for (induced) matchings.

There are some results on {\sc Matching Deletion to $\Hcal$} when $k=n$ and $\Hcal$ is the class of forests~\cite{ProttiS18decy,LimaRSS17decy} or bipartite graphs (see~\cite{LimaRSS22onth} for a small survey on the subject).

\smallskip
\defproblem{(Induced) Matching Contraction to $\Hcal$}
{A graph $G$ and $k\in\bN$.}
{Is there an (induced) matching $M\subseteq E(G)$ of size at most $k$ such that $G/M\in\Hcal$?}
$(G,k)$ is a \yes-instance of {\sc (Induced) Matching Contraction to $\Hcal$} if and only if $(G,2k)$ is a \yes-instance of {\sc $\Lcal_{{\sf mCon},k}$-R-$\Hcal$}, where
$\Lcal_{{\sf mCon},k}$ is defined similarly to $\Lcal_{{\sf Con},k}$ above, but for (induced) matchings.

\smallskip
\defproblem{Induced Star Deletion to $\Hcal$}
{A graph $G$ and $k\in\bN$.}
{Is there a set $F\subseteq E(G)$ inducing a star $K_{1,k'}$ with $k'\le k$ such that $G-F\in\Hcal$?}
$(G,k)$ is a \yes-instance of {\sc Star Deletion to $\Hcal$} if and only if $(G,k+1)$ is a \yes-instance of  {\sc $\Lcal_{{\sf StarDel},k}$-R-$\Hcal$}, where
$\Lcal_{{\sf StarDel},k}$ is the function that maps any graph $H_1$ to the set of pairs $(H_1-F,{\sf id}_{V(H_1)})$ over all sets $F\subseteq E(G)$ inducing a subgraph of $K_{1,k}$.\medskip

Given a graph $G$, the \emph{complement} of $G$, denoted by $\overline{G}$, is graph with vertex set $V(G)$ and edge set the edges that do not belong to $E(G)$.

\defproblem{Subgraph Complementation to $\Hcal$}
{A graph $G$ and $k\in\bN$.}
{Is there a set $S\subseteq V(G)$ of size at most $k$ such that the graph obtained after replacing $G[S]$ with its complement $\overline{G[S]}$ belongs to $\Hcal$?}
{\sc Subgraph Complementation to $\Hcal$} reduces to {\sc $\Lcal_{\sf Comp}$-R-$\Hcal$}, where
$\Lcal_{\sf Comp}$ is the function that maps any graph $H_1$ to the singleton containing the pair $(\overline{H_1},{\sf id}_{V(H_1)})$.

The problem got recently studied when $k=n$ for various target classes. We refer the reader to~\cite{AntonyPS25algo} for a survey on the subject.

\subparagraph{Remark.}
Note that some of the R-actions $\Lcal$ corresponding to a graph modification problem above depend on the parameter $k$.
This implies that the corresponding algorithm is {\sl non-uniform} in $k$.
However, this is just an illusion due to the way we define {\sc $\Lcal$-R-$\Hcal$} so that it generalizes all problems at once:
we quantify on the size of the set $S$ of modified vertices, while some problems may use a different quantification, such as the number of  modified edges.
Given a specific graph modification problem {\sc $\Pi$ to $\Hcal$}, we can easily tune the algorithms of this paper so that they work  exactly for the modification and the quantification we consider, and in this case, the algorithm would be {\sl uniform} in~$k$.

\vspace{-1mm}\section{Flat walls}\label{sec_flat}

In this section we deal with flat walls, that are essential to find an irrelevant vertex, using the framework of \cite{SauST24amor}.
More precisely, in \autoref{@postulated}, we introduce walls and several notions concerning them.
In \autoref{@commiseratio}, we provide the definitions of a rendition and a painting.
Using the above notions, in \autoref{@tugendlehre}, we define flat walls and provide some results about them, including two versions of the Flat Wall theorem.
Finally in \autoref{@commodation}, we define canonical partitions, that we use to find an obligatory vertex.
We note that most of the definitions of this section can also be found in \cite{SauST24amor,SauST23kapiI,MorelleSST24fast} with more details and illustrations.

\subsection{Walls and subwalls}\label{@postulated}
We start with some basic definitions about walls.

\subparagraph{Walls.}
Let $k,r\in\bN$. The
\emph{$(k\times r)$-grid} is the
graph whose vertex set is $[k]\times[r]$ and two vertices $(i,j)$ and $(i',j')$ are adjacent if and only if $|i-i'|+|j-j'|=1$.
An \emph{elementary $r$-wall}, for some odd integer $r\geq 3$, is the graph obtained from a
$(2 r\times r)$-grid
with vertices $(x,y)
	\in[2r]\times[r]$,
after the removal of the
``vertical'' edges $\{(x,y),(x,y+1)\}$ for odd $x+y$, and then the removal of
all vertices of degree one.
Notice that, as $r\geq 3$, an elementary $r$-wall is a planar graph
that has a unique (up to topological isomorphism) embedding in the plane $\bR^{2}$
such that all its finite faces are incident to exactly six
edges.
The \emph{perimeter} of an elementary $r$-wall is the cycle bounding its infinite face,
while the cycles bounding its finite faces are called \emph{bricks}.
Also, the vertices
in the perimeter of an elementary $r$-wall that have degree two are called \emph{pegs},
while the vertices $(1,1), (2,r), (2r-1,1), (2r,r)$ are called \emph{corners} (notice that the corners are also pegs).

\begin{figure}[ht]
	\begin{center}
		\includegraphics[scale=0.87]{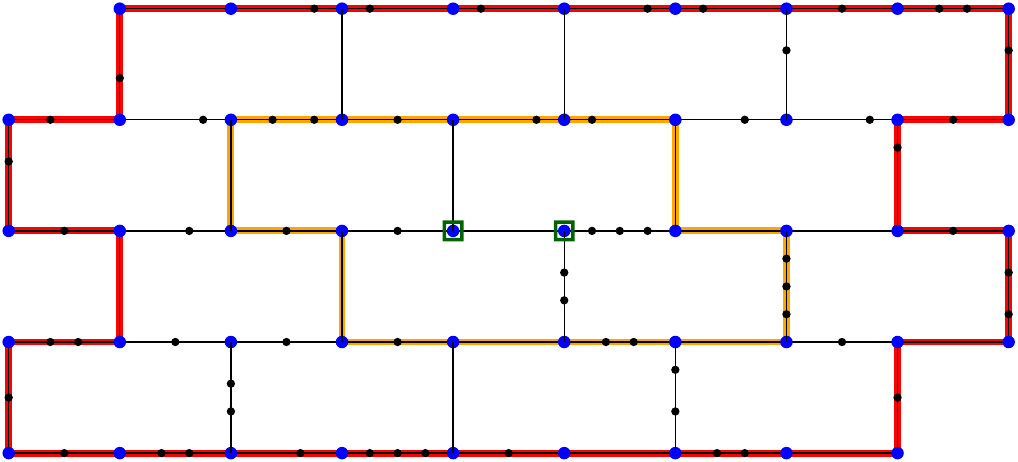}
	\end{center}
	\caption{A $5$-wall. Its first layer is depicted in red and its second layer in orange. Its central vertices are depicted in a green square.}
	\label{@manuscript}
\end{figure}

An \emph{$r$-wall} is any graph $W$ obtained from an elementary $r$-wall $\bar{W}$
after subdividing edges (see \autoref{@manuscript}). A graph $W$ is a \emph{wall} if it is an $r$-wall for some odd $r\geq 3$
and we refer to $r$ as the \emph{height} of $W$. Given a graph $G$,
a \emph{wall of} $G$ is a subgraph of $G$ that is a wall.
We insist that, for every $r$-wall, the number $r$ is always odd.

We call the vertices of degree three of a wall $W$ \emph{3-branch vertices}.
A cycle of $W$ is a \emph{brick} (resp. the \emph{perimeter}) of $W$
if its 3-branch vertices are the vertices of a brick (resp. the perimeter) of $\bar{W}$.
We denote by ${\cal C}(W)$ the set of all cycles of $W$.
We use $D(W)$ in order to denote the perimeter of the wall $W$.
A brick of $W$ is \emph{internal} if it is disjoint from $D(W)$.

\medskip
The following proposition provides an algorithm to find a wall in a graph.

\begin{proposition}[\!\!\cite{SauST22kapiII}]\label{@transforma}
Let $\Fcal$ be a finite collection of graphs.
There exist a function $\newfun{@veneration}: \bN\to\bN$ and an algorithm with the following specifications:\medskip

	\noindent{\tt Find-Wall}$(G,r,k)$\\
	\noindent{\textbf{Input}:} A graph $G$, an odd $r\in\bN_{\geq 3}$, and $k\in\bN$.\\
	\noindent{\textbf{Output}:} One of the following:
	\begin{itemize}
		\item {\bf Case 1:} Either a report that $(G,k)$ is a \no-instance of {\sc Vertex Deletion to $\exc(\Fcal)$}, or
		\item {\bf Case 2:} a report that $G$ has treewidth at most $\funref{@veneration}(s_\Fcal)\cdot r+k$, or
		\item {\bf Case 3:} an $r$-wall $W$ of $G$.
	\end{itemize}
	Moreover, $\funref{@veneration}(s_\Fcal)=2^{\Ocal(s_\Fcal^{2} \cdot \log s_\Fcal)}$, and the algorithm runs in time $2^{\Ocal_{\ell_\Fcal}(r^2+(k+r) \cdot \log (k+r))}\cdot n$.
\end{proposition}

\subparagraph{Subwalls.}
Given an elementary $r$-wall $\bar{W}$, some odd $i\in \{1,3,\ldots,2r-1\}$, and $i'=(i+1)/2$,
the \emph{$i'$-th vertical path} of $\bar{W}$ is the one whose
vertices, in order of appearance, are $(i,1),(i,2),(i+1,2),(i+1,3),
(i,3),(i,4),(i+1,4),(i+1,5),
(i,5),\ldots,(i,r-2),(i,r-1),(i+1,r-1),(i+1,r)$.
Also, given some $j\in[2,r-1]$ the \emph{$j$-th horizontal path} of $\bar{W}$
is the one whose
vertices, in order of appearance, are $(1,j),(2,j),\ldots,(2r,j)$.

A \emph{vertical} (resp. \emph{horizontal}) path of an $r$-wall $W$ is one
that is a subdivision of a vertical (resp. horizontal) path of $\bar{W}$.
Notice that the perimeter of an $r$-wall $W$
is uniquely defined regardless of the choice of the elementary $r$-wall $\bar{W}$.
A \emph{subwall} of $W$ is any subgraph $W'$ of $W$
that is an $r'$-wall, with $r' \leq r$, and such the vertical (resp. horizontal) paths of $W'$ are subpaths of the
{vertical} (resp. {horizontal}) paths of $W$.
	
\subparagraph{Layers.}
The \emph{layers} of an $r$-wall $W$ are recursively defined as follows.
The first layer of $W$ is its perimeter.
For $i=2,\ldots,(r-1)/2$, the $i$-th layer of $W$ is the $(i-1)$-th layer of the subwall $W'$
obtained from $W$ after removing from $W$ its perimeter and
removing recursively all occurring vertices of degree one.
We refer to the $(r-1)/2$-th layer as the \emph{inner layer} of $W$.
The \emph{central vertices} of an $r$-wall are its two 3-branch vertices that do not belong to any of its layers and that are connected by a path of $W$ that does not intersect any layer. See \autoref{@manuscript} for an illustration of the notions defined above.

\subparagraph{Central walls.}
Given an $r$-wall $W$ and an odd $q\in\bN_{\geq 3}$ where $q\leq r$,
we define the \emph{central $q$-subwall} of $W$, denoted by $W^{(q)}$,
to be the $q$-wall obtained from $W$ after removing
its first $(r-q)/2$ layers and all occurring vertices of degree one.

\subparagraph{Tilts.}
The \emph{interior} of a wall $W$ is the graph obtained
from $W$ if we remove from it all edges of $D(W)$ and all vertices
of $D(W)$ that have degree two in $W$.
Given two walls $W$ and $\tilde{W}$ of a graph $G$,
we say that $\tilde{W}$ is a \emph{tilt} of $W$ if $\tilde{W}$ and $W$ have identical interiors.

\subparagraph{Minor models grasped by walls.}
Let $G$ be a graph and $W$ be an $r$-wall in $G$. Let $P_1,\ldots,P_r$ be the horizontal paths and $Q_1,\ldots,Q_r$ be the vertical paths of $W$.
Let $t\in\bN_{\ge 1}$.
A model $\{S_v\mid v\in V(K_t)\}$ of $K_t$ in $G$ is \emph{grasped} by $W$ if,
for all $v\in V(K_t)$, there exist $(i,j)\in[r]^2$ such that $V(P_i)\cap V(Q_j)\subseteq S_v$.

\subsection{Paintings and renditions}\label{@commiseratio}

In this subsection, we present the notions of renditions and paintings, originating in the work of Robertson and Seymour \cite{RobertsonS95XIII}.
The definitions presented here were introduced in \cite{KawarabayashiTW18anew} (see also \cite{BasteST23hittIV,SauST24amor}).
\subparagraph{Paintings.}
A \emph{closed} (resp. \emph{open}) \emph{disk} is a set homeomorphic to the set
$\{(x,y)\in \bR^{2}\mid x^{2}+y^{2}\leq 1\}$ (resp. $\{(x,y)\in \bR^{2}\mid x^{2}+y^{2}< 1\}$).
Let $\Delta$ be a closed disk.
Given a subset $X$ of $\Delta$, we
denote its closure by $\bar{X}$ and its boundary by $\bd(X)$.
A \emph{{$\Delta$}-painting} is a pair $\Gamma=(U,N)$
where
\begin{itemize}
	\item $N$ is a finite set of points of $\Delta$,
	\item $N \subseteq U \subseteq \Delta$, and
	\item $U \setminus N$ has finitely many arcwise-connected components, called \emph{cells}, where for every cell~$c$,
	      \begin{itemize}
		      \item the closure $\bar{c}$ of $c$
		            is a closed disk
		            and
		      \item  $|\tilde{c}|\leq 3$, where $\tilde{c}:=\bd(c)\cap N$.
	      \end{itemize}
\end{itemize}

We use the notation $U(\Gamma) := U$,
$N(\Gamma) := N$ and denote the set of cells of $\Gamma$
by $C(\Gamma)$.
For convenience, we may assume that each cell of $\Gamma$ is an open disk of $\Delta$.

Notice that, given a $\Delta$-painting $\Gamma$,
the pair $(N(\Gamma),\{\tilde{c}\mid c\in C(\Gamma)\})$ is a hypergraph whose hyperedges have cardinality at most three and $\Gamma$ can be seen as a plane embedding of this hypergraph in $\Delta$.

\begin{figure}[h]
\center
\includegraphics[scale=0.5]{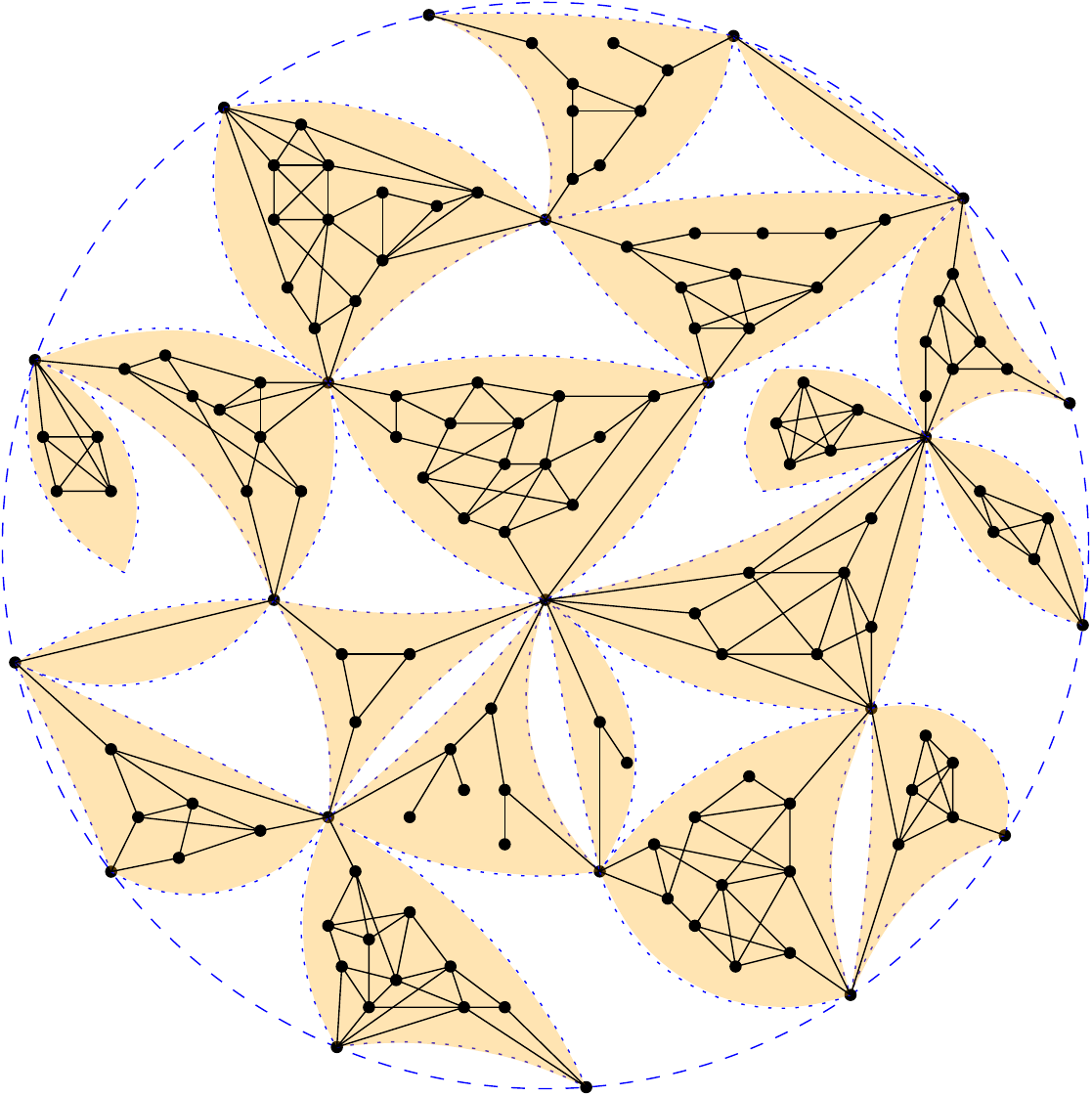}
\caption{A rendition. The cells are depicted in light orange.}
\label{fig_rendition}
\end{figure}

\subparagraph{Renditions.}
See \autoref{fig_rendition} for an illustration.
Let $G$ be a graph and let $\Omega$ be a cyclic permutation of a subset of $V(G)$ that we denote by $V(\Omega)$. By an \emph{$\Omega$-rendition} of $G$, we mean a triple $(\Gamma, \sigma, \pi)$, where
\begin{itemize}
	\item[(a)] $\Gamma$ is a $\Delta$-painting for some closed disk $\Delta$,
	\item[(b)] $\pi: N(\Gamma)\to V(G)$ is an injection, and
	\item[(c)] $\sigma$ assigns to each cell $c \in C(\Gamma)$ a subgraph $\sigma(c)$ of $G$, such that
	      \begin{enumerate}
		      \item[(1)] $G=\bigcup_{c\in C(\Gamma)}\sigma(c)$,
		      \item[(2)]  for distinct $c, c' \in  C(\Gamma)$,  $\sigma(c)$ and $\sigma(c')$  are edge-disjoint,
		      \item[(3)] for every cell $c \in  C(\Gamma)$, $\pi(\tilde{c}) \subseteq V (\sigma(c))$,
		      \item[(4)]  for every cell $c \in  C(\Gamma)$,
		            $V(\sigma(c)) \cap \bigcup_{c' \in  C(\Gamma) \setminus  \{c\}}V(\sigma(c')) \subseteq \pi(\tilde{c})$, and
		      \item[(5)]  $\pi(N(\Gamma)\cap \bd(\Delta))=V(\Omega)$, such that the points
		            in $N(\Gamma)\cap \bd(\Delta)$ appear in $\bd(\Delta)$ in the same ordering
		            as their images, via $\pi$, in $\Omega$.
	      \end{enumerate}
\end{itemize}

\subsection{Flatness pairs}\label{@tugendlehre}

In this subsection we define the notion of a flat wall.
In order for the formal statements of this section to be mathematically correct, we need to introduce a number of notions originating in~\cite{SauST24amor}.
We would like to stress that these notions are needed for the formal statements of the results, but that most of them are not fundamental for the main conceptual contributions of this article.
We refer the reader to \cite{SauST24amor} for a more detailed exposition of these definitions and the reasons for which they were introduced.

\subparagraph{Flat walls.}
Let $G$ be a graph and let $W$ be an $r$-wall of $G$, for some odd integer $r\geq 3$.
We say that a pair $(P,C)\subseteq V(D(W))\times V(D(W))$ is a \emph{choice of pegs and corners for $W$} if $W$ is a subdivision of an elementary $r$-wall $\bar{W}$ where $P$ and $C$ are the pegs and the corners of $\bar{W}$, respectively (clearly, $C\subseteq P$).
To get more intuition, notice that a wall $W$ can occur in several ways from the elementary wall $\bar{W}$,
depending on the way the vertices in the perimeter of $\bar{W}$ are subdivided.
Each of them gives a different selection $(P,C)$ of pegs and corners of $W$.

We say that $W$ is a \emph{flat $r$-wall}
of $G$ if there is a separation $(X,Y)$ of $G$ and a choice $(P,C)$
of pegs and corners for $W$ such that:
\begin{itemize}
	\item $V(W)\subseteq Y$,
	\item $P\subseteq X\cap Y\subseteq V(D(W))$, and
	\item if $\Omega$ is the cyclic ordering of the vertices $X\cap Y$ as they appear in $D(W)$,
	      then there exists an $\Omega$-rendition $(\Gamma,\sigma,\pi)$ of  $G[Y]$.
\end{itemize}

We say that $W$ is a \emph{flat wall}
of $G$ if it is a flat $r$-wall for some odd integer $r \geq 3$.

\begin{figure}[h]
\center
\scalebox{1.2}{\includegraphics{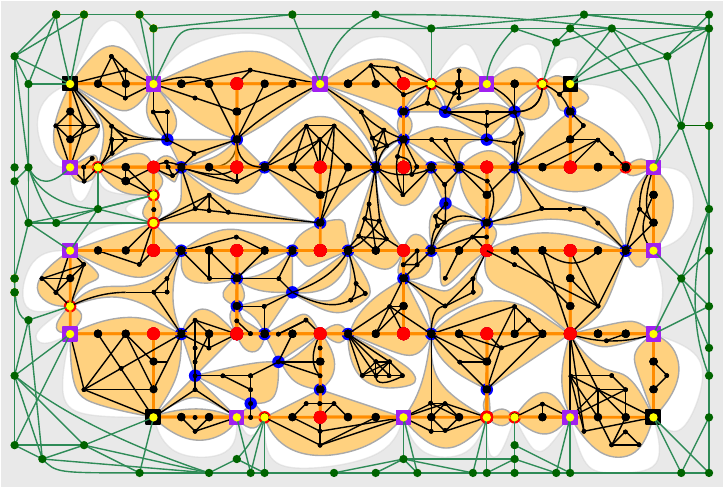}}
\caption{Illustration of a flatness pair $(W,\frak{R})$ of a graph $G$ (adapted from \cite[Figure 4]{BasteST23hittIV}). The edges of $W$ are depicted in orange and the $\Rcal$-compass of $W$ is the union of all parts of $G$ that are drawn in light orange cells. The yellow vertices are the vertices of $V(\Omega)$ and the squared vertices are the choice of pegs (in purple) and corners (in black) of $W$.}
\label{fig_flat}
\end{figure}

\subparagraph{Flatness pairs.}
Given the above, we say that the choice of the 7-tuple $\mathfrak{R}=(X,Y,P,C,\Gamma,\sigma,\pi)$
\emph{certifies that $W$ is a flat wall of $G$}.
We call the pair $(W,\mathfrak{R})$ a \emph{flatness pair} of $G$ and define
the \emph{height} of the pair $(W,\mathfrak{R})$ to be the height of $W$.
See \autoref{fig_flat} for an illustration.

We call the graph $G[Y]$ the \emph{$\mathfrak{R}$-compass} of $W$ in $G$,
denoted by $\compass_\mathfrak{R}(W)$.
We can assume that $\compass_\mathfrak{R} (W)$ is connected, updating $\mathfrak{R}$ by removing from $Y$ the vertices of all the connected components of $\compass_\mathfrak{R} (W)$
except for the one that contains $W$ and including them in $X$ ($\Gamma, \sigma, \pi$ can also be easily modified according to the removal of the aforementioned vertices from $Y$).
We define the \emph{flaps} of the wall $W$ in $\mathfrak{R}$ as
$\flaps_\mathfrak{R}(W):=\{\sigma(c)\mid c\in C(\Gamma)\}$.
Given a flap $F\in \flaps_\mathfrak{R}(W)$, we define its \emph{base}
as $\partial F:=V(F)\cap \pi(N(\Gamma))$.

\medskip
The Flat Wall theorem is a result of Robertson and Seymour \cite{RobertsonS95XIII} that says that a graph $G$ either contains a big clique as a minor, or has bounded treewidth, or contains a flat wall.
Many versions of the Flat Wall theorem were proved over time, including the following one.

\begin{proposition}[\!\!\cite{KawarabayashiTW18anew}]\label{@possession}
There are two functions $\newfun{@classifications}, \newfun{@collaboration}: \bN\to\bN$, such that the images of $\funref{@classifications}$ are odd integers, and an algorithm with the following specifications:\medskip

\noindent{\tt Grasped-or-Flat}$(G,r,t,W)$\\
\noindent{\textbf{Input}:} A graph $G$, an odd $r\in\bN_{\geq 3}$, $t\in\bN_{\geq 1}$, and an $\funref{@classifications}(t)\cdot r$-wall $W$ of $G$.\\
\noindent{\textbf{Output}:} One of the following:
\begin{itemize}
	\item Either a model of a $K_t$-minor in $G$ grasped by $W$, or
	\item a set $A\subseteq V(G)$ of size at most $\funref{@collaboration}(t)$ and a flatness pair $(W',\mathfrak{R}')$ of $G- A$ of heigth $r$ such that $W'$ is a tilt of some subwall $\tilde{W'}$ of $W$.
\end{itemize}
	Moreover, $\funref{@classifications}(t)=\Ocal (t^{26})$, $\funref{@collaboration}(t)=\Ocal (t^{24})$, and the algorithm runs in time $\Ocal(t^{24} m+n)$.
\end{proposition}

\subparagraph{Cell classification.}
Given a cycle $C$ of $\compass_\mathfrak{R}(W)$, we say that
$C$ is \emph{$\mathfrak{R}$-normal} if it is {\sl not} a subgraph of a flap $F\in \flaps_\mathfrak{R}(W)$.
Given an $\mathfrak{R}$-normal cycle $C$ of $\compass_\mathfrak{R}(W)$,
we call a cell $c$ of $\mathfrak{R}$ \emph{$C$-perimetric} if
$\sigma(c)$ contains some edge of $C$.
Notice that if $c$ is $C$-perimetric, then $\pi(\tilde{c})$ contains two points $p,q\in N(\Gamma)$
such that $\pi(p)$ and $\pi(q)$ are vertices of $C$ where one,
say $P_{c}^{\rm in}$, of the two $(\pi(p),\pi(q))$-subpaths of $C$ is a subgraph of $\sigma(c)$ and the other,
denoted by $P_{c}^{\rm out}$, $(\pi(p),\pi(q))$-subpath contains at most one internal vertex of $\sigma(c)$,
which should be the (unique) vertex $z$ in $\partial\sigma(c)\setminus\{\pi(p),\pi(q)\}$.
We pick a $(p,q)$-arc $A_{c}$ in $\hat{c}:={c}\cup\tilde{c}$ such that $\pi^{-1}(z)\in A_{c}$ if and only if $P_{c}^{\rm in}$ contains
the vertex $z$ as an internal vertex.

We consider the circle $K_{C}=\cupall\{A_{c}\mid \mbox{$c$ is a $C$-perimetric cell of $\mathfrak{R}$}\}$
and we denote by $\Delta_{C}$ the closed disk bounded by $K_{C}$ that is contained in $\Delta$.
A cell $c$ of $\mathfrak{R}$ is called \emph{$C$-internal} if $c\subseteq \Delta_{C}$
and is called \emph{$C$-external} if $\Delta_{C}\cap c=\emptyset$.
Notice that the cells of $\mathfrak{R}$ are partitioned into $C$-internal, $C$-perimetric, and $C$-external cells.

A cell $c$ of $\mathfrak{R}$ is \emph{untidy} if $\pi(\tilde{c})$ contains a vertex
$x$ of ${W}$ such that two of the edges of ${W}$ that are incident to $x$ are edges of $\sigma(c)$. Notice that if $c$ is untidy then $|\tilde{c}|=3$.
A cell $c$ of $\mathfrak{R}$ is \emph{tidy} if it is not untidy.

Let $c$ be a tidy $C$-perimetric cell of $\mathfrak{R}$ where $|\tilde{c}|=3$. Notice that $c\setminus A_{c}$ has two arcwise-connected components and one of them is an open disk $D_{c}$ that is a subset of $\Delta_{C}$.
If the closure $\overline{D}_{c}$ of $D_{c}$ contains only two points of $\tilde{c}$ then we call the cell $c$ \emph{$C$-marginal}.

\subparagraph{Influence.}
For every $\mathfrak{R}$-normal cycle $C$ of $\compass_\mathfrak{R}(W)$ we define the set
$\influence_\mathfrak{R}(C)=\{\sigma(c)\mid \mbox{$c$ is a cell of $\mathfrak{R}$ that is not $C$-external}\}$.

A wall $W'$ of $\compass_\mathfrak{R}(W)$ is \emph{$\mathfrak{R}$-normal} if $D(W')$ is $\mathfrak{R}$-normal.
Notice that every wall of $W$ (and hence every subwall of $W$) is an $\mathfrak{R}$-normal wall of $\compass_\mathfrak{R}(W)$. We denote by ${\cal S}_\mathfrak{R}(W)$ the set of all $\mathfrak{R}$-normal walls of $\compass_\mathfrak{R}(W)$. Given a wall $W'\in {\cal S}_\mathfrak{R}(W)$ and a cell $c$ of $\mathfrak{R}$,
we say that $c$ is \emph{$W'$-perimetric/internal/external/marginal} if $c$ is $D(W')$-perimetric/internal/external/marginal, respectively.
We also use $K_{W'}$, $\Delta_{W'}$, $\influence_\mathfrak{R}(W')$ as shortcuts
for $K_{D(W')}$, $\Delta_{D(W')}$, $\influence_\mathfrak{R}(D(W'))$, respectively.

\subparagraph{Regular flatness pairs.}
We call a flatness pair $(W,\mathfrak{R})$ of a graph $G$ \emph{regular}
if none of its cells is $W$-external, $W$-marginal, or untidy.

\medskip
The next result is another version of the Flat Wall theorem.
Compared to \autoref{@possession}, in \autoref{@unimportant},
we lose the fact that the clique-minor is grasped by the wall given in the input, but we gain the fact that the compass of the flat wall has bounded treewidth.

\begin{proposition}[\!\!\cite{SauST24amor}]\label{@unimportant}
There exist a function $\newfun{@corollaries}:\bN\to\bN$ and an algorithm with the following specifications:\medskip

	\noindent{\tt Clique-Or-twFlat}$(G,r,t)$\\
	\noindent{\textbf{Input}:} A graph $G$, an odd $r\in\bN_{\geq 3}$, and $t\in\bN_{\geq 1}$.\\
	\noindent{\textbf{Output}:} One of the following:
	\begin{itemize}
		\item Either a report that $K_t$ is a minor of $G$, or
		\item a tree decomposition of $G$ of width at most $\funref{@corollaries}(t)\cdot r$, or
		\item a set $A\subseteq V(G)$ of size at most $\funref{@collaboration}(t)$ and a regular flatness pair $(W',\mathfrak{R}')$ of $G- A$ of height $r$ whose $\mathfrak{R}'$-compass has treewidth at most $\funref{@corollaries}(t)\cdot r$.
	\end{itemize}
	Moreover, $\funref{@corollaries}(t)=2^{\Ocal(t^2\log t)}$ and this algorithm runs in time $2^{\Ocal_t(r^2)}\cdot n$. The algorithm can be modified to obtain an explicit dependence on $t$ in the running time, namely  $2^{2^{\Ocal(t^2\log t)}\cdot r^3\log r}\cdot n$.
\end{proposition}

\subparagraph{Tilts of flatness pairs.}
Let $(W,\mathfrak{R})$ and $(\tilde{W}',\tilde{\mathfrak{R}}')$ be two flatness pairs of a graph $G$ and let $W'\in {\cal S}_\mathfrak{R}(W)$.
We assume that $\mathfrak{R}=(X,Y,P,C,\Gamma,\sigma,\pi)$ and $\tilde{\mathfrak{R}}'=(X',Y',P',C',\Gamma',\sigma',\pi')$.
We say that $(\tilde{W}',\tilde{\mathfrak{R}}')$ is a \emph{$W'$-tilt} of $(W,\mathfrak{R})$ if
\begin{itemize}
	\item $\tilde{\mathfrak{R}}'$ does not have $\tilde{W}'$-external cells,
	\item $\tilde{W}'$ is a tilt of $W'$,
	\item the set of $\tilde{W}'$-internal cells of $\tilde{\mathfrak{R}}'$ is the same as the set of $W'$-internal
	      cells of $\mathfrak{R}$ and their images via $\sigma'$ and ${\sigma}$ are also the same,
	\item $\compass_{\tilde{\mathfrak{R}}'}(\tilde{W}')$ is a subgraph of $\cupall\influence_\mathfrak{R}(W')$, and
	\item if $c$ is a cell in $C(\Gamma') \setminus C(\Gamma)$, then $|\tilde{c}| \leq 2$.
\end{itemize}

Also, given a regular flatness pair $(W,\mathfrak{R})$ of a graph $G$ and a $W'\in {\cal S}_\mathfrak{R}(W)$,
for every $W'$-tilt $(\tilde{W}', \tilde{\mathfrak{R}}')$ of $(W,\mathfrak{R})$, by definition none of its cells is $\tilde{W}'$-external, $\tilde{W}'$-marginal, or untidy -- thus, $(\tilde{W}', \tilde{\mathfrak{R}}')$ is regular.
Therefore, regularity of a flatness pair is a property that its tilts ``inherit''.

\begin{observation}\label{label_expressionism}
	If $(W,\mathfrak{R})$ is a regular flatness pair, then for every $W'\in {\cal S}_\mathfrak{R}(W)$, every $W'$-tilt of $(W,\mathfrak{R})$ is also regular.
\end{observation}
\medskip

Furthermore, we need the following propositions, that are the main results of~\cite{SauST24amor}.

\begin{proposition}[\!\!\cite{SauST24amor}]\label{@expurgated}
There exists an algorithm that, given a graph $G$, a flatness pair $({W},\mathfrak{R})$ of $G$, and a wall $W'\in {\cal S}_\mathfrak{R}(W)$, outputs a $W'$-tilt of $({W},\mathfrak{R})$ in time $\Ocal(n+m)$.
\end{proposition}

\begin{proposition}[\!\!\cite{SauST24amor}]\label{prop_regular}
There exists an algorithm that, given a graph $G$ and a flatness pair $({W},\mathfrak{R})$ of $G$, outputs in time $\Ocal(n+m)$ a regular flatness pair $({W}^*,\mathfrak{R}^*)$ of $G$ with the same weight as $(W,\frak{R})$ such that $\compass_{\frak{R}^*}(W^*)\subseteq\compass_{\frak{R}}(W)$.
\end{proposition}

\subsection{Canonical partitions}\label{@commodation}

In this subsection, we define the notion of canonical partition of
a graph $G$ with respect to some wall $W$ of $G$. This refers to a partition of the vertex set of $G$ into bags that follow the structure of $W$.
Essentially, the goal is to be able to contract each of these bags to obtain a grid that is a minor of $W$ and thus of $G$.
In particular, we prove in \autoref{sec_obl} that if $G$ contains as a minor a grid $\Gamma$ along with a set $A$  whose vertices have sufficiently many neighbors in the grid, then some vertex in $A$ is obligatory.
We use canonical partitions here to easily find such a structure given a wall of $G$.

For this reason, we start by defining the canonical partition of a wall, as a ``canonical'' way to partition the vertices of the wall into connected subsets that preserve the grid-like structure of the wall.

\subparagraph{Canonical partition of a wall.}
Let $r\geq 3$ be an odd integer.
Let $W$ be an $r$-wall and let $P_{1}, \ldots, P_{r}$ (resp. $L_{1},\ldots, L_{r}$) be its vertical (resp. horizontal) paths.
For every even (resp. odd) $i\in[2,r-1]$ and every $j\in[2,r-1]$, we define ${A}^{(i,j)}$ to be the subpath of $P_{i}$ that starts from a vertex of $P_{i}\cap L_{j}$ and finishes at a neighbor of a vertex in $L_{j+1}$ (resp. $L_{j-1}$), such that $P_{i}\cap L_{j}\subseteq A^{(i,j)}$ and $A^{(i,j)}$ does not intersect $L_{j+1}$ (resp. $L_{j-1}$).
Similarly, for every $i,j\in[2,r-1]$, we define $B^{(i,j)}$ to be the subpath of $L_{j}$ that starts from a vertex of $P_{i}\cap L_{j}$ and finishes at a neighbor of a vertex in $P_{i-1}$, such that $P_{i}\cap L_{j}\subseteq B^{(i,j)}$ and $B^{(i,j)}$ does not intersect $P_{i-1}$.

\begin{figure}[ht]
	\centering
	\includegraphics[scale=0.7]{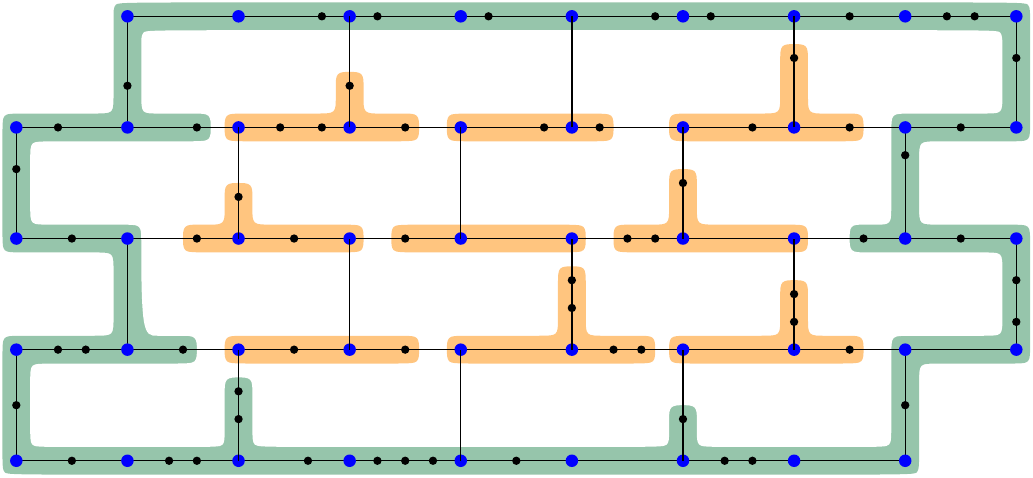}
	\caption{\small A $5$-wall and its canonical partition ${\cal Q}$. The green bag is the external bag $Q_{\rm ext}$ and the orange bags are the internal bags of $\Qcal$. Observe that if we contract each internal bag of $\Qcal$, then we obtain a $(3\times 3)$-grid.}
	\label{@aberration}
\end{figure}

For every $i,j\in[2,r-1]$, we denote by $Q^{(i,j)}$ the graph $A^{(i,j)}\cup B^{(i,j)}$ and by ${Q_{\rm ext}}$ the graph $W- \bigcup_{i,j\in[2,r-1]} V(Q_{i,j})$.
Now consider the collection ${\cal Q}=\{Q_{\rm ext}\}\cup\{Q_{i,j}\mid i,j\in[2,r-1]\}$
and observe that the graphs in ${\cal Q}$ are connected subgraphs of $W$ and their vertex sets form a partition of $V(W)$.
We call ${\cal Q}$ the \emph{canonical partition} of $W$. Also, we call every $Q_{i,j}, i,j\in[2,r-1]$ an \emph{internal bag} of ${\cal Q}$, while we refer to $Q_{\rm ext}$ as the \emph{external bag} of ${\cal Q}$. See \autoref{@aberration} for an illustration of the notions defined above.
For every $i\in[(r-1)/2]$, we say that a set $Q\in {\cal Q}$ is an \emph{$i$-internal bag of ${\cal Q}$} if $V(Q)$ does not contain any vertex of the first $i$ layers of $W$.
Notice that the $1$-internal bags of ${\cal Q}$ are the internal bags of ${\cal Q}$.

\subparagraph{Canonical partitions of a graph with respect to a wall.}
Let $W$ be a wall of a graph $G$.
Consider the canonical partition ${\cal Q}$ of $W$.
The \emph{enhancement} of the canonical partition $\Qcal$ on $G$ is the following operation.
We set $\tilde{\cal Q}:={\cal Q}$
and, as long as there is a vertex $x\in G- V(\cupall \tilde{\cal Q})$ that is adjacent to a vertex of a graph $Q\in \tilde{\cal Q}$, we update $\tilde{\cal Q}:=\tilde{\cal Q}\setminus \{Q\}\cup \{\tilde{Q}\}$, where $\tilde{Q}=G[\{x\}\cup V(Q)]$.
We call the $\tilde{Q}\in\tilde{\cal Q}$ that contains $Q_{\rm ext}$ as a subgraph the \emph{external bag} of $\tilde{\cal Q}$, and we denote it by $\tilde{Q}_{\rm ext}$, while we call \emph{internal bags} of $\tilde{\cal Q}$ all graphs in $\tilde{\cal Q}\setminus \{\tilde{Q}_{\rm ext}\}$.
Moreover, we enhance $\tilde{\cal Q}$ by adding all vertices of $G-\bigcup_{\tilde{Q}\in\tilde{\Qcal}}V(\tilde{Q})$ to its external bag, i.e., by updating $\tilde{Q}_{\rm ext}: = G[V(\tilde{Q}_{\rm ext})\cup V(G)\setminus\bigcup_{\tilde{Q}\in\tilde{\Qcal}}V(\tilde{Q})]$.

We call such a partition $\tilde{\cal Q}$ a \emph{$W$-canonical partition of $G$.}
Notice that a $W$-canonical partition of $G$ is not unique, since the sets in ${\cal Q}$ can be ``expanded'' arbitrarily when introducing vertex $x$.
However, the canonical partition $\Qcal$ of $W$ is unique.

Let $W$ be an $r$-wall of a graph $G$, for some odd integer $r\geq 3$ and let
$\tilde{\cal Q}$ be a $W$-canonical partition of $G$.
For every $i\in[(r-1)/2]$, we say that a set
$Q\in \tilde{\cal Q}$ is an \emph{$i$-internal bag of $\tilde{\cal Q}$}
if it contains an $i$-internal bag of ${\cal Q}$ as a subgraph.
\medskip

The next result is proved in \cite{SauST23kapiI} and
intuitively states that, given a flatness pair $(W,\mathfrak{R})$ of big enough height and a $W$-canonical partition $\tilde\Qcal$ of $G$, we can find a packing of subwalls of $W$ that are inside some central part of $W$ and such that the vertex set of every internal bag of $\tilde{\cal Q}$ intersects the vertices of the flaps in the influence of at most one of these walls.

\begin{proposition}[\!\!\cite{SauST23kapiI}]\label{@prohibitions}
There exist a function $\newfun{@idealistic}: \bN^3 \to \bN$ and an algorithm with the following specifications:

\medskip
\noindent{\tt Packing}$(l,r,p,G,W,\frak{R},\tilde\Qcal)$\\
\noindent\textbf{Input:} Integers $l,r,p\in\bN_{\geq 1}$, where $r\ge3$ is odd, a graph $G$, and a flatness pair $(W,\mathfrak{R})$ of $G$ of height at least $\funref{@idealistic}(l,r,p)$.\\
\noindent\textbf{Output:} A collection $\Wcal=\{W^1, \ldots, W^l\}$ of $r$-subwalls of $W$ such that, for every $W$-canonical partition $\tilde\Qcal$ of $G$,
	\begin{itemize}
		\item for every $i \in [l]$, $\bigcup \influence_\mathfrak{R}(W^i)$ is a subgraph of $\bigcup \{Q\mid Q \text{ is a $p$-internal bag of }\tilde\Qcal\}$ and
		\item for every $i,j\in[l]$, with $i\neq j$, there is no internal bag of $\tilde\Qcal$ that contains vertices of both $V(\bigcup \influence_\mathfrak{R}(W^i))$ and $V(\bigcup \influence_\mathfrak{R}(W^j))$.
	\end{itemize}
	Moreover, $\funref{@idealistic}(l,r,p)= \Ocal(\sqrt{l}\cdot r + p)$ and the algorithm runs in time $\Ocal(n+m)$.
\end{proposition}

\vspace{-1mm}
\section{The algorithms}\label{sec_algo}

In this section, we provide our two algorithms.
In \autoref{subsec_prelim}, we state the three main ingredients necessary for the algorithms,  that will be provided in later sections.
Namely, in \autoref{subsec_general}, we give the algorithm for the general case.
In \autoref{subsec_planar}, we explain how to improve the algorithm in the special case where $\Hcal$ is the class of graphs embeddable in a surface of bounded genus.

\subsection{Main ingredients}\label{subsec_prelim}

The first ingredient is a result stating that an irrelevant vertex can be found in a big enough flat wall whose compass has bounded treewidth.
The proof is deferred to \autoref{subsec_irr}.

\begin{restatable}{theorem}{thirr}
\label{th_irr}
Let $\Fcal$ be a finite collection of graphs and $\Lcal$ be a hereditary R-action. 
There exist a function $\newfun{something}:\bN^2\to\bN$, whose images are odd integers, and an algorithm with the following specifications:\medskip

\noindent{\tt Irrelevant-Vertex}$(G,S',H_2',\phi',k,A,a,W,\frak{R},t)$\\
\noindent\textbf{Input:} Integers $k,a,t\in\bN$, a graph $G$, a set $S'\subseteq V(G)$ of size at most $k$, $(H_2',\phi')\in\Lcal(G[S'])$, a set $A\subseteq V(G')$ of size at most $a$, where $G':=G_{(H_2,\phi)}^{S'}$, and a regular flatness pair $(W,\frak{R})$ of $G'-A$ of height at least $\funref{something}(k,a)$ whose $\frak{R}$-compass has treewidth at most $t$ and does not intersect $\phi'(S')$.\\
\noindent\textbf{Output:} A vertex $v\in V(G)\setminus S'$ such that $(G,S',H_2',\phi',k)$ and $(G-v,S',H_2',\phi',k)$ are equivalent instances of \apb.

Moreover, $\funref{something}(k,a)=\Ocal_{a,\ell_\Fcal}(k^c)$, where $c:=\funref{@withdrawing}(a,\funref{label_unbelievability}(a,\ell_{\mathcal{F}}))=\Ocal_{a,\ell_\Fcal}(1)$, and the algorithm runs in time $2^{\Ocal_{a,\ell_\Fcal}(k \log k+t\log t)}\cdot(n+m)$.
\end{restatable}

Here is a result with a better dependence on $k$, and a better running time that we will be able to apply in the bounded genus case.
In this case, we do not ask for our flat wall to have bounded treewidth, but to have a planar embedding instead.
The proof is deferred to \autoref{subsec_pl}.
Note that here, instead of a single vertex $v$, we might sometimes find an entire planar block of vertices $V$ that is irrelevant.

\begin{restatable}{theorem}{thirrpl}
\label{th_irr_pl}
Let $\Lcal$ be a hereditary R-action and $\Fcal$ be the collection of obstructions of the graphs embeddable in a surface of genus at most $g$.
There exist a function $\newfun{else}:\bN\to\bN$, whose images are odd integers, and an algorithm with the following specifications:\medskip

\noindent{\tt Planar-Irrelevant-Vertex}$(G,S',H_2',\phi',k,W,\frak{R})$\\
\noindent\textbf{Input:} An integer $k\in\bN$, a graph $G$, a set $S'\subseteq V(G)$ of size at most $k$, $(H_2',\phi')\in\Lcal(G[S'])$, and a flatness pair $(W,\frak{R}=(X,Y,P,C,\Gamma,\sigma,\pi))$ of $G_{(H_2',\phi')}^{S'}$ of height at least $\funref{else}(k)$ whose $\frak{R}$-compass does not intersect $\phi'(S')$ and is embeddable in a disk with $X\cap Y$ on the boundary.\\
\noindent\textbf{Output:} A non-empty set $Y\subseteq V(G)\setminus S'$ such that $(G,S',H_2',\phi',k)$ and $(G-Y,S',H_2',\phi',k)$ are equivalent instances of {\sc $\Lcal$-AR-$\exc(\Fcal)$}.

Moreover, $\funref{else}(k)=\Ocal(k)$ and the algorithm runs in time $\Ocal(n+m)$.
\end{restatable}

The next result essentially states that a part of the solution $S$ can be found in a set $A$ of size $a_\Fcal$ such that each vertex of $A$ is adjacent to many vertices of a big enough wall. This is our ``obligatory vertex'' method.
See \autoref{fig_obl_vtx} for an illustration.
The proof is deferred to \autoref{sec_obl}.

\begin{figure}[h]
\center
\includegraphics[scale=0.6]{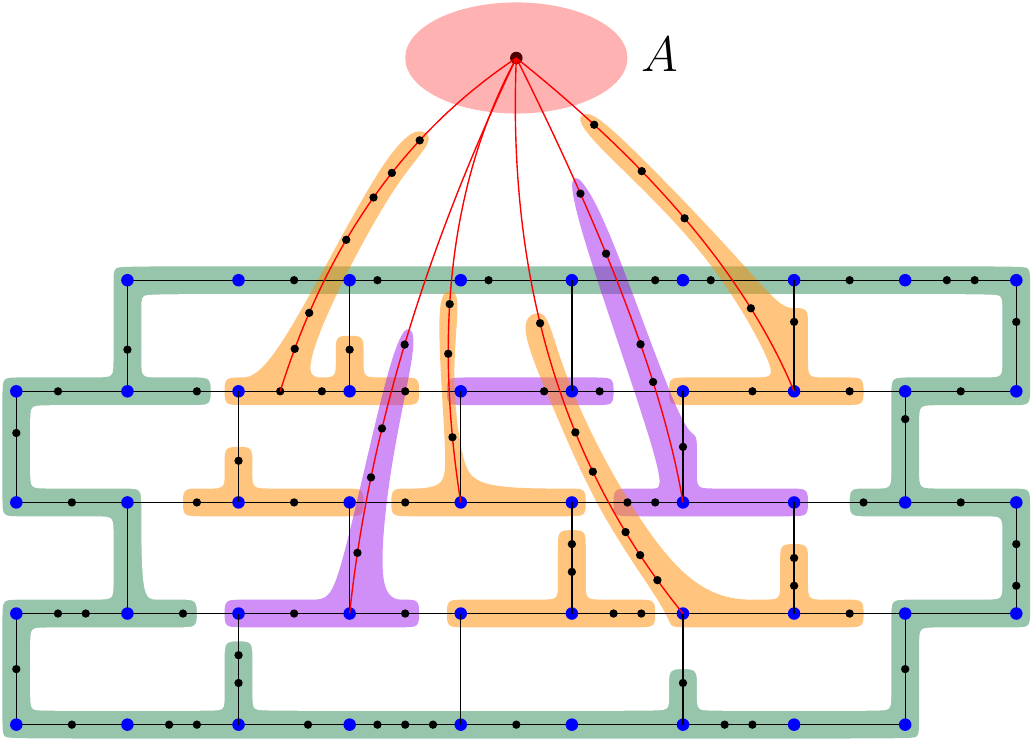}
\caption{Illustration of \autoref{lem_obl}.}
\label{fig_obl_vtx}
\end{figure}

\begin{restatable}{lemma}{lemobl}
\label{lem_obl}
Let $\mathcal{F}$ be a finite collection of graphs and $\Lcal$ be a hereditary R-action.
There exist three functions $\newfun{label_un}, \newfun{label_deux},\newfun{label_trois}: \mathbb{N}\to \mathbb{N}$ such that the following holds.
	
Let $k\in\mathbb{N}$.
Let $G$ be a graph, $S'\subseteq V(G)$ be a set of size at most $k$, and $(H_2',\phi')\in\Lcal(G[S'])$.
Suppose that $G':=G^{S'}_{(H_2',\phi')}$ contains a set $A\subseteq V(G')$ of size at least $a_\Fcal$ and that there is a wall $W$
in $G'-A$ of height $\funref{label_un}(k)$.
Suppose also that there is a $W$-canonical partition $\tilde{\Qcal}$ of $G'-A$ such that each vertex of $A$ is adjacent to at least $\funref{label_deux}(k)$ many $\funref{label_trois}(k)$-internal bags of $\tilde{\Qcal}$.

Then,
for every solution $(S,H_2,\phi)$ of \apb for $(G,S',H_2',\phi')$,
it holds that $A'\ne\emptyset$, where $A':=(S\setminus S')\cap A$, and that
 $|\phi^+(A')|<|A'|$.

Moreover $\funref{label_un}(k)=\Ocal_{s_\Fcal}(k^2)$, $\funref{label_deux}(k)=\Ocal_{s_\Fcal}(k^3)$, and $\funref{label_trois}(k)=\Ocal_{s_\Fcal}(k^2)$.
\end{restatable}

Finally, the dynamic programming algorithm presented in \autoref{sec_tw} gives the following for graphs of bounded treewidth.

\begin{restatable}{theorem}{thtw}
\label{th_tw}
Let $\Fcal$ be a finite collection of graphs and $\Lcal$ be an R-action.
There is an algorithm that, given
$k\in\mathbb{N}$,
a graph $G$ of treewidth at most $w$, a set $S'\subseteq V(G)$ of size at most $k$, and $(H_2',\phi')\in\Lcal(G[S'])$,
 in time $2^{\Ocal_{\ell_\Fcal}(k^2+(k+w)\log(k+w))}\cdot n$ either
outputs a solution of \apb for the instance $(G,S',H_2',\phi',k)$,
or reports a \no-instance.
\end{restatable}

\subsection{The general case: proof of \autoref{th_an}}\label{subsec_general}

We now prove our result in the general case. We restate \autoref{th_an} for the sake of readability.

\than*

Let $\Lcal$ be a hereditary R-action and $\Fcal$ be a finite collection of graphs.
Let $G$ be a graph, $S'\subseteq V(G)$, $(H_2',\phi')\in\Lcal(G[S'])$, and $k\in\bN$.
Let us describe here how to solve \apb on $(G,S',H_2',\phi',k)$.

\smallskip
We set $G':=G^{S'}_{(H_2',\phi')}$ and define the following constants, where $c=\funref{@withdrawing}(a+b,\funref{label_unbelievability}(a,\ell_{\mathcal{F}}))=\Ocal_{\ell_\Fcal}(1)$.
\begin{align*}
	a =				&	\ \funref{@collaboration}(s_\Fcal+a_\Fcal-1)=\Ocal_{\ell_\Fcal}(1), &
	b =				& \ \funref{@collaboration}(s_\Fcal)=\Ocal_{\ell_\Fcal}(1), \\
	q = 			& \ \funref{label_deux}(k)=\Ocal_{\ell_\Fcal}(k^3), &
	p =       & \ \funref{label_trois}(k)=\Ocal_{\ell_\Fcal}(k^2), \\
	l =       & \ (q-1)\cdot (k+b)=\Ocal_{\ell_\Fcal}(k^4), &
	r_5 =     & \ \funref{something}(k,a+b)=\Ocal_{\ell_\Fcal}(k^c), \\
	t =				& \ \funref{@corollaries}(s_\Fcal)\cdot r_5=\Ocal_{\ell_\Fcal}(k^c), &
	r_4 =     & \ \odd(t+3)=\Ocal_{\ell_\Fcal}(k^c), \\
	r_3 = 		& \ \funref{@idealistic}(a_\Fcal+k,r_4,1)=\Ocal_{\ell_\Fcal}(k^{c+\frac{1}{2}}), &
	r_2 = 		& \ 2+\funref{@classifications}(s_\Fcal+a_\Fcal-1) \cdot r_3=\Ocal_{\ell_\Fcal}(k^{c+\frac{1}{2}}), \\
	r_2' =    & \ \odd(\max\{\funref{label_un}(k), \funref{@idealistic}(l+1,r_2,p)\})=\Ocal_{\ell_\Fcal}(k^{c+\frac{5}{2}}), &
	r_1 =     & \ \odd( \funref{@classifications}(s_\Fcal)\cdot r_2'+k)=\Ocal_{\ell_\Fcal}(k^{c+\frac{5}{2}}).
\end{align*}

Recall from the conventions in \autoref{subsec_defpb} that we assume that $G$ has $\Ocal_{s_\Fcal}(k\sqrt{\log k}\cdot n)$ edges.

\medskip
Given that the algorithm is rather convoluted, we split it into three parts.
In the initial steps (Steps 1 and 2), we either find a big enough wall or conclude.
Then we analyze what happens when $(G,S',H_2',\phi',k)$ is a \yes-instance  of \apb containing a big enough wall.
That is, we prove that after Step 3, in case of a \yes-instance, we either find a flat wall whose compass has bounded treewidth in which case  we find an irrelevant vertex in Step 4, or
we go to Step 5 and find an apex set intersecting any solution, on which we can branch.
Hence, we can apply the final steps (Step 3 to 5), where we either recurse or output a \no-instance.

\subsection*{Initial steps}

\subparagraph{Step 1 (basic check).}
If $|S'|>k$, we can safely report a \no-instance.
Hence, we assume in what follows that $|S'|\le k$.

\subparagraph{Step 2 (finding a wall).}
We run the algorithm {\tt Find-Wall} from \autoref{@transforma} with input $(G,r_1,k)$
and, in time $2^{\Ocal_{\ell_\Fcal}(r_1^2+(k+r_1)\log(k+r_1))}\cdot n=2^{\Ocal_{\ell_\Fcal}(k^{2(c+5/2)})}\cdot n$, we either
\begin{itemize}
	\item conclude that $(G,k)$ is a \no-instance of {\sc Vertex Deletion to $\exc(\Fcal)$}, and thus, by \autoref{obs_deltomodif}, that $(G,S',H_2',\phi',k)$ is a \no-instance of \apb, or
	\item conclude that $\tw(G)\leq \funref{@veneration}(s_\Fcal)\cdot r_1+k$ and solve \apb on 
	$(G,S',H_2',\phi',k)$ in time $2^{\Ocal_{\ell_\Fcal}(k^2+(r_1+k)\log(r_1+k))}\cdot n=2^{\Ocal_{\ell_\Fcal}(k^{c+5/2} \cdot \log k)} \cdot n$ using the algorithm of \autoref{th_tw}, or
	\item obtain an $r_1$-wall $W_1$ of $G$.
\end{itemize}
Since we conclude in the first two cases above, we assume henceforth that we have found a $r_1$-wall $W_1$ of $G$.

\subsection*{Interlude: what happens when $(G,S',H_2',\phi',k)$ is a \yes-instance}
Given a solution $(S,H_2,\phi)$ of \apb of the instance $(G,S',H_2',\phi',k)$, if it exists, let us set $S_r:=S\setminus S'$.
Note that $G'-S_r$ is a subgraph of $G^S_{(H_2,\phi)}$ and thus belongs to $\exc(\Fcal)$.

\begin{figure}[h]
\center
\includegraphics[scale=0.7]{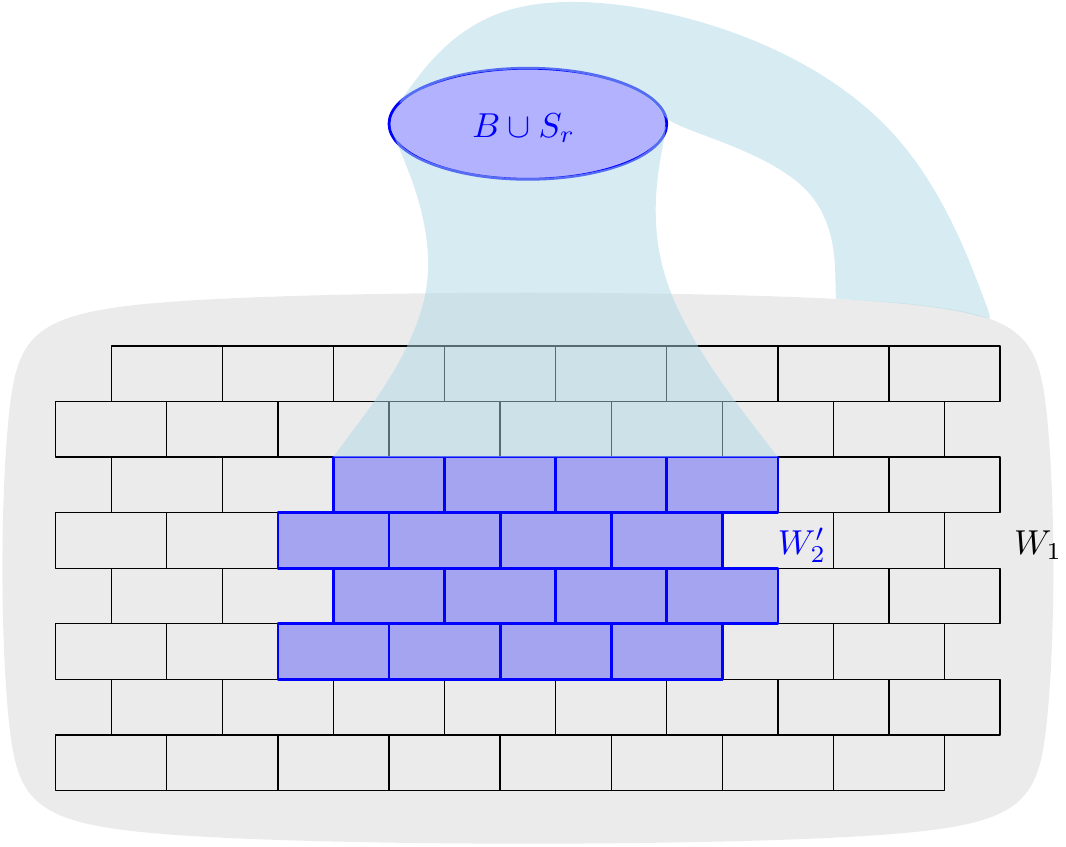}
\caption{$(W_2',\mathfrak{R}_2')$ is a flatness pair of $G'-(S_r\cup B)$.}
\label{fig_W2}
\end{figure}

\begin{claim}\label{int1}
If $(S,H_2,\phi)$ is a solution of \apb for the instance $(G,S',H_2',\phi',k)$, then
there exists a set $B\subseteq V(G')$, with $|B|\leq b$,
and a flatness pair $(W_2',\mathfrak{R}_2')$ of $G'-(S_r\cup B)$ of height $r_2'$ such that $W_2'$ is a $\tilde{W}_2$-tilt of some subwall $\tilde{W}_2$ of $W_1$.
\end{claim}

\begin{cproof}
Since $r_1\geq \funref{@classifications}(s_\Fcal)\cdot r_2'+k$, there is an $( \funref{@classifications}(s_\Fcal)\cdot r_2')$-subwall
of $W_1$, say $W_1^*$, that does not contain vertices of $S$ (by removing the at most $k$ rows and columns containing vertices of $S$).
Hence, $W_1^*$ is a wall of $G-S$ and thus of $G'-S_r\in\exc(\Fcal)$.\smallskip

Since $G'-S_r$ does not contain $K_{s_\Fcal}$ as a minor,
by \autoref{@possession} with input $(G'-S_r, r_2',s_\Fcal, W_1^*)$,
we know that there is a set $B\subseteq V(G')$, with $|B|\leq b$,
and a flatness pair $(W_2',\mathfrak{R}_2')$ of $G'-(S_r\cup B)$ of height $r_2'$ such that $W_2'$ is a $\tilde{W}_2$-tilt of some subwall $\tilde{W}_2$ of $W_1^*$.
\end{cproof}

Let $(W_2',\mathfrak{R}_2')$ be the flatness pair given by \autoref{int1}.
See \autoref{fig_W2} for an illustration.
Let $\Qcal$ be the canonical partition of $W_2'$.
Let $G'_\Qcal$ be the graph obtained by contracting each bag $Q$ of $\Qcal$
to a single vertex $v_Q$, and adding a new vertex $v_{\rm all}$ and making it adjacent to each $v_{Q}$ such that $Q$ is an internal bag of $\Qcal$.
Let $\tilde{A}$ be the set of vertices $y$ of $G'- V(W_2')$ such that there are $q$ internally vertex-disjoint paths from $v_{\rm all}$ to $y$ in $G'_\Qcal$.

Note, as we will use it in Step 5, that, if $\Qcal'$ is the canonical partition of $\tilde{W}_2$, then $\tilde{A}$ is also the set of vertices $y$ of $G'- V(\tilde{W}_2)$ such that there are $q$ internally vertex-disjoint paths from $v_{\rm all}$ to $y$ in $G'_{{\Qcal'}}$.

\begin{claim}\label{int2}
$\tilde{A}\subseteq S_r\cup B$.
\end{claim}

\begin{cproof}
To show this, we first prove that, for every $y\in V(G')\setminus (V(W_2')\cup S_r\cup B)$, the maximum number of internally vertex-disjoint paths from $v_{\sf all}$ to $y$ in $G'_\Qcal$ is
$k+b+4$.

Indeed,
if $y$ is a vertex in the $\mathfrak{R}_2'$-compass of $W_2'$ (but not of $V(W_2')$), then
there are at most $k+b$ such paths that intersect the set $S_r\cup B$ and
at most four paths that do not intersect $S_r\cup B$ (in the graph $G'_\Qcal- (S_r\cup B)$)
due to the fact that $(W_2',\mathfrak{R}_2')$ is a flatness pair of $G'-(S_r\cup B)$.

If $y$ is not a vertex in the $\mathfrak{R}_2'$-compass of $W_2'$, then, since by the definition of flatness pairs the perimeter of $W_2'$ together with the set $S_r\cup B$ separate $y$ from the $\mathfrak{R}_2'$-compass of $W_2'$,
every collection of internally vertex-disjoint paths from $v_{\rm all}$ to $y$ in $G'_\Qcal$ should intersect the set $\{v_{Q_{\rm ext}}\}\cup S_r\cup B$, where $Q_{\rm ext}$ is the external bag of $\Qcal$.

Therefore, in both cases, the maximum number of internally vertex-disjoint paths from $v_{\sf all}$ to $y$ in $G'$ is
$k+b+4$.
Since $k+b+4<q$, we have that $y\notin\tilde{A}$.
Hence, given that $\tilde{A}\subseteq V(G')\setminus V(W_2')$, we conclude that $\tilde{A}\subseteq S_r\cup B$.
\end{cproof}

Given a $W_2'$-canonical partition $\tilde\Qcal$ of $G'-(S_r\cup B)$,
we set $A_{\tilde\Qcal}$ to be the set of vertices in $S_r\cup B$ that are adjacent to vertices of at least $q$ $p$-internal bags of $\tilde\Qcal$.
Note that $A_{\tilde\Qcal}\subseteq\tilde{A}$ and therefore $|A_{\tilde\Qcal}|\leq|\tilde{A}|$.
Remember that $\tilde\Qcal$ is obtained by enhancing $\Qcal$ on $G'-(S_r\cup B)$ and is not unique.

\begin{claim}\label{int3}
If there is a $W_2'$-canonical partition $\tilde\Qcal$ of $G'-(S_r\cup B)$ such that $|A_{\tilde\Qcal}|< a_\Fcal$, then
\begin{itemize}
\item[(a)] there is an $r_2$-subwall  $W_2$ of $W_1$ such that
the algorithm {\tt Grasped-or-Flat} of \autoref{@possession} with input $(D_{W_2},r_3,s_\Fcal+a_\Fcal-1,W_2^*)$ outputs a set $A\subseteq V(D_{W_2})$ with $|A|\leq a$ and a flatness pair $(W_3,\mathfrak{R}_3)$ of $D_{W_2}- A$ of height $r_3$, such that $W_3$ is a tilt of some subwall $\tilde{W}_3$ of $W_2$, where
\begin{itemize}
\item $W_2^*$ is the central $(r_2-2)$-subwall of $W_2$ and
\item $D_{W_2}$ is the graph obtained from $G'$ after removing the perimeter of $W_2$ and taking the connected component containing $W_2^*$, and
\end{itemize}
\item[(b)] the algorithm {\tt Clique-or-twFlat} of \autoref{@unimportant} with input $(D_{W_4},r_5,s_\Fcal)$ outputs a set $A'$ of size at most $b$ and a regular flatness pair $(W_5,\mathfrak{R}_5)$ of $D_{W_4}- A'$ of height $r_5$ whose $\mathfrak{R}_5$-compass has treewidth at most $t$ and does not intersect $\phi'(S')$, where
\begin{itemize}
\item $W_4$ is a wall in the collection $\Wcal=\{W^1, \ldots,W^{a_\Fcal+k}\}$,
\item $W_4^{*}$ is the central $(r_4-2)$-subwall of $W_4$, and
\item $D_{W_4}$ is the graph obtained from $D_{W_2}$ after removing $A$ and the perimeter of $W_4$ and taking the connected component containing $W_4^{*}$.
\end{itemize}
\end{itemize}
\end{claim}

\begin{cproof}
Given that $| A_{\tilde\Qcal}|<a_\Fcal$,
at most $a_\Fcal-1$ vertices of $S_r\cup B$ are adjacent to vertices of at least $q$ $p$-internal bags of $\tilde\Qcal$.
This means that the $p$-internal bags
of $\tilde\Qcal$ that contain vertices adjacent to some vertex of $(S_r\cup B)\setminus  A_{\tilde\Qcal}$ are at most $(q-1)\cdot(k+ b)=l$.\smallskip

Given that $r_2'\geq \funref{@idealistic}(l+1,r_2,p)$,
there is a collection $\Wcal=\{W^{1}, \ldots, W^{l+1}\}$ of $l+1$ $r_2$-subwalls of $W_2'$ in $G'$ respecting the properties of the output of the algorithm {\tt Packing} of \autoref{@prohibitions} with input $(l+1,r_2,p,G',W_2',\mathfrak{R}_2')$.
The fact that the $p$-internal bags
of $\tilde\Qcal$ that contain vertices adjacent to some vertex of $(S_r\cup B)\setminus  A_{\tilde\Qcal}$ are at most $l$ implies that
there exists an $i\in[l+1]$ such that
no vertex of $V(\bigcup \influence_{\mathfrak{R}_2'}({W^i}))$ is adjacent, in $G'$, to a vertex in $(S_r\cup B)\setminus  A_{\tilde\Qcal}$.
Let $W_2$ be the subwall of $W_1$ such that $W^i$ is a tilt of $W_2$.
It exists given that $W^i$ is a subwall of $W_2'$, that is a tilt of some subwall $W'$ of $W_1$.
Remember that $W_2^*$ is the central $(r_2-2)$-subwall of $W_2$, which is also the central $(r_2-2)$-subwall of $W^i$, and that $D_{W_2}$ is the graph obtained from $G'$ by removing the perimeter of $W_2$ and taking the connected component that contains $W_2^*$.

\begin{figure}[h]
\center
\includegraphics[scale=0.8]{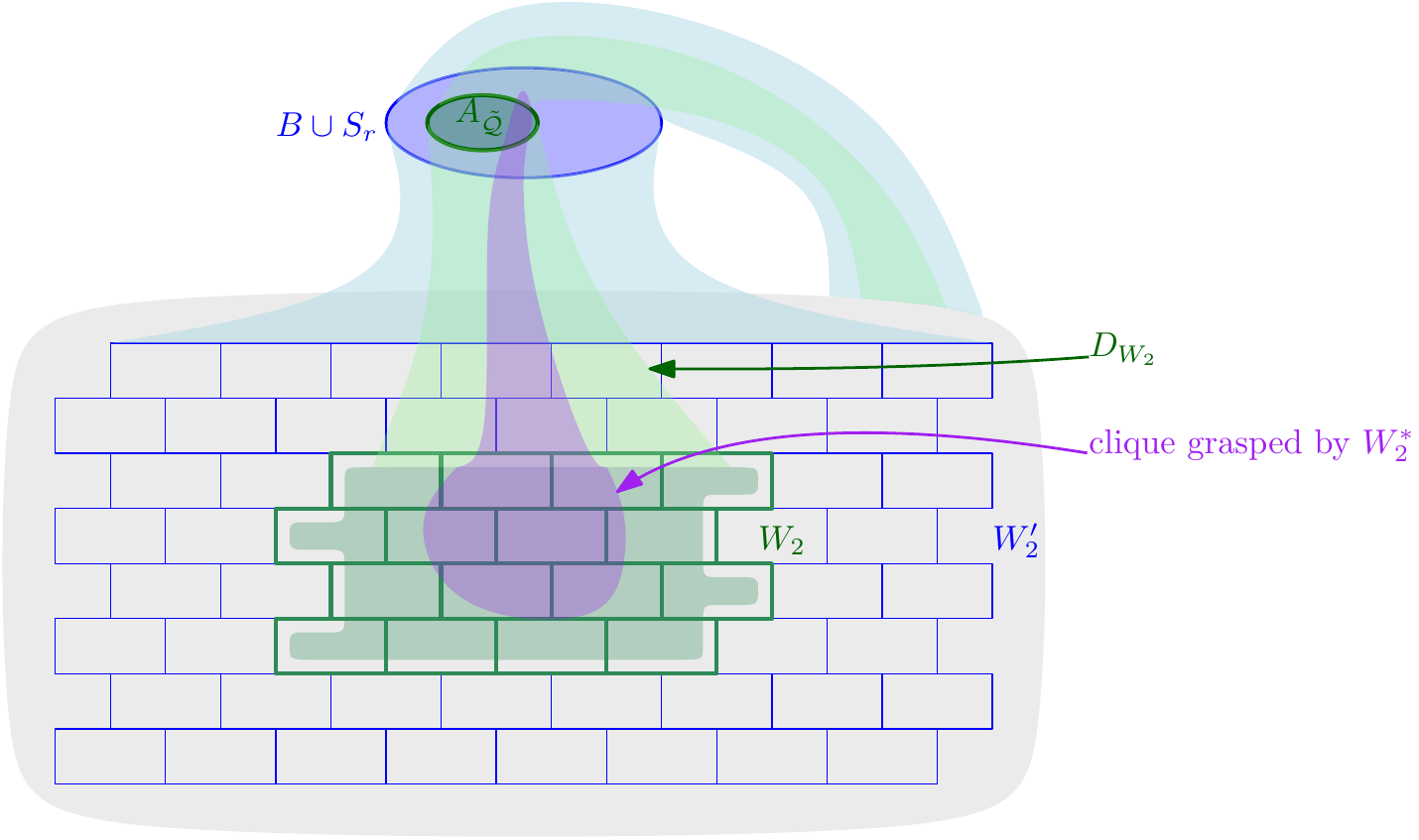}
\caption{$V(\bigcup \influence_{\mathfrak{R}_2'}({W_2}))$ is not adjacent to any vertex in $(S_r\cup B)\setminus  A_{\tilde\Qcal}$.}
\label{fig_W3bis}
\end{figure}

\smallskip
Since no vertex of $V(\bigcup \influence_{\mathfrak{R}_2'}({W^i}))$ is adjacent, in $G'$, to a vertex in $(S_r\cup B)\setminus  A_{\tilde\Qcal}$,
any path in $D_{W_2}$ going from a vertex of $W_2^*$ to a vertex in $S_r$ must intersect a vertex of $ A_{\tilde\Qcal}$.
Thus, there is no model of $K_{s_\Fcal+a_\Fcal-1}$ grasped by $W_2^*$ in $D_{W_2}$, because otherwise, $K_{s_\Fcal}$ would be a minor of the connected component of $D_{W_2}- A_{\tilde\Qcal}$ containing $W_2^*$, and thus of $G'- S_r$.
See \autoref{fig_W3bis} for an illustration.
So, by applying the algorithm {\tt Grasped-or-Flat} of \autoref{@possession} with input $(D_{W_2},r_3,s_\Fcal+a_\Fcal-1,W_2^*)$, since $r_2-2\geq\funref{@classifications}(s_\Fcal+a_\Fcal-1) \cdot r_3$, we should find a set $A\subseteq V(D_{W_2})$ with $|A|\leq a$ and a flatness pair $(W_3,\mathfrak{R}_3)$ of $D_{W_2}- A$ of height $r_3$, such that $W_3$ is a tilt of some subwall $\tilde{W}_3$ of $W_2$.\medskip

Let $\tilde{\Qcal}'$ be a $W_3$-canonical partition of $D_{W_2}- A$.
Given that $r_3\geq \funref{@idealistic}(a_\Fcal+k,r_4,1)$,
there is a collection $\Wcal'=\{W^1,\ldots,W^{a_\Fcal+k}\}$ of $r_4$-subwalls of $W_3$
respecting the properties of the output of the algorithm {\tt Packing} of \autoref{@prohibitions} with input $(a_\Fcal+k,r_4,1,D_{W_2}- A,W_3,\mathfrak{R}_3)$.
Since $| A_{\tilde\Qcal}|<a_\Fcal$ and $|(\phi')^+(S')|\le|S'|\le k$, there is an $i\in[a_\Fcal+k]$ such that $V(\bigcup \influence_{\mathfrak{R}_3}(W^i))$ does not intersect $ A_{\tilde\Qcal}$ nor $\phi'(S')$.
See \autoref{fig_W4} for an illustration.

\begin{figure}[h]
\center
\includegraphics[scale=0.78]{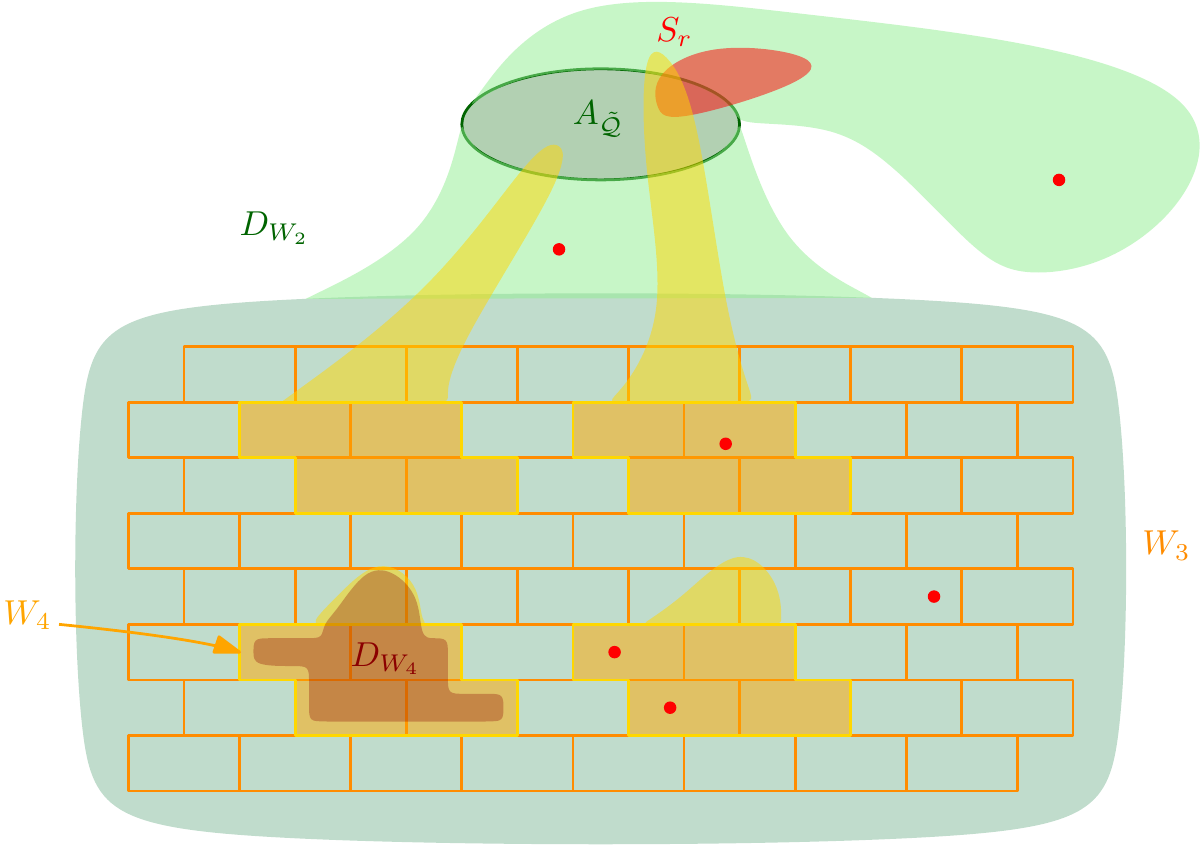}
\caption{$V(\bigcup \influence_{\mathfrak{R}_3}({W_4}))$ does not contain any vertex from $\phi'(S')$ (red vertices), nor from $A_{\tilde{Q}}$ (and thus from $S_r$).
The small walls in orange are the walls of $\Wcal'$, and their influence is represented in yellow.}
\label{fig_W4}
\end{figure}

Let $W_4:=W^i$.
Remember that $W_4^*$ is the central $(r_4-2)$-subwall of $W_4$ and that $D_{W_4}$ is the graph obtained from $D_{W_2}$ after removing $A$ and the perimeter of $W_4$ and taking the connected component containing $W_4^*$.
Observe that any path between a vertex of $S_r$ and a vertex of $V(\bigcup \influence_{\mathfrak{R}_3}(W_4))$ in $D_{W_2}$ intersects $ A_{\tilde\Qcal}$.
Since $ A_{\tilde\Qcal}$ does not intersect $V(\bigcup \influence_{\mathfrak{R}_3}(W_4))$, it implies that $ A_{\tilde\Qcal}$ does not intersect $D_{W_4}$, and thus $S_r\cap D_{W_4}=\emptyset$.
Therefore, $D_{W_4}$ is a subgraph of $G'- S_r$ and $K_{s_\Fcal}$ is not a minor of $D_{W_4}$.
Moreover, $W_4^*$ is a wall of $D_{W_4}$ of height $r_4-2\geq t+1$, so $\tw(D_{W_4})>t=\funref{@corollaries}(s_\Fcal)\cdot r_5$.
Therefore, by applying the algorithm {\tt Clique-or-twFlat} of \autoref{@unimportant} with input $(D_{W_4},r_5,s_\Fcal)$, we should obtain a set $A'$ of size at most $b$ and a regular flatness pair $(W_5,\mathfrak{R}_5)$ of $D_{W_4}- A'$ of height $r_5$ whose $\mathfrak{R}_5$-compass has treewidth at most $t$.
\end{cproof}

\subsection*{Final steps}

\subparagraph{Step 3a (finding a flat wall).}
We consider all the $\binom{r_1}{r_2}^2=2^{\Ocal_{\ell_\Fcal}(k^{c+1/2}\log k)}$ $r_2$-subwalls of $W_1$ not containing a vertex of $S'$ (so that they are walls of both $G$ and $G'$).
For each one of them, say $W_2$, let $W_2^*$ be the central $(r_2-2)$-subwall of $W_2$ and let $D_{W_2}$ be the graph obtained from $G'$ after removing the perimeter of $W_2$ and taking the connected component containing $W_2^*$.
Given that $r_2-2=\funref{@classifications}(s_\Fcal+a_\Fcal-1) \cdot r_3$,
we can apply the algorithm {\tt Grasped-or-Flat} of \autoref{@possession} with input $(D_{W_2},r_3,s_\Fcal+a_\Fcal-1,W_2^*)$.
This can be done in time $\Ocal_{s_\Fcal}(k\sqrt{\log k}\cdot n)$.
If, for some of these subwalls, the result is
a set $A\subseteq V(D_{W_2})$ with $|A|\leq a$ and a flatness pair $(W_3,\mathfrak{R}_3)$ of $D_{W_2}- A$ of height $r_3$ then
we proceed to {\bf Step 3b} for each such a subwall.
Otherwise, we proceed to {\bf Step 5}.

\smallskip
\noindent{\bf Step 3b (finding a flat wall whose compass has bounded treewidth).}
Given that $r_3=\funref{@idealistic}(a_\Fcal+k,r_4,1)$, we
can apply the algorithm {\tt Packing} of \autoref{@prohibitions} with input $(a_\Fcal+k,r_4,1,D_{W_2}-A,W_3,\mathfrak{R}_3)$ to
compute in time $\Ocal_{s_\Fcal}(k\sqrt{\log k}\cdot n)$
a collection $\Wcal=\{W^1,\ldots,W^{a_\Fcal+k-1}\}$ of $r_4$-subwalls of $W_3$
that respects the properties of the output of \autoref{@prohibitions}.

\smallskip
For $i\in[a_\Fcal+k-1]$,
let $W^{i*}$ be the central $(r_4-2)$-subwall of $W^i$ and let $D_{W^i}$ be the graph obtained from $D_{W_2}$ after removing $A$ and the perimeter of $W^i$ and taking the connected component containing $W^{i*}$.
Run the algorithm {\tt Clique-or-twFlat} of \autoref{@unimportant} with input $(D_{W^i},r_5,s_\Fcal)$.
This takes time $2^{\Ocal_{\ell_\Fcal}(r_5^2)}\cdot n=2^{\Ocal_{\ell_\Fcal}(k^{2c})}\cdot n$.
If for one of these subwalls the result is a set $A'$ of size at most $b$ and a regular flatness pair $(W_5,\mathfrak{R}_5)$ of $D_{W^i}- A'$ of height $r_5$ whose $\mathfrak{R}_5$-compass has treewidth at most $t$ and does not intersect $S'$, then we set $W_4:=W^i$ and we proceed to {\bf Step~4}.

\smallskip
If, for every
subwall $W_2$, we did not find such a pair $(W_5,\mathfrak{R}_5)$,
then we proceed to {\bf Step~5}.

\subparagraph{Step 4 (irrelevant vertex case).}
Let $\mathfrak{R}_5'$ be the 7-tuple obtained
by adding all vertices of $G'- V(D_{W_4})-A$ to the set in the first coordinate of $\mathfrak{R}_5$.
\begin{claim}\label{cl_newrend}
$(W_5,\mathfrak{R}_5')$ is a regular flatness pair of $G'- (A\cup A')$ whose $\mathfrak{R}_5'$-compass has treewidth at most $t$ and does not intersect $\phi'(S')$.
\end{claim}
\begin{cproof}
Remember that, given a wall $W$, $D(W)$ is the perimeter of $W$.

$(W_5,\mathfrak{R}_5)$ is a flatness pair of $D_{W_4}- A'$.
By the definition of $D_{W_4}$, the vertices of $D_{W_2}-V(D_{W_4})-A-D(W_4)$ are only adjacent to $D(W_4)$ and $A$ in $D_{W_2}$.
Therefore, $(W_5,\mathfrak{R}_5'')$ is a flatness pair of $D_{W_2}-(A\cup A')$, where $\mathfrak{R}_5''$ is the 7-tuple obtained
by adding all vertices of $D_{W_2}-V(D_{W_4})-A$ to the set in the first coordinate of $\mathfrak{R}_5$.

Also, by the definition of $D_{W_2}$, the vertices of $G'-V(D_{W_2})-D(W_2)$ are only adjacent to $D(W_2)$ in $G'$.
Therefore, $(W_5,\mathfrak{R}_5')$ is a flatness pair of $G'-(A\cup A')$, where $\mathfrak{R}_5'$ is the 7-tuple obtained
by adding all vertices of $G'-D_{W_2}$ to the set in the first coordinate of $\mathfrak{R}_5''$.
Therefore, $\mathfrak{R}_5'$ is indeed the 7-tuple obtained
by adding all vertices of $G'- V(D_{W_4})-A$ to the set in the first coordinate of $\mathfrak{R}_5$.

Given that $\compass_{\frak{R}_5}(W_5)=\compass_{\frak{R}_5'}(W_5)$ and that $(W_5,\mathfrak{R}_5)$ is regular with a $\mathfrak{R}_5$-compass of treewidth at most $t$ that does not intersect $\phi'(S')$, this is also the case for $(W_5,\mathfrak{R}_5')$.
Hence the result.
\end{cproof}
Given that $r_5=\funref{something}(k,a+b)$,
we can apply the algorithm {\tt Irrelevant-Vertex} of \autoref{th_irr} with input $(G,S',H_2',\phi',k,A\cup A',a+b,W_5,\mathfrak{R}_5',t)$, which outputs, in time $2^{\Ocal_{\ell_\Fcal}(t\log t + k\log k)}\cdot (n+m)=2^{\Ocal_{\ell_\Fcal}(k^c\log k)}\cdot n$, a vertex $v$ such that $(G,S',H_2',\phi',k)$ and $(G- v,S',H_2',\phi',k)$
are equivalent instances of {\sc $\Lcal$-R-$\exc(\Fcal)$}.
Then the algorithm runs recursively on the equivalent instance $(G- v,S',H_2',\phi',k)$.

\subparagraph{Step 5 (branching case).}
Consider all the $r_2'$-subwalls of $W_1$ that do not contain vertices of $S'$, which are at most $\binom{r_1}{r_2'}^2=2^{\Ocal_{\ell_\Fcal}(k^{c+5/2}\log k)}$ many,
and for each of them, say $\tilde{W}_2$, compute its canonical partition $\Qcal'$.
Note that $\tilde{W}_2$ is a wall of $G-S'$, and thus of $G'$.
Then, in $G'$, we contract each bag $Q$ of $\Qcal'$ to a single vertex $v_Q$, and add a new vertex $v_{\rm all}$ and make it adjacent to each $v_{Q}$ such that $Q$ is an internal bag of $\Qcal'$. In the resulting graph $G'_{\Qcal'}$, for every vertex $y$ of $G'- V(\tilde{W}_2)$, check, using
augmenting paths from usual maximum flow techniques~\cite{Diestel17grap}, whether there are $q$ internally vertex-disjoint paths from $v_{\rm all}$ to $y$ in time $\Ocal(q\cdot m)=\Ocal_{\ell_\Fcal}(k^4\sqrt{\log k}\cdot n)$.
Let $\tilde{A}$ be the set of all such $y$'s.\smallskip

Note that, if $(G,S',H_2',\phi',k)$ is a \yes-instance, then, by \autoref{int1} and \autoref{int2}, there is a $\tilde{W}_2$ such that $|\tilde{A}|\le k+b$, and by \autoref{int3}, $|\tilde{A}|\ge a_\Fcal$, since otherwise we would have gone to Step 4.\smallskip

Hence, for each $\tilde{W}_2$ such that $a_\Fcal\leq|\tilde{A}|\leq k+b$, we do the following.
We consider all the $\binom{\tilde A}{a_\Fcal}=2^{\Ocal_{\ell_\Fcal}(\log k)}$ subsets of $\tilde{A}$ of size $a_\Fcal$.
For each one of them, say $A^*$, construct a $\tilde{W}_2$-canonical partition $\tilde\Qcal'$ of $G'- A^*$ by
enhancing $\Qcal'$ on $G'- A^*$, such that we first greedily increase the size of the external bag $Q_{\sf ext}$.
Note that if $\tilde{W}_2$ is the wall of \autoref{int1}, then there is a $W_2'$-canonical partition $\tilde{\Qcal}$ of $G'-(S_r\cup B)$ and a set $A_{\tilde\Qcal}\subseteq \tilde A$ such that every bag of $\tilde\Qcal'$ contains exactly one bag of $\tilde\Qcal$.
Therefore, by \autoref{int3} and given that we did not go to Step 4, we conclude that, if $(G,S',H_2',\phi',k)$ is a \yes-instance, then there is a set $A^*$ whose vertices are all adjacent to vertices of $q$ $p$-internal bags of $\tilde\Qcal'$.
Therefore, if, for every $\tilde{W}_2$ such that $a_\Fcal\leq|\tilde{A}|\leq k+b$ and for every $A^*\subseteq \tilde{A}$, the vertices of $A^*$ are not all adjacent to vertices of $q$ $p$-internal bags of $\tilde\Qcal'$, we report a \no-instance.\smallskip

Assume now that we found $\tilde{W}_2$ and $A^*$ such that the vertices of $A^*$ are all adjacent to vertices of $q$ $p$-internal bags of $\tilde\Qcal'$.
Then,
given that $r_2'\ge\funref{label_un}(k)$, $q=\funref{label_deux}(k)$, and $p=\funref{label_trois}(k)$,
by \autoref{lem_obl},
for every solution $(S,H_2,\phi)$ of \apb for the instance $(G,S',H_2',\phi',k)$,
it holds that
$(S\setminus S')\cap A^*\ne\emptyset$.
Let $S'':=S'\cup (S\cap A^*)$ and $(H_2'',\phi'')$ be the restriction of $(H_2,\phi)$ to $S''$.
Hence, we guess $(S'',H_2'',\phi'')$ and solve the instance $(G,S'',H_2'',\phi'',k)$.
Since we add at most $a_\Fcal$ vertices to extend $(S',H_2',\phi')$ to $(S'',H_2'',\phi'')$ and given that $H_2$ has at most $k$ vertices, there are
at most $2^{a_\Fcal}$ choices for $S''$, at most  $2^{a_\Fcal\cdot k}$ choices for $H_2''$, and at most $(k+1)^{a_\Fcal}$ choices for $\phi''$,
which means at most $2^{\Ocal_{\ell_\Fcal}(k)}$ possible guesses for $(S'',H_2'',\phi'')$.
Therefore, the algorithm runs recursively on $(G,S'',H_2'',\phi'',k)$ for each such $(S'',H_2'',\phi'')$.
If one of them is a \yes-instance with solution $(S,H_2,\phi)$, then $(G,S',H_2',\phi',k)$ is a \yes-instance with the same solution.
Otherwise, we report a \no-instance.

\subparagraph{Running time.}
Notice that Step~5, when applied, takes time $2^{\Ocal_{\ell_\Fcal}(k^{c+5/2}\log k)} \cdot n^2$, because we apply the flow algorithm to each of the $2^{\Ocal_{\ell_\Fcal}(k^{c+5/2}\log k)}$ $r_2'$-subwalls and for each vertex of $G$.
However, the search tree created by the branching technique has at most $2^{\Ocal_{\ell_\Fcal}(k)}$ branches and depth at most $k$, since the size of the partial solution strictly increase each time. So Step~5 cannot be applied more than $2^{\Ocal_{\ell_\Fcal}(k^2)}$ times during the course of the algorithm.
Since Step~1 runs in time $\Ocal_{\ell_\Fcal}(1)$, Step~2 runs in time $2^{\Ocal_{\ell_\Fcal}(k^{2(c+5/2)})} \cdot n$, Step~3 in time $2^{\Ocal_{\ell_\Fcal}(k^{2c})}\cdot n$, and
Step~4 in time $2^{\Ocal_{\ell_\Fcal}(k^{c\log c})}\cdot n$,
and that they all may be applied at most $n$ times, the claimed time complexity follows: the algorithm runs in time $2^{\Ocal_{\ell_\Fcal}(k^{2(c+5/2)})} \cdot n^2$.

\subsection{The special case of bounded genus: proof of \autoref{thpl_an}}\label{subsec_planar}

When $\Hcal=\exc(\Fcal)$ is the class of graphs embeddable in a surface $\Sigma_g$ of Euler genus at most $g$, we can modify the general algorithm so that the degree of $k$ in the running time does not depend on $\Hcal$.
This is due to two facts.
First, we now have $a_\Fcal=1$, given that, for some $t$ that depends on $g$, $K_{3,t}$, which has apex number one, does not embed in $\Sigma_g$.
Hence, when applying \autoref{lem_obl}, the obligatory set $A$ contains a unique vertex $v$.
This implies that we do not need to branch on $A$, but instead, $v$ is an ``obligatory vertex''.
In particular, given that the size of $A$ must strictly decrease after the modification, this further implies that $v$ must be deleted.

Second, by \autoref{th_irr_pl}, we can find an irrelevant vertex inside a flat wall whose height is far smaller than the one required in \autoref{th_irr} for the general case.
This changes the algorithm a bit, because the flat wall we require now needs to have a compass \emph{embeddable in a disk} instead of a compass of {\sl bounded treewidth}.
Hence, as in the general case, we find a flat wall (Steps 1, 2, 3a).
But, while in the general case we used this flat wall to further find a flat wall $W$ whose compass has bounded treewidth (Step 3b), we now
either find an obligatory vertex in {Step 4}, or proceed to {Step 5} where we find a flat subwall $W'$ of $W$ whose compass is in $\Hcal$, in which we argue using the argument of \autoref{subsec_genus} that there is another flat subwall $W''$ of $W'$ whose compass is embeddable in a disk, where we finally find an irrelevant vertex.
If we did not find a flat wall in Step 3a, then, similarly to the branching case (Step 5) in the general setting, we find an obligatory vertex in Step 7.

\subsubsection{Finding a disk-embeddable wall}\label{subsec_genus}

In the algorithm, we will need to prove that, if the compass $C$ of a flat wall is embedded in a surface of bounded Euler genus, then $C$ contains a smaller flat wall that is embeddable in a disk such that its perimeter is on the boundary of the disk. Namely, we need this property to apply \autoref{@disreputablepl} when finding an irrelevant vertex. 
To do so, we use the following result from Demaine, Hajiaghayi, and Thilikos~\cite{DemaineHT06theb}.

\begin{proposition}[Lemma 4.7 in \cite{DemaineHT06theb}]\label{prop:nice-embedding}
Let G be a graph $(\emptyset, \emptyset)$-embeddable in a surface $\Sigma$ of Euler genus $g$ and assume that $\tw(G) ≥ 4(r - 12g)(g + 1)$. Then there exists some $(r - 12g, g)$-gridoid $H$, $(\emptyset, F)$-embeddable in $\mathbb{S}_0$ for some $F \subseteq E(H)$ with $|F| \leq g$,  such that there exists some contraction mapping from $G$ to $H$ with respect to their corresponding embeddings.
\end{proposition}

Let us briefly introduce the undefined terms used in the above statement.  A graph $G$ is \emph{$(S, F)$-embeddable} in a surface $\Sigma$, where $S \subseteq V(G)$, $F \subseteq E(G)$, and $F$ is a superset of the edges with an endvertex in $S$, if the graph $G^{-}:=(V (G) \setminus S, E(G) \setminus F)$ admits a \emph{2-cell} embedding in $\Sigma$, that is, an embedding in which every face is homeomorphic to an open disk. For two positive integers $r,k$, a graph $G$ is an  \emph{$(r,k)$-gridoid} if it is $(S,F)$-embeddable in $\mathbb{S}_0$ for some pair $S,F$, where $|F| ≤ k$, $F(G[S]) = \emptyset $, and $G^{-}$ is a
\emph{partially triangulated} $(r' \times r')$-grid embedded in $\mathbb{S}_0$ for some $r' ≥ r$, that is, a graph obtained from an $(r' \times r')$-grid by adding some chords to some of its faces. Finally, without entering into unnecessary technical details, a \emph{contraction mapping} is a strengthening of a graph being a contraction
of another graph that preserves some aspects of the embedding in a surface during the contractions. In a nutshell, the statement of \autoref{prop:nice-embedding} should be interpreted as $H$ occurring as a contraction in $G$ in such a way that $H$ can be embedded ``nicely inside the original embedding of $G$'', in particular the preimages of the faces of $H$ via the contraction mapping are faces of $G$; see \cite[Section 3.1]{DemaineHT06theb} for more details.

In our setting, our goal by using \autoref{prop:nice-embedding} is to guarantee that, in \textbf{Step 5} below, given a large flat wall embedded in a surface of bounded genus, it is possible to find inside it a still large flat {\sl plane} subwall that is {\sl nicely embedded} in the original one. To this end, it is enough to argue that, once we have at hand the gridoid $H$ given by \autoref{prop:nice-embedding}, we can find a large piece of it that is $(\emptyset, \emptyset)$-embeddable in $\mathbb{S}_0$.

Note that, in order to apply \autoref{prop:nice-embedding} in our algorithm and obtain an overall quadratic time, we need to obtain the embedded gridoid $H$ in \emph{linear time}, so that an irrelevant vertex can indeed be found in linear time. Let us briefly sketch how this can be done. The proof of~\cite[Lemma 4.7]{DemaineHT06theb} proceeds by induction on the Euler genus $g$, the base case following by the planar exclusion theorem of Robertson and Seymour~\cite{RobertsonST94quic}. It then distinguishes two cases according to whether the \emph{representativity} of $G$ (also called \emph{face-width} in the literature) it at least  $\ell:=4(r - 12g)$ or not. This can be decided in time $\Ocal(g  \ell n)$ by~\cite[Theorem 10]{CabelloVL12algo}, where $n = |V(G)|$. If this is indeed the case, applying~\cite[Lemma 3.3]{DemaineFHT05sube} yields the desired output. The proof of~\cite[Lemma 3.3]{DemaineFHT05sube} consists in simple local operations based on a notion of distance that uses the {\sl existence} of an object called \emph{respectful tangle}, proved to exist in~\cite[Theorem 4.1]{RobertsonS94XI}. It is easy to verify that these local operations and the definition of the distance function can be done in linear time.

Otherwise, if the representativity is less than $\ell$, the proof first reduces the Euler genus of the surface by applying a so-called \emph{splitting} operation, and then uses the induction hypothesis and a sequence of  local operations that identify a set of edges to be contracted to obtain the desired output. Again, these local operations are easily seen to be done in linear time.

\subsubsection{The algorithm}

We now have all the necessary ingredients to prove the algorithm.
Before proceeding to the proof of \autoref{thpl_an}, we restate it for the sake of readability.

\thplan*

Let $\Lcal$ be a hereditary R-action.
Let $G$ be a graph, $S'\subseteq V(G)$, $(H_2',\phi')\in\Lcal(G[S'])$, and $k\in\bN$.
Let us describe in what follows how to solve {\sc $\Lcal$-AR-$\exc(\Fcal)$} on $(G,S',H_2',\phi',k)$.

\smallskip
We set $G':=G^{S'}_{(H_2',\phi')}$ and define the following constants. Remember that $a_\Fcal=1$ in this section.
\begin{align*}
	a,b =				&	\ \funref{@collaboration}(s_\Fcal)=\Ocal_{\ell_\Fcal}(1), &
	q = 			& \ \funref{label_deux}(k)=\Ocal_{\ell_\Fcal}(k^3), \\
	 l =       & \ (q-1)\cdot (k+b)=\Ocal_{\ell_\Fcal}(k^4), &
	p =       & \ \funref{label_trois}(k)=\Ocal_{\ell_\Fcal}(k^2), \\
	d =				& \ (q-1)\cdot a+k+1=\Ocal_{\ell_\Fcal}(k^3), &
	r_6 = & \ \funref{else}(k)=\Ocal_{\ell_\Fcal}(k), \\
	r_5 = & \ \odd(\max\{12g,(2g+2)\cdot 2r_6\}) = \Ocal_{\ell_\Fcal}(k),&
	r_4 =     & \ 4 r_5\cdot(g+1)+1= \Ocal_{\ell_\Fcal}(k),\\
	r_3 = 		& \ \funref{@idealistic}(d,r_4,1)=\Ocal_{\ell_\Fcal}(k^{5/2}), &
	r_2 = 		& \ 2+\funref{@classifications}(s_\Fcal) \cdot r_3=\Ocal_{\ell_\Fcal}(k^{5/2}), \\
	r_2' =    & \ \odd(\max\{\funref{label_un}(k), \funref{@idealistic}(l+1,r_2,p)\})=\Ocal_{\ell_\Fcal}(k^{9/2}), &
	r_1 =     & \ \odd( \funref{@classifications}(s_\Fcal)\cdot r_2'+k)=\Ocal_{\ell_\Fcal}(k^{9/2}).
\end{align*}
Recall from the conventions in \autoref{subsec_defpb} that we assume that $G$ has $\Ocal(k\sqrt{\log k}\cdot n)$ edges.

\medskip
{\bf Step 1, 2, and 3a} are done as in the general case.
If, for some of the $r_2$-subwalls of $W_1$ not containing a vertex of $S'$, we find
a set $A\subseteq V(D_{W_2})$ with $|A|\leq a$ and a flatness pair $(W_3,\mathfrak{R}_3)$ of $D_{W_2}- A$ of height $r_3$ then
we go to {\bf Step 4}.
Otherwise, we go to {\bf Step 7}. 

\subparagraph{Step 4 (obligatory vertex case 1).}
Let $\frak{R}_3'$ be the 7-tuple obtained by adding all the vertices of $G'-V(D_{W_2})- A$ in the set in the first coordinate of $\frak{R}_3$.
Similarly to \autoref{cl_newrend}, we get that $(W_3,\mathfrak{R}_3')$ is a flatness pair of $G'-A$.
We apply the algorithm of \autoref{prop_regular} to find in time $\Ocal(k\sqrt{\log k}\cdot n)$ a regular flatness pair $(W_3^*,\frak{R}_3^*)$ of $G'-A$ of height $r_3$ such that $\compass_{\frak{R}_3^*}(W_3^*)\subseteq\compass_{\frak{R}_3'}(W_3)$.

Let $\tilde{\Qcal}$ be a $W_3^*$-canonical partition of $G'-A$.
If there is a vertex $v\in A$ that has neighbors in at least $q$ $p$-internal bags of $\tilde{Q}$, then,
given that $r_2'\ge\funref{label_un}(k)$, $q=\funref{label_deux}(k)$, and $p=\funref{label_trois}(k)$, by \autoref{lem_obl}, for every solution $(S,H_2,\phi)$ of \apbpl for the instance $(G,S',H_2',\phi',k)$,
it holds that
$v\in S\setminus S'$ and that
 $|\phi^+(v)|<|v|=1$. In other words, $\phi(v)=\emptyset$, that is, $v$ must be deleted.
Let $S'':=S'\cup\{v\}$ and $\phi'':=\phi'\cup(v\mapsto\emptyset)$.
Hence, if $(H_2',\phi'')\in\Lcal(G[S''])$, then
$(G,S,H_2',\phi',k)$ and $(G,S'',H_2,\phi'',k)$ are equivalent instances of \apbpl.
Hence, the algorithm runs recursively on $(G,S'',H_2,\phi'',k)$ and outputs its result.
Otherwise, if $(H_2',\phi'')\notin\Lcal(G[S''])$, then we report a \no-instance. Thus, we can now assume that every vertex of $A$ has neighbors in at most $q-1$ $p$-internal bags of $\tilde{Q}$ and go to {\bf Step 5}.

\subparagraph{Step 5 (finding a flat wall whose compass is disk-embeddable).}
Given that $r_3=\funref{@idealistic}(d,r_4,1)$,
we apply the algorithm {\tt Packing} of \autoref{@prohibitions} with input $(d,r_4,1,D_{W_2}-A,W_3,\frak{R}_3)$ to find in time $\Ocal(k\sqrt{\log k}\cdot n)$
a collection $\Wcal=\{W^1,\ldots,W^{d}\}$ of $r_4$-subwalls of $W_3^*$
respecting the properties of the output of \autoref{@prohibitions}.
For $i\in[d]$, we apply the algorithm of \autoref{@expurgated} to find in time $\Ocal_{s_\Fcal}(k\sqrt{\log k}\cdot n)$ a $W^i$-tilt $(\tilde{W}^i,\tilde{\frak{R}}_i)$ of $(W_3^*,\frak{R}_3^*)$.

\smallskip
Given that each vertex of $A$ has neighbors in at most $q-1$ $p$-internal bags of $\tilde{Q}$ and that $d\ge(q-1)\cdot a+k+1$, by the pigeonhole principle,
there is $I\subseteq[d]$ of size at least $k+1$ such that, for each $i\in I$, $\compass_{\tilde{\frak{R}}_i}(\tilde{W}^i)$ has no neighbors in $A$.
Note that, by the properties of \autoref{@prohibitions}, the graphs $\compass_{\tilde{\frak{R}}_i}(\tilde{W}^i)$ are pairwise disjoint for $i\in I$.
Let $\tilde{\frak{R}}_i'$ be the 7-tuple obtained by adding all the vertices of $A$ in the set in the first coordinate of $\tilde{\frak{R}}_i$.
Therefore, for each $i\in I$, $(\tilde{W}^i,\tilde{\frak{R}}_i')$ is a 
flatness pair of $G'$.

\smallskip
For each $i\in I$, let $C_i:=\compass_{\tilde{\frak{R}}_i'}(\tilde{W}^i)$ and let $C_i^+$ be the graph obtained from $C_i$ by adding a vertex $v_i$ adjacent to each vertex of $X_i\cap Y_i$, where $X_i$ (resp. $Y_i$) is the first (resp. second) coordinate of $\tilde{\frak{R}}_i$.
For each $i\in I$, we check whether $C_i^+$ embeds in $\Sigma_g$, and if it does, we find such an embedding. This can be done in linear time by using the algorithm of Mohar~\cite{Mohar99alin}.
Given that $|I|\ge k+1$, for any solution $(S,H_2,\phi)$ of \apbpl for the instance $(G,S',H_2',\phi',k)$, there is $j\in I$ such that
$C_j$ does not contain a vertex of $\phi(S)$ (and in particular of $\phi'(S')$).
Therefore, $C_j^+$ must have Euler genus at most $g$.
Thus, if, for each $i\in I$, $C_i^+$ is not embeddable in $\Sigma_g$, we report a \no-instance.
Otherwise, there is $j\in I$ such that $C_j^+$ is embeddable in $\Sigma_g$ and does not contain a vertex of $\phi'(S')$.
Then, by \autoref{prop:nice-embedding} and the discussion after it, given that $\tw(C_j^+)\ge r_4\ge 4 r_5\cdot(g+1)$ and that $r_5\ge 12g$, we can find in linear time an edge set $F\subseteq E(C_j^+)$ of size at most $g$ and an $(r_5,g)$-gridoid $H$ that is $(\emptyset, F)$-embeddable in $\mathbb{S}_0$, such that there exists a contraction mapping from $C_j^+$ to $H$ with respect to their corresponding embeddings.
Let $M$ be the union of $v_j$ and the set of vertices of $H$ incident to edges in $F$. We have $|M|\le 2g+1$.
Given that $r_5\ge (2g+2)\cdot 2r_6$, again by the pigeonhole principle there is an induced subgraph of $H-M$ that is a partially triangulated $(2r_6\times 2r_6)$-grid $\Gamma$.
In particular, there is a contraction mapping from $C_j^+$ to $\Gamma$ with respect to their corresponding embeddings, with $\Gamma$ embedded in $\mathbb{S}_0$.
Given that $\Gamma$ is 3-connected (if we dissolve the corners), {by Whitney33isom's theorem~\cite{Whitney33isom}}, it has a unique embedding in $\mathbb{S}_0$.
In particular, the perimeter of $\Gamma$ bounds a face.
Also, there is an elementary $r_6$-wall $W_6$ that is a subgraph of $\Gamma$.
Given that the maximum degree in $W_6$ is three and that it avoids the vertices in $M$, it implies that $C_j^+$, and thus $G'$, contains a $r_6$-wall $W_6'$ as a subgraph (whose contraction gives $W_6$).
Hence, given that the contractions preserve the embedding, we conclude that there is some rendition $\frak{R}_6$ such that $(W_6',\frak{R}_6)$ is a flat wall of $G'$ whose $\frak{R}_6$-compass does not contain a vertex of $\phi'(S')$ and is embeddable in a disk with $X_6\cap Y_6$ on its boundary, where $X_6$ (resp. $Y_6$) is the first (resp. second) coordinate of $\tilde{\frak{R}}_6$.

\subparagraph{Step 6 (irrelevant vertex case).}
Thus, given that $r_6=\funref{else}(k)$, we can apply the algorithm {\tt Planar-Irrelevant-Vertex} of \autoref{th_irr_pl} with input $(G,S',H_2',\phi',k,W_6',\frak{R}_6)$.
It outputs in time $\Ocal(k\sqrt{\log k}\cdot n)$ a non-empty set $Y\subseteq V(G)\setminus S'$ such that $(G,S',H_2',\phi')$ and $(G-Y,S',H_2',\phi')$ are equivalent instances of \apbpl.
Hence, the algorithm runs recursively for the instance $(G-Y,S',H_2',\phi')$ and concludes.

\subparagraph{Step 7 (obligatory vertex case 2).}
This step is essentially the same as Step 5 of the general case.
Consider all the $r_2'$-subwalls of $W_1$ that do not contain vertices of $S'$, which are at most $\binom{r_1}{r_2'}^2=2^{\Ocal_{\ell_\Fcal}(k^{9/2}\log k)}$ many,
and for each of them, say $\tilde{W}_2$, compute its canonical partition $\Qcal'$.
Note that $\tilde{W}_2$ is a wall of $G-S'$, and thus of $G'$.
Then, in $G'$, we
contract each bag $Q$ of $\Qcal'$ to a single vertex $v_Q$, and add a new vertex $v_{\rm all}$ and make it adjacent to each $v_{Q}$ such that $Q$ is an internal bag of $\Qcal'$. In the resulting graph $G'_{\Qcal'}$, for every vertex $v$ of $G'- V(\tilde{W}_2)$, check, using
augmenting paths from usual maximum flow techniques~\cite{Diestel17grap},
whether there are $q$ internally vertex-disjoint paths from $v_{\rm all}$ to $v$ in time $\Ocal(q\cdot m)=\Ocal_{\ell_\Fcal}(k^4\sqrt{\log k}\cdot n)$.
If this is the case for some $v$, then by \autoref{lem_obl},
it holds that, for every solution $(S,H_2,\phi)$ of \apbpl for the instance $(G,S',H_2',\phi',k)$,
it holds that
$v\in S\setminus S'$ and that
 $|\phi^+(v)|<|v|=1$.
 Hence, we conclude as in {\bf Step 4}.

\smallskip
Note that, if $(G,S',H_2',\phi',k)$ is a \yes-instance, then, by \autoref{int3}, there is such a $v$, since otherwise we would have gone to Step 4.
Therefore, if there is no such a $v$, then we report a \no-instance.

\subparagraph{Running time.}
Step~1, 2, 3a, 4, 5, 6, and 7 respectively take time $\Ocal_{\ell_\Fcal}(1)$, $2^{\Ocal_{\ell_\Fcal}(k^{9})}\cdot n$, $2^{\Ocal_{\ell_\Fcal}(k^{5/2}\log k)}\cdot n$, $\Ocal_{\ell_\Fcal}(k\sqrt{\log k}\cdot n)$, $\Ocal_{\ell_\Fcal}(k\sqrt{\log k}\cdot n)$, $\Ocal_{\ell_\Fcal}(k\sqrt{\log k}\cdot n)$, and $2^{\Ocal_{\ell_\Fcal}(k^{9/2}\log k)}\cdot n^2$.
Given that Step 6 can be applied at most $k$ times, since the size of $S'$ increases by one each time, and  that the other steps can be applied at most $n$ times, the algorithm thus runs in time $2^{\Ocal_{\ell_\Fcal}(k^{9})}\cdot n^2$.

\vspace{-1mm}
\section{Irrelevant vertex}\label{sec_irr}

This section is dedicated to proving \autoref{th_irr} and \autoref{th_irr_pl}.
The irrelevant vertex technique, which originates from \cite{RobertsonS95XIII}, essentially consists in finding inside a flat wall $W$ a smaller flat wall $W'$ that is ``tight'' and ``homogeneous'' (\autoref{@disreputable}), and then arguing that the central vertices of $W'$ are irrelevant with respect to the considered problem, in the sense that they can be removed without affecting the type (positive or negative) of the instance (\autoref{lemma_biiid_irr}).

The proof of \autoref{lemma_biiid_irr} in \autoref{subsec_irr} takes inspiration from the proof of \cite[Lemma 16]{SauST23kapiI}, which corresponds to the particular case of {\sc Vertex Deletion to $\exc(\Fcal)$}.
Indeed, both proofs combine a result of \cite{BasteST23hittIV} (\autoref{label_panlatinismo}) with an auxiliary result (\autoref{subsec_aux}) claiming the existence, inside a flat wall $W$ with a central vertex $v$, of a smaller flat wall $W'$ avoiding vertices of a solution aside from its central part, that also contains $v$.
The more general case of {\sc $\Lcal$-AR-$\exc(\Fcal)$} is however more involved, given that doing some modification is not as straightforward as removing vertices, and that we now have annotations.
In particular, it requires to give in \autoref{@indigenous} a new definition of ``homogeneous flat wall'', that encompasses the one used in \cite{BasteST23hittIV,SauST23kapiI,SauST22kapiII,MorelleSST24fast,SauST24amor}.

In the bounded genus case, instead of \autoref{@disreputable}, we prove in \autoref{@disreputablepl} that a flat wall $W$ can be slightly modified to become homogeneous if it respects some additional planar embeddability conditions.
\autoref{@disreputable} and \autoref{@disreputablepl} are the core ingredients explaining the gap in the running time between the general and the bounded genus case.
\autoref{label_panlatinismo} requires a flat wall that is both ``homogeneous'' and ``tight''.
In \cite{BasteST23hittIV,SauST23kapiI,SauST22kapiII,MorelleSST24fast,SauST24amor}, the tightness condition is implicit, given that a flat wall can always is transformed in a tight flat wall (\autoref{prop_tight}).
In the bounded genus case however, if we transform our homogeneous flat wall, we might lose the homogeneity condition, so we need to define tightness in \autoref{subsec_tight} and to explicitly prove that our homogeneous wall is also tight in \autoref{@disreputablepl}.
Finally, we prove \autoref{th_irr_pl} in \autoref{subsec_pl}.

\subsection{Homogeneous walls}\label{@indigenous}

In this subsection, we define homogeneous flat walls.
 Intuitively, homogeneous flat walls are flat walls
 where each brick has the same ``color'', where a ``color'' express what kind of topological minor can be routed in the (augmented) flaps of the brick.
 This essentially implies that a topological minor, and by extension a minor, can be routed similarly through any brick of a homogeneous flat wall.
The results presented in this subsection are from~\cite{SauST23kapiI,SauST22kapiII}.
First, we need to define boundaried graphs.

\subparagraph{Boundaried graphs.} Let $t\in \bN$. A \emph{$t$-boundaried graph} is a triple ${\bf G}=(G,B,\rho)$ where $G$ is a graph, $B\subseteq V(G)$, $|B|=t$, and $\rho:B\to\bN$ is an injection.
In other words, a $t$-boundaried graph is a graph $G$ with a boundary $B$ of $t$ vertices and a labeling $\rho$ of the vertices of $B$.
We say that two $t$-boundaried graphs ${\bf G}_1=(G_1,B_1,\rho_1)$ and ${\bf G}_2=(G_2,B_2,\rho_2)$ are \emph{isomorphic} if $\rho_1(B_1)=\rho_2(B_2)$ and there is an isomorphism from $G_1$ to $G_2$ that extends the bijection $\rho_2^{-1}\circ\rho_1$.
The triple $(G,B,\rho)$ is a \emph{boundaried graph} if it is a $t$-boundaried graph for some $t\in\bN$.
We denote by $\Bcal^t$ the set of all (pairwise non-isomorphic) $t$-boundaried graphs.
We also set $\Bcal=\bigcup_{t\in\bN}\Bcal^t$.

\subparagraph{Topological minors of boundaried graphs and folios.}
Let $v$ be a vertex of degree two in a graph $G$.
The \emph{dissolution} of $v$ is the operation that consists in deleting $v$ and adding an edge between its two neighbors.
We say that ${\bf H}\in\Bcal$ is a \emph{topological minor} of ${\bf G}=(G,B,\rho)\in\Bcal$ if
$\bf H$ can be obtained from $\bf G$ after a sequence of deletion of edges of $G$ and deletion and dissolution of vertices of $G-B$.
Given $\textbf{G}\in {\cal B}$ and $\ell\in\bN$, we define the \emph{$\ell$-folio} of ${\bf G}$, denoted by $\ell\mbox{\sf -folio}({\bf G})$,
as the set of all boundaried graphs ${\bf H}$ such that $\bf H$ is a topological minor of $\bf G$ of detail at most $\ell$.

\subparagraph{Augmented flaps.}
Let $G$ be a graph, $A$ be a subset of $V(G)$ of size $a$, and $(W,\mathfrak{R})$ be a flatness pair of $G-  A$.
For each flap $F\in \flaps_\mathfrak{R}(W)$ we consider a labeling $\ell_{F}: \partial F\rightarrow\{1,2,3\}$ such that
the set of labels assigned by $\ell_{F}$ to $\partial F$ is one of $\{1\}$, $\{1,2\}$, $\{1,2,3\}$.
We consider a bijection $\rho_{A}: A\to [a]$.
The labelings in ${\cal L}=\{\ell_{F} \mid F\in \flaps_\mathfrak{R}(W)\}$ and the labeling $\rho_A$ will be useful for defining a set of boundaried graphs that we will call augmented flaps.
We first need some more definitions.

Given a flap $F\in\flaps_\mathfrak{R}(W)$, we define an ordering
$\Omega(F)=(x_{1},\ldots,x_{q})$, with $q\leq 3$, of the vertices of $\partial{F}$
so that
\begin{itemize}
	\item $(x_{1},\ldots,x_{q})$ is a counter-clockwise cyclic ordering of the vertices of $\partial F$ as they appear in the corresponding cell of $C(\Gamma)$. Notice that this cyclic ordering is significant only when $|\partial F|=3$,
in the sense that $(x_{1},x_{2},x_{3})$ remains invariant under shifting, i.e., $(x_{1},x_{2},x_{3})$ is the same as $ (x_{2},x_{3},x_{1})$ but not under inversion, i.e., $(x_{1},x_{2},x_{3})$ is not the same as $(x_{3},x_{2},x_{1})$, and
	\item for $i\in[q]$, $\ell_{F}(x_{i})=i$.
\end{itemize}
Notice that the second condition is necessary for completing the definition of the ordering $\Omega(F)$,
and this is the reason why we set up the labelings in ${\cal L}$.\medskip

For each $F\in \flaps_\mathfrak{R}(W)$ with $t_{F}:=|\partial F|$,
we fix $\rho_{F}: \partial F\to [{a}+1,{a}+t_F]$ such that
$(\rho^{-1}_{F}({a}+1),\ldots,\rho^{-1}_{F}({a}+t_F))= \Omega(F)$.
Also, we define the boundaried graph $$\textbf{F}^{{A}}:=(G[{A}\cup F],{A}\cup \partial F,\rho_{{A}}\cup \rho_F)$$
and we denote by $F^{{A}}$ the underlying graph of $\textbf{F}^{{A}}$. We call $\textbf{F}^{{A}}$ an \emph{${A}$-augmented flap} of the flatness pair $(W,\mathfrak{R})$ of $G-  A$
in $G$.

\subparagraph{Palettes and homogeneity.}
For each $\mathfrak{R}$-normal cycle $C$ of $\compass_\mathfrak{R} (W)$, we define $({A},\ell)\mbox{\sf -palette}(C)=\{\ell\mbox{\sf -folio}({\bf F}^{{A}})\mid F\in {\sf influence}_\mathfrak{R}(C)\}$.
We say that the flatness pair $(W,\mathfrak{R})$ of $G-  A$ is \emph{$\ell$-homogeneous with respect to ${A}$} if every {\sl internal} brick of ${W}$ has the {\sl same} $({A},\ell)$\mbox{\sf -palette} (seen as a cycle of $\compass_\mathfrak{R} (W)$).
Given $a\in\bN$ and a graph $G$, let \emph{${\sf ext}_a(G)$} denote the set of all pairs $(G',A)$ such that $A\subseteq V(G')$ has size at most $a$ and $G=G'-A$.
We say that a flatness pair $(W,\mathfrak{R})$ of a graph $G$ is \emph{$(a,\ell)$-homogeneous} if, for each $(G',A)\in{\sf ext}_a(G)$, $(W,\mathfrak{R})$, that is a flatness pair of $G=G'-A$, is $\ell$-homogeneous with respect to $A$.

\medskip
The following observation is a consequence of the fact that, given a wall $W$ and a  subwall $W'$ of $W,$ every internal brick of a tilt $W''$ of $W'$ is also an internal brick of $W.$

\begin{observation}\label{label_convenciones}
	Let $a,\ell\in\mathbb{N},$ $G$ be a graph, and $(W,\mathfrak{R})$  be a flatness pair of $G.$ If $(W,\mathfrak{R})$ is  $(a,\ell)$-homogeneous, then for every subwall $W'$ of $W,$ every $W'$-tilt of $(W,\mathfrak{R})$ is also $(a,\ell)$-homogeneous.
\end{observation}

\subsection{Tight renditions}\label{subsec_tight}

A tight rendition is a $\Omega$-rendition with a few more properties, the main one being that there are $|\tilde{c}|$ disjoint paths from each cell $c$ to $V(\Omega)$.

\subparagraph{Tight renditions.}
We say that an $\Omega$-rendition $(\Gamma,\sigma,\pi)$ is \emph{tight} if the following conditions are satisfied:
\begin{enumerate}
\item if there are two points $x,y$ of $N(\Gamma)$ such that $e=\{\pi(x),\pi(y)\}\in E(G)$, then there is a cell $c\in C(\Gamma)$ such that
$\sigma(c)$ is the two-vertex connected graph $(e,\{e\})$,
\item for every $c\in C(\Gamma)$, every two vertices in $\pi(\tilde{c})$ belong to some path of $\sigma(c)$,
\item for every $c\in C(\Gamma)$ and every connected component $C$ of the graph $\sigma(c)- \pi(\tilde{c})$,
if $N_{\sigma(c)}(V(C))\ne\emptyset$, then $N_{\sigma(c)}(V(C))=\pi(\tilde{c})$,
\item there are no two distinct non-trivial cells $c_1$ and $c_2$ such that $\pi(\tilde{c}_1)=\pi(\tilde{c}_2)$, and
\item for every $c\in C(\Gamma)$, there are $|\tilde{c}|$ vertex-disjoint paths in $G$ from $\pi(\tilde{c})$ to the set $V(\Omega)$.
\end{enumerate}
We say that a flatness pair is \emph{tight} if the underlying rendition is tight.

\bigskip
Because of the next result, renditions are often implicitly assumed to be tight in papers such as~\cite{BasteST23hittIV,SauST24amor}.

\begin{proposition}[\!\!\cite{SauST24amor}]\label{prop_tight}
There is a linear-time algorithm that, given a graph $G$ and an $\Omega$-rendition $(\Gamma,\sigma,\pi)$ of $G$,
outputs a tight $\Omega$-rendition of $G$.
\end{proposition}

\subparagraph{Irrelevant sets.}
Let $G$ be a graph and let $\ell\in \mathbb{N}.$
We say that a vertex set $X\subseteq V(G)$ is
\emph{$\ell$-irrelevant} if every graph $H$ with detail at most $\ell$ that is a minor of $G$ is also a minor of $G- X.$ \medskip

\subparagraph{Linkages.}
A \emph{linkage} $L$ of \emph{order} $k$ in a graph $G$ is the union of a collection of $k$ pairwise vertex-disjoint paths of $G$. The set of pairs
of vertices corresponding to the endpoints of these paths is the \emph{pattern} of $L$.
The Unique Linkage Theorem, proven in \cite{RobertsonS09XXI,RobertsonS12XXII} and also \cite{KawarabayashiW10asho}, asserts that there is a function $f_{\sf ul}$
such that if  $L$  is a linkage of pattern $\Pcal$ of order $k$ in a graph $G$ with $V(G) = V(L)$ and $L$ is unique with pattern $\Pcal$, then the treewidth of $G$ is at most $f_{\sf ul}(k)$. The linkage function appears in the general dependency of several results related to the application of the irrelevant vertex technique (see~\cite{AdlerKKLST17irre,BasteST23hittIV,FominGSST25comp,GolovachST23comb,GolovachST23mode,SauST24amor,SauST23kapiI}).

We state the following result from \cite{BasteST23hittIV}.
In fact, \cite[Theorem 5.9]{BasteST23hittIV} is stated for boundaried graphs. \autoref{label_panlatinismo} is derived by the same proof
if we consider graphs with empty boundary.

\begin{proposition}\label{label_panlatinismo}
There exist two functions
$\newfun{label_conversational}: \mathbb{N}^3\to\mathbb{N}$ and
$\newfun{label_unbelievability}: \mathbb{N}^2\to\mathbb{N},$
where the images of $\funref{label_conversational}$ are odd numbers,
such that the following holds.

Let
$a,\ell\in\mathbb{N},$ $q\in\mathbb{N}_{\geq 3}$ be an odd integer, and $G$ be a graph.
Let $A$ be a subset of $V(G)$ of size at most $a$
and  $(W,\mathfrak{R})$ be  a regular tight
flatness pair 
of $G- A$ of height at least $\funref{label_conversational}(a,\ell,q)$ that is $\funref{label_unbelievability}(a,\ell)$-homogeneous with respect to $A.$

Then the vertex set of the compass of every $W^{(q)}$-tilt of $(W,\mathfrak{R})$ is $\ell$-irrelevant.

Moreover, it holds that $\funref{label_conversational}(a,\ell,q)=\mathcal{O}( (f_\mathsf{ul}(16a+12\ell))^3   + q)$
and $\funref{label_unbelievability}(a,\ell)=a+\ell+3,$ where $f_\mathsf{ul}$ is the function of the Unique Linkage Theorem.
\end{proposition}

The following result states that a tight and homogeneous flatness pair can be found inside any big enough flatness pair.
Actually, the result was stated in \cite{SauST22kapiII} (and previously in \cite{SauST24amor}) for $\ell$-homogeneity with respect to every subset of $A$, but the proofs all work the same way for the more general case of $(a,\ell)$-homogeneity.

\begin{proposition}[\!\!\cite{SauST22kapiII}]\label{@disreputable}
There is a function $\newfun{@philistines}:\bN^4\to \bN$, whose images are odd integers, and an algorithm with the following specifications:\medskip

	\noindent{\tt Homogeneous}$(r,a,\ell,t,G,W,\cal R)$\\
	\noindent{\textbf{Input}:} Integers $r\in\bN_{\geq 3}$, $a,\ell,t\in \bN$, a graph $G$, and a flatness pair $(W,\mathfrak{R})$ of $G$ of height $\funref{@philistines}(r,a,\ell)$ whose $\mathfrak{R}$-compass has treewidth at most $t$.\\
	\noindent{\textbf{Output}:} A flatness pair $(\breve{W},\breve{\mathfrak{R}})$ of $G$ of height $r$ that is tight, $(a,\ell)$-homogeneous, and is a $W'$-tilt of $(W,\mathfrak{R})$ for some subwall $W'$ of $W$.\\
Moreover, $\funref{@philistines}(r,a,\ell) = \Ocal(r^{\newfun{@withdrawing}(a,\ell)})$ where $\funref{@withdrawing}(a,\ell)=2^{2^{\Ocal((a+\ell)\cdot\log(a+\ell))}}$ and the algorithm runs in time $2^{\Ocal(\funref{@withdrawing}(a,\ell)\cdot r \log r+t\log t)}\cdot(n+m)$.
\end{proposition}

The size of the flatness pair $(W,\frak{R})$ necessary to find a homogeneous flatness pair in \autoref{@disreputable} is very large
and is the main cause of the huge degree of the polynomial in $k$ in the running time of \autoref{th}.

However, in the bounded genus case, we can find a homogeneous wall inside a flatness pair of smaller size if we additionally ask that its compass is embeddable in a disk and that no ``leaf-block'' of the graph is planar.
A \emph{leaf-block} in a graph $G$ is either a connected component $C$ of $G$, or the graph $G[V(C)\cup\{v\}]$ for some vertex $v\in V(G)$ and some connected component $C$ of $G-v$.
Given a leaf-block $B$, with denote by $V_B$ the set of $V(C)$.

\begin{lemma}\label{@disreputablepl}
There exists an algorithm with the following specifications:\medskip

	\noindent{\tt Planar-Homogeneous}$(G,W,\cal R)$\\
	\noindent{\textbf{Input}:} A graph $G$ whose leaf-blocks are not planar and a flatness pair $(W,\mathfrak{R}=(X,Y,P,C,\Gamma,\sigma,\pi))$ of $G$ whose $\mathfrak{R}$-compass is embeddable in a disk with $X\cap Y$ on its boundary.\\
	\noindent{\textbf{Output}:} A 7-tuple $\frak{R}'$ such that $(W,\mathfrak{R}')$ is a flatness pair of $G$ that is regular, tight, and $\ell$-homogeneous with respect to $\emptyset$ for any $\ell\in\bN$.\\
Moreover, the algorithm runs in time $\Ocal(n+m)$.
\end{lemma}

\begin{proof}
Let $\Omega$ be the cyclic ordering of the vertices of $X\cap Y$ as they appear in $D(W)$.
Given that $G[Y]$ is embeddable in a disk with $X\cap Y$ on its boundary, there is a $\Omega$-rendition  $(\Gamma',\sigma',\pi')$ of $G[Y]$ such that $\pi'(N(\Gamma'))=V(G)$ and each cell of $\Gamma'$ contains exactly one edge of $G$, i.e., for each $c\in C(\Gamma')$, there is $e\in E(G)$ such that $\sigma'(c)=(e,\{e\})$.
Note that, for each $c\in C(\Gamma')$, $|\tilde{c}|=2$.

Let us transform $(\Gamma',\sigma',\pi')$ into a tight rendition of $G[Y]$.
Note that the only item that is not easily verified is item 5 of the definition of a tight rendition.
That is, there might be a $c\in C(\Gamma')$ such that there are strictly less than $|\tilde{c}|=2$ vertex-disjoint paths in $G$ from $\pi(\tilde{c})$ to $V(\Omega)$.
Let $c$ be such a cell.
If there is no path in $G$ from $\pi(\tilde{c})$ to $V(\Omega)$, then $\sigma(c)$ belongs to a connected component $C$ of $G$ that does not contain vertices of $V(\Omega)$.
Therefore, $C$ is a planar leaf-block of $G$, a contradiction.
Otherwise, there is $v\in V(G)\setminus \pi(\tilde{c})$ such that every path from $\pi(\tilde{c})$ to $V(\Omega)$ contains $v$, so $v$ is a cut vertex of $G$.
Therefore, $\sigma(c)$ belongs to a connected component $C$ of $G-v$ not containing vertices of $V(\Omega)$, and thus a planar leaf-block of $G$, again a contradiction.
Therefore, $(\Gamma',\sigma',\pi')$ is a tight rendition.

We set $\Rcal':=(X,Y,P,C,\Gamma',\sigma',\pi')$.
Given that $(W,\Rcal)$ is a flatness pair of $G$ and that none of the cells of $\Gamma'$ is $W$-external, $W$-marginal, or untidy,
we conclude that $(W,\Rcal')$ is also a regular flatness pair of $G$.

Finally, given that each cell of $\Gamma'$ contains exactly one edge and has a boundary of size two, we conclude that, for any $\ell\in\bN$, the $\ell$-folio of $\bf F$ is the same for each ${\bf F}\in{\sf Flaps}_\Rcal(W)$.
Therefore, $(W,\Rcal')$ is $\ell$-homogeneous with respect to $\emptyset$, hence the result.
\end{proof}

\subsection{An auxiliary lemma}\label{subsec_aux}

The following lemma says that given a big enough flat wall $W$ and a vertex set $S$ of size at most $k$, there is a smaller flat wall $W^*$ such that the central vertices of $W$ and the intersection of $S$ with the compass of $W^*$ are contained in the compass of the  central wall of height five of $W^*$.
This result is used in \cite{SauST23kapiI}, but is not stated as a stand-alone result, so we reprove it here for completeness.

\begin{lemma}\label{claim_irr}
There exists a function $\newfun{label_irr}: \mathbb{N}^{3}\to \mathbb{N},$
whose images are odd integers, such that the following holds.

Let
$a,d,k,q,z\in\mathbb{N}$, with odd $q\geq 3$ and odd $z\ge 5$,
$G$ be a graph,
$S\subseteq V(G)$, where $|S|\le k$,
$(W,\mathfrak{R})$ be a regular and tight flatness pair of $G$ of height at least $\funref{label_irr}(k,z,q)$ that is $(a,d)$-homogeneous,
and $(W',\mathfrak{R}')$ be a $W^{(q)}$-tilt of $(W,\mathfrak{R})$.
Then, there is a flatness pair $(W^*,\frak{R}^*)$ of $G$ such that:
\begin{itemize}
\item $(W^*,\frak{R}^*)$ is a $\tilde{W}'$-tilt of $(W,\frak{R})$ for some $z$-subwall $\tilde{W}'$ of $W,$
\item $(W^*,\frak{R}^*)$ is regular, tight, and $(a,d)$-homogeneous, and
\item $V(\textsf{Compass}_{\mathfrak{R}'}(W'))$ and $S\cap V(\textsf{Compass}_{\mathfrak{R}^*}(W^*))$ are both subsets of the vertex set of the compass of every $W^{*(5)}$-tilt of $(W^*,\frak{R}^*)$.
\end{itemize}
Moreover, $\funref{label_irr}(k,z,q)=\odd((k+1)\cdot (z+1)+q).$
\end{lemma}

\begin{proof}
Let $r:=\funref{label_irr}(k,z,q)$.
For every $i\in[r],$ we denote by $P_{i}$ (resp. $Q_{i}$) the $i$-th vertical (resp. horizontal) path of $W.$ Let $z':=\frac{z+1}{2}$
and observe that, since $z$ is odd, we have $z'\in \mathbb{N}.$
We also define, for every $i\in[k+1]$ the graph
\[B_{i}:=\bigcup_{j\in [z'-1]}P_{f_{z'}(i,j)} \cup\bigcup_{j\in [z']} P_{r+1-f_{z'}(i,j)}\cup \bigcup_{j\in [z'-1]}Q_{f_{z'}(i,j)} \cup \bigcup_{j\in [z']} Q_{r+1-f_{z'}(i,j)},\]
where $f_{z'}(i,j):=j+(i-1)\cdot (z'+1)$.
For every $i\in[k+1],$ we define ${W}_{i}$ to be the graph obtained from $B_{i}$
after repeatedly removing from $B_{i}$ all vertices of degree one (see \autoref{label_desasosegado} for an example).
Since $z=2z'-1,$ for every $i\in[k+1]$, ${W}_{i}$ is a $z$-subwall of $W.$
For every $i\in[k+1],$  we set $L^{i}_\mathsf{in}$ to be the inner layer of $W_{i}.$
Notice that $L^{i}_\mathsf{in},$ for $i\in[k+1],$ and $D(W^{(q)})$ are $\mathfrak{R}$-normal cycles of $\mathsf{Compass}_{\mathfrak{R}}(W).$

\begin{figure}[ht]
\centering
\scalebox{1}[1]{\includegraphics[width=11cm]{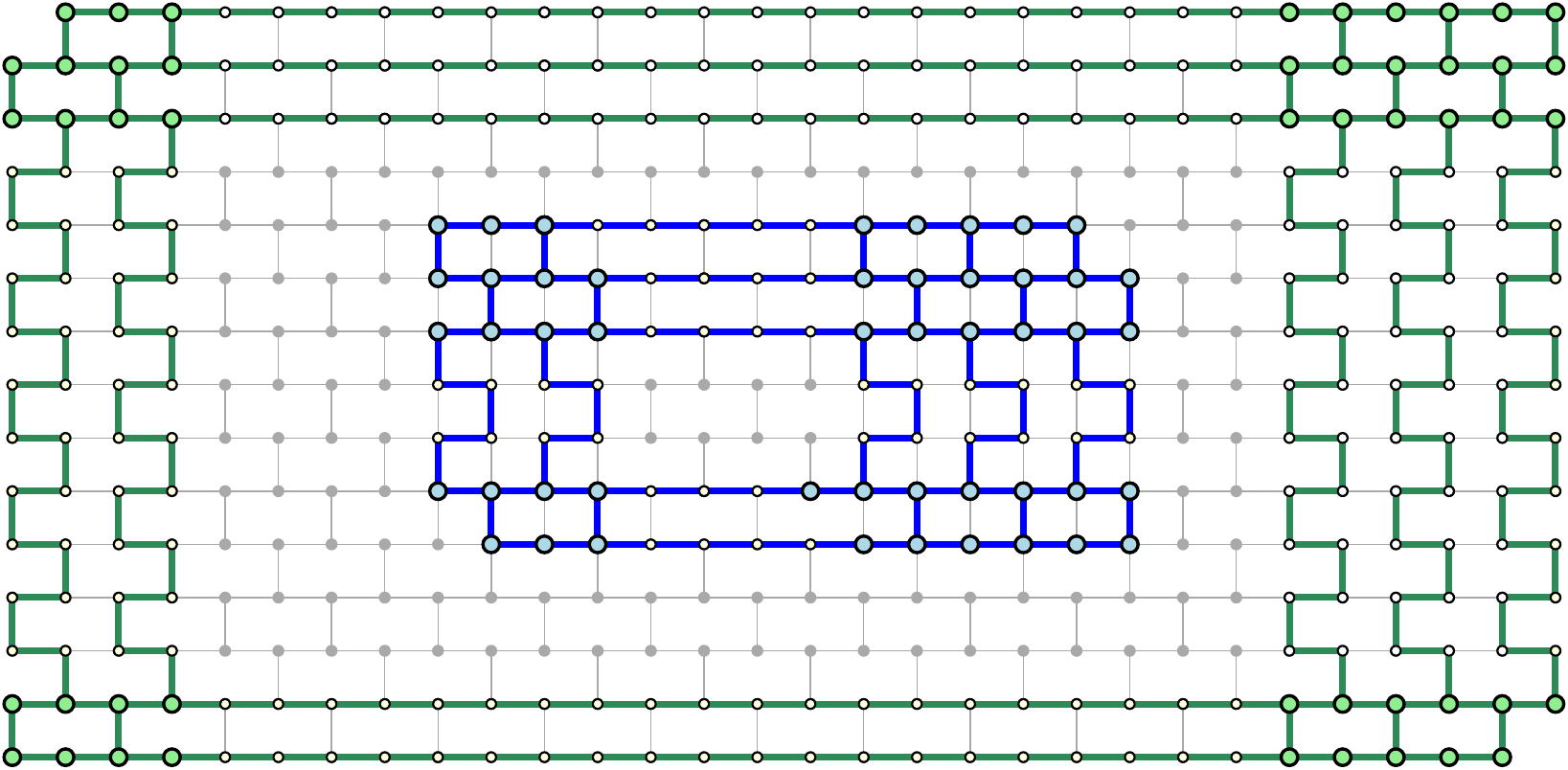}}
\caption{A 15-wall and the 5-walls $W_1$ and $W_2$ as in the proof of  \autoref{lemma_biiid_irr}, depicted in green and blue, respectively. The white vertices are subdivision vertices of the walls $W_1$ and $W_2$.
This figure is adapted from \cite[Figure 3]{SauST23kapiI}.}
\label{label_desasosegado}
\end{figure}

\smallskip
By definition of a tilt of a flatness pair, it holds that $V(\textsf{Compass}_{\mathfrak{R}'}(W'))\subseteq V(\cupall\mathsf{Influence}_{\mathfrak{R}}(W^{(q)}))$.
Moreover, for every $i\in[k+1],$ the fact that $r\geq (k+1)\cdot (z+2)+q$ implies that
$\cupall\mathsf{Influence}_{\mathfrak{R}} (W^{(q)})$ is a subgraph of $\cupall\mathsf{Influence}_{\mathfrak{R}} (L_{\mathsf{in}}^{i}).$
Hence, for every $i\in[k+1],$ we have that
$V(\textsf{Compass}_{\mathfrak{R}'}(W'))\subseteq V(\cupall\mathsf{Influence}_{\mathfrak{R}}(L_{\mathsf{in}}^{i})).$

\smallskip
For every $i\in[k+1],$ let $(W_{i}',\mathfrak{R}_{i})$ be a flatness pair of $G$ that is a $W_{i}$-tilt of $(W,\mathfrak{R})$ ({which exists due to \autoref{@expurgated}}).
Also, note that, for every $i\in[k+1],$  $L_{\mathsf{in}}^i$ is the inner layer of $W_{i}$ and therefore it is
an $\mathfrak{R}_i$-normal cycle of $\mathsf{Compass}_{\mathfrak{R}_i}(W_i').$
Additionally, for every $i\in[k+1],$ $(W_{i}',\mathfrak{R}_{i})$ is
$(a,d)$-homogeneous due to \autoref{label_convenciones}, and, due to \autoref{label_expressionism}, $(W_{i},\mathfrak{R}_{i})$ is also regular.
Also, by \autoref{prop_tight}, we can assume $(W_{i}',\mathfrak{R}_{i})$ to be tight.

\smallskip
For every $i\in[k+1],$ we set
$D_{i} := V(\textsf{Compass}_{\mathfrak{R}_i }(W_i'))\setminus V(\cupall\mathsf{Influence}_{\mathfrak{R}_i }(L_\mathsf{in}^{i}))$.
Given that the vertices of $V(W_i)$ are contained between the $((i-1)\cdot(z'+1)+1)$-th and the $(i\cdot(z'+1)-1)$-th layers of $W$ for $i\in[k+1]$, it implies that the vertex sets $D_i,$ $i\in[k+1]$, are pairwise disjoint.
Therefore, since that $|S|\le k$, there exists a $j\in [k+1]$ such that
$S\cap D_j=\emptyset.$
Thus, $S\cap V(\textsf{Compass}_{\mathfrak{R}_j}(W_j'))\subseteq V(\cupall\mathsf{Influence}_{\mathfrak{R}_j}(L_\mathsf{in}^j))$.

\smallskip
Let $Y$ be the vertex set of the compass of some $W_j^{(5)}$-tilt
of $(W_j',\mathfrak{R}_j)$.
Note that $L_{\sf in}^j$ is the perimeter of $W_j^{(3)}$, and therefore, we have
$S\cap V(\textsf{Compass}_{\mathfrak{R}_j}(W_j'))\subseteq V(\cupall\mathsf{Influence}_{\mathfrak{R}_j}(L_\mathsf{in}^j))\subseteq Y$ and $V(\textsf{Compass}_{\mathfrak{R}'}(W'))\subseteq V(\cupall\mathsf{Influence}_{\mathfrak{R}}(L_\mathsf{in}^j))\subseteq Y$.
Therefore, $(W_{j}',\mathfrak{R}_{j})$ is the desired flatness pair.
\end{proof}

\subsection{Finding an irrelevant vertex in a homogeneous flat wall}
\label{subsec_irr}

The next lemma states that the central vertex of a big enough homogeneous flat wall is irrelevant.
As mentioned previously, \autoref{lemma_biiid_irr} takes inspiration from  \cite[Lemma 16]{SauST23kapiI}, though the proof is more involved in \autoref{lemma_biiid_irr}, due to the more general modifications allowed and the annotation.

\begin{lemma}\label{lemma_biiid_irr}
Let $\Fcal$ be a finite collection of graphs
and $\Lcal$ be a hereditary R-action.
There exists a function $\newfun{label_interpersonal}: \mathbb{N}^{4}\to \mathbb{N},$
whose images are odd integers, such that the following holds.

Let
$k, q, a\in\mathbb{N}$, with odd $q\geq 3$.
Let $G$ be a graph,
$S'\subseteq V(G)$ be a set of size at most $k$,
$(H_2',\phi')\in\Lcal(G[S'])$, and $G':=G_{(H_2',\phi')}^{S'}$.
Let $A\subseteq V(G')$ be a subset of size at most $a$,
$(W,\mathfrak{R})$ be a regular and tight flatness pair of $G'- A$ of height at least $\funref{label_interpersonal}(a,\ell_{\mathcal{F}},q,k)$ that is $(a,\funref{label_unbelievability}(a,\ell_{\mathcal{F}}))$-homogeneous and such that $\phi'(S')\cap V(\textsf{Compass}_{\mathfrak{R}}(W))=\emptyset$.
Let $(W',\mathfrak{R}')$ be a $W^{(q)}$-tilt of $(W,\mathfrak{R})$ and $Y:=V(\textsf{Compass}_{\mathfrak{R}'}(W'))$.

Then,
$(G,S',H_2',\phi',k)$ and
$(G-Y,S',H_2',\phi',k)$ are equivalent instances of \apb.

Moreover, $\funref{label_interpersonal}(a,\ell_{\mathcal{F}},q,k)=\funref{label_irr}(k, \funref{label_conversational}(a,\ell_{\mathcal{F}},5),q ).$
\end{lemma}

\begin{proof}
Let $z:=\funref{label_conversational}(a,\ell_{\mathcal{F}},5),$ $d:= \funref{label_unbelievability}(a,\ell_{\mathcal{F}}),$ and $r=:\funref{label_interpersonal}(a,\ell_{\mathcal{F}},q,k)=\funref{label_irr}(k,z,q)$.

\smallskip
The forward direction is immediate given that $\Lcal$ is hereditary.
Indeed, suppose $(S,H_2,\phi)$ is a solution of \apb for the instance $(G,S',H_2',\phi',k)$.
Let $S^*=S\setminus Y\subseteq V(G)\setminus Y.$
Then, because $\Lcal$ is hereditary,
the restriction of $(H_2,\phi)$ to $S^*$ is in $\Lcal((G-Y)[S^*])$.
Given that $Y\subseteq V(\textsf{Compass}_{\mathfrak{R}}(W))\subseteq V(G')\setminus \phi^+(S')= V(G)\setminus S'$, it follows that
$S'\cap Y=\emptyset$, and thus that $S'\subseteq S^*$.
Therefore,
the restriction of $(H_2^*,\phi^*)$ to $S'$ is $(H_2',\phi')$.
Moreover, $G^*:=(G-Y)^{S^*}_{(H_2^*,\phi^*)}=G^S_{(H_2,\phi)}-Y$, so $G^*\in\exc(\mathcal{F})$.
We conclude that $(S^*,H_2^*,\phi^*)$ is a solution of $(G-Y,S',H_2',\phi',k)$.

\medskip
Suppose now that $(S,H_2,\phi)$ is a solution of \apb for the instance $(G-Y,S',H_2',\phi',k)$.
Let us now prove that
there is $S^*\subseteq S$ such that
$(S^*,H_2[\phi^+(S^*)],\phi|_{S^*})$ is a solution of \apb for the instance $(G,S',H_2',\phi',k)$.

\smallskip
By \autoref{claim_irr}, there exists a regular and tight flatness pair $(W^*,\frak{R}^*)$ of $G'-A$ of height $z$ that is $d$-homogeneous with respect to $2^A$ such that $Y$ and $S\cap Y^*$ are both subsets of the vertex set of the compass of every $W^{*(5)}$-tilt of $(W^*,\frak{R}^*)$, where $Y^*=V(\textsf{Compass}_{\mathfrak{R}^*}(W^*))$.
Additionally, $Y\subseteq V(\textsf{Compass}_{\mathfrak{R}}(W))$, so $S'\cap Y^*=\emptyset$.
Let $S^*=S\setminus Y^*\supseteq S'$.
Given that $\Lcal$ is hereditary,
the restriction $(H_2^*,\phi^*)$ of $(H_2,\phi)$ to $S^*$ is in $\Lcal(G[S^*])$.
Moreover,
the restriction of $(H_2^*,\phi^*)$ to $S'$ is $(H_2',\phi')$.
It thus remains to prove that $G^*:=G^{S^*}_{(H_2^*,\phi^*)}\in\exc(\mathcal{F})$.

\smallskip
Let $Y'$ be the compass of some $W^{*(5)}$-tilt of $(W^*,\frak{R}^*)$.
Hence we have $Y,S\cap Y^*\subseteq Y'$ and $S^*\cap Y'\ne\emptyset$.
Therefore, we have $G^*-Y'=(G-Y)^S_{(H_2,\phi)}-Y'$.
Given that $(G-Y)^S_{(H_2,\phi)}\in\exc(\mathcal{F})$ and that $\exc(\mathcal{F})$ is minor-closed, it implies that $G^*-Y'\in\exc(\mathcal{F})$.

\begin{claim}\label{claim_irr2}
$Y'$ is $\ell_\Fcal$-irrelevant in $G^*$.
\end{claim}

\begin{cproof}
Let $S_r=S^*\setminus S'$.
We set $A^*:=A\setminus S_r\cup\phi^+(A\cap S_r)$.
If $\mathfrak{R}^* = (X^*,Y^*,P,C,\Gamma,\sigma,\pi),$
then we set $\mathfrak{R}^*_\Lcal = (X^*_\Lcal,Y^*,P,C,\Gamma,\sigma,\pi),$ where $X^*_\Lcal=(X^*\setminus S_r)\cup\phi^+(S_r)\setminus A^*$.
Given that $S_r\cap Y^*=\emptyset,$ it implies that
$(W^*,\mathfrak{R}^*_\Lcal)$ is a flatness pair of $G^*- A^*$.
Notice that the $\mathfrak{R}^*$-compass and the $\mathfrak{R}^*_\Lcal$-compass of $W^*$ are identical,
which implies that $(W^*,\mathfrak{R}^*_\Lcal)$ is a regular
flatness pair of $G^*- A^*$ that is $(a,d)$-homogeneous.
Given that $|A^*|\le|A|\le a$, it implies in particular that $(W^*,\mathfrak{R}^*_\Lcal)$ is a regular
flatness pair of $G^*- A^*$ that is $d$-homogeneous with respect to $A^*$.

\smallskip
Given that $z=\funref{label_conversational}(a,\ell_{\mathcal{F}},5)$ and $d:= \funref{label_unbelievability}(a,\ell_{\mathcal{F}})$,
we can thus apply \autoref{label_panlatinismo} with input $(a, \ell_\Fcal, 5 ,$ $G^*, A^*, (W^*,\mathfrak{R}^*_\Lcal))$
which implies that the vertex set of the compass of every $W^{* (5)}$-tilt
of $(W^*,\mathfrak{R}^*_\Lcal)$
 is $\ell_{\mathcal{F}}$-irrelevant.
This is in particular the case of $Y'$, given that the $\mathfrak{R}^*$-compass and the $\mathfrak{R}^*_\Lcal$-compass of $W^*$ are identical.
\end{cproof}

Given that $Y'$ is $\ell_\Fcal$-irrelevant in $G^*$ by \autoref{claim_irr2} and that every graph in $\Fcal$ has detail at most $\ell_\Fcal$, it implies that $G^*\in\exc(\mathcal{F})$ if and only if $G^*-Y'\in\exc(\mathcal{F})$, hence the result.
\end{proof}

After combining \autoref{@disreputable} and \autoref{lemma_biiid_irr}, we finally get our algorithm to find an irrelevant vertex inside a flat wall in the general case.

\begin{proof}[Proof of \autoref{th_irr}]
Let $r=\funref{label_interpersonal}(a,\ell_{\mathcal{F}},3,k)$, $\ell=\funref{label_unbelievability}(a,\ell_\Fcal)$, and $\funref{something}(k,a)=\funref{@philistines}(r,a,\ell)=\Ocal_{a,\ell_\Fcal}(k^c)$.

We apply the algorithm {\tt Homogeneous} of \autoref{@disreputable} with input $(r,a,\ell,t,G'-A,W,\Rcal)$.
It outputs in time $2^{\Ocal(\funref{@withdrawing}(a,\ell)\cdot r \log r+t\log t)}\cdot(n+m)$ a flatness pair $(\breve{W},\breve{\mathfrak{R}})$ of $G'-A$ of height $r$ that is tight, $(a,\ell)$-homogeneous, and is a $W'$-tilt of $(W,\frak{R})$ for some subwall $W'$ of $W$.
Given that $(W,\frak{R})$ is a regular flatness pair, by \autoref{label_expressionism}, so is $(\breve{W},\breve{\mathfrak{R}})$.
Given that $r=\funref{label_interpersonal}(a,\ell_{\mathcal{F}},3,k)$, that $\ell=\funref{label_unbelievability}(a,\ell_\Fcal)$, and that $\phi'(S')$ does not intersect the $\frak{R}$-compass of $W$, and thus neither the $\breve{\mathfrak{R}}$-compass of $\breve{W}$,
we conclude by \autoref{lemma_biiid_irr}, that for any $W^{(3)}$-tilt $(W',\frak{R}')$ of $(\breve{W},\breve{\mathfrak{R}})$,
$(G,S',H_2',\phi',k)$ and $(G-Y,S',H_2',\phi',k)$ are equivalent instances of \apb, where $Y:=V(\compass_{\frak{R}'}(W'))$.
Let $v$ be a central vertex of $\breve{W}$.
Given that $v\in Y$, $(G,S',H_2',\phi',k)$ and $(G-v,S',H_2',\phi',k)$ are in particular equivalent instances of \apb.
Hence the result.
\end{proof}

\subsection{Irrelevant vertex in the bounded genus case}
\label{subsec_pl}

The next lemma essentially states that planar leaf-blocks are irrelevant in the bounded genus case.

\begin{lemma}\label{lem_irr_pl2}
Let $\Lcal$ be a hereditary R-action  and $\Fcal$ be the collection of obstructions of the graphs embeddable in a surface $\Sigma$ of genus at most $g$.
Let $G$ be a graph and $k\in\bN$.
Let $S'\subseteq V(G)$ be a set of size at most $k$ and $(H_2',\phi')\in\Lcal(G[S'])$.
Suppose that there is a planar leaf-block $B$ of $G':=G_{(H_2',\phi')}^{S'}$
such that $V_B\cap\phi'(S')=\emptyset$.
Then, $(G,S',H_2',\phi',k)$ and $(G-V_B,S',H_2',\phi',k)$ are equivalent instances of {\sc $\Lcal$-AR-$\exc(\Fcal)$}.
\end{lemma}

\begin{proof}
Suppose that there is a solution $(S,H_2,\phi)$ of {\sc $\Lcal$-AR-$\exc(\Fcal)$} for  the instance $(G,S',H_2',\phi',k)$.
Let $S^*:=S\setminus V_B$. Note that $S'\subseteq S^*$.
Let $(H_2^*,\phi^*)$ be the restriction of $(H_2,\phi)$ to $S^*$.
Given that $\Lcal$ is hereditary, $(H_2^*,\phi^*)\in\Lcal(G[S^*])$.
Given that $(G-V_B)_{(H_2^*,\phi^*)}^{S^*}=G_{(H_2,\phi)}^{S}-V_B$ is a subgraph of $G_{(H_2,\phi)}^{S}\in\exc(\Fcal)$ and that $\exc(\Fcal)$ is hereditary, it implies that $(G-V_B)_{(H_2^*,\phi^*)}^{S^*}\in\exc(\Fcal)$.
So $(S^*,H_2^*,\phi^*)$ is a solution of {\sc $\Lcal$-AR-$\exc(\Fcal)$} for  the instance $(G-V_B,S',H_2',\phi',k)$.

Suppose now that there is a solution $(S,H_2,\phi)$ of {\sc $\Lcal$-AR-$\exc(\Fcal)$} for  the instance $(G-V_B,S',H_2',\phi',k)$.
Let $G'':=(G-V_B)_{(H_2,\phi)}^S$ and $G^*:=G_{(H_2,\phi)}^S$.
Note that $G^*$ is obtained by taking the disjoint union of $G''$ and $B$ and identifying at most one vertex of both sides (that is, the vertex $v\in V(B)\setminus V_B$, if it exists).
Given that $G''$ is embeddable in $\Sigma$  and that $B$ is planar, we conclude that $G^*$ is embeddable in $\Sigma$  as well.
Therefore, $(S,H_2,\phi)$ is a solution of {\sc $\Lcal$-AR-$\exc(\Fcal)$} for  the instance $(G,S',H_2',\phi',k)$.
\end{proof}

After combining \autoref{@disreputablepl}, \autoref{lemma_biiid_irr}, and \autoref{lem_irr_pl2}, we finally get our algorithm to find an irrelevant vertex inside a flat wall in the bounded genus case.

\begin{proof}[Proof of \autoref{th_irr_pl}]
Let $r=\funref{else}(k):=\funref{label_interpersonal}(0,\ell_{\mathcal{F}},3,k)$ and $\ell=\funref{label_unbelievability}(a,\ell_\Fcal)$.

We can find all the cut vertices of $G'$ using a depth-first search algorithm in time $\Ocal(n+m)$.
Therefore, if there is a planar leaf-block $B$ in $G'$, then we can find it in time $\Ocal(n+m)$.
In that case, we can return $V_B$ by \autoref{lem_irr_pl2}.

Otherwise, we apply the algorithm {\tt Planar-Homogeneous} of \autoref{@disreputablepl} with input $(G',W,\frak{R})$, which outputs a 7-tuple $\frak{R}'$ such that $(W,\frak{R}')$ is a flatness pair of $G'$ that is regular, tight, and $(0,\ell)$-homogeneous.
Given that $r=\funref{label_interpersonal}(0,\ell_{\mathcal{F}},3,k)$, that $\ell=\funref{label_unbelievability}(a,\ell_\Fcal)$, and that $\phi'(S')$ does not intersect the $\frak{R}$-compass of $W$, which is also the ${\mathfrak{R}'}$-compass of ${W}$,
we conclude by \autoref{lemma_biiid_irr} that for any $W^{(3)}$-tilt $(W^*,\frak{R}^*)$ of $({W},{\mathfrak{R}'})$,
$(G,S',H_2',\phi',k)$ and $(G-Y,S',H_2',\phi',k)$ are equivalent instances of \apb, where $Y:=V(\compass_{\frak{R}^*}(W^*))$.
Let $v$ be a central vertex of ${W^*}$.
Given that $v\in Y$, $(G,S',H_2',\phi',k)$ and $(G-v,S',H_2',\phi',k)$ are in particular equivalent instances of \apb.
So we can return $Y:=\{v\}$, hence the result.
\end{proof}

\vspace{-1mm}

\section{Obligatory vertex}\label{sec_obl}

This section is dedicated to proving \autoref{lem_obl}.
The proof is quite similar to \cite[Lemma 13]{SauST23kapiI}, though the notation are more involved.
However we require, and thus prove a stronger result, that is that, not only we find a set $A$ containing a vertex in the solution, but the size of this set $A$ must decrease after the modification.
We prove in \autoref{label_beschaffenheit} that if $G$ contains a complete $A$-apex grid as an $A$-fixed minor (see below for the definitions), then $A$ intersects any solution $S$.
More specifically, after the modification of $G$ restricted to $A$ is done, the size of $G$ must decrease.
We then derive \autoref{lem_obl} from \autoref{label_beschaffenheit}, that is merely a translation that helps us in our setting of walls to easily find such an $A$-fixed minor.

\subparagraph{Central grids.}
Let  $k,r\in\mathbb{N}_{\geq 2}.$
We define the \emph{perimeter} of a $(k\times r)$-grid to be the unique cycle of the grid of length at least three that does not contain vertices of degree four.
We shorten the notation $(r\times r)$-grid as an \emph{$r$-grid}.

Let $r\in \mathbb{N}_{\geq 2}$ and $\Gamma$ be an $r$-grid.
Given an $i\in\lceil \frac{r}{2}\rceil,$ we define the \emph{$i$-th layer} of $\Gamma$ recursively as follows.
The first layer of $\Gamma$ is its perimeter, while, if $i\geq 2,$ the $i$-th layer of $\Gamma$ is
the $(i-1)$-th layer of the grid created if we remove from $\Gamma$ its perimeter.
Given two  odd integers $q,r\in\mathbb{N}_{\geq 3}$ such that $q\leq r$ and an $r$-grid $\Gamma,$
we define the \emph{central $q$-grid} of $\Gamma$ to be the graph obtained from $\Gamma$
if we remove from $\Gamma$ its $\frac{r-q}{2}$ first layers.

\subparagraph{Apex grids.}
Let $H$ be a graph, $A\subseteq V(H)$, and $r\in\bN$.
$H$ is an \emph{$A$-apex $r$-grid} if $H-A$ is a $r$-grid.
$H$ is a \emph{complete $A$-apex $r$-grid} if it is a $A$-apex $r$-grid and that there is an edge between each vertex of $A$ and each vertex of $H-A$.

\subparagraph{Fixed minors.}
Given a graph $G$ and a set $A\subseteq V(G),$ we say
that a graph $H$ is a \emph{$A$-fixed minor} of $G$ if $H$ can be obtained from a subgraph $G'$ of $G$ where $A\subseteq V(G')$ after contracting edges without endpoints  in $A.$
For example, the graph of \autoref{fig_obl_vtx} contains an $A$-apex 3-grid as an $A$-fixed minor.
\medskip

The following result says that a complete $A$-apex grid is always an $A$-fixed minor of a big enough $A$-apex grid $\Gamma$ such that each vertex of $A$ has sufficiently many neighbors in the central part of $\Gamma$.
\begin{proposition}[\!\!\cite{SauST23kapiI}]\label{label_emporteroient}
	There exist three functions $\newfun{label_impercepbbly}, \newfun{label_einbegreifen}: \mathbb{N}^{2}\to \mathbb{N},$ and $\newfun{label_presupongamos}:\mathbb{N}\to \mathbb{N}$
	such that if $r,a\in \mathbb{N},$ $H$ is an $A$-apex $h$-grid, where $A\subseteq V(H)$ has size at most $a$, $h\geq \funref{label_impercepbbly}(r,a)+2\cdot \funref{label_presupongamos}(r),$ and
each vertex of $A$ has at least $\funref{label_einbegreifen}(r,a)$ neighbors in the central $\funref{label_impercepbbly}(r,a)$-grid of $H- A$, then $H$ contains as an $A$-fixed minor a complete $A$-apex $r$-grid.

Moreover,  $\funref{label_impercepbbly}(r,a)=\mathcal{O}(r^{4}\cdot 2^a),$
	$\funref{label_einbegreifen}(r,a)=\mathcal{O}(r^{6}\cdot 2^{a}),$ and $\funref{label_presupongamos}(r)=\mathcal{O}(r^{2}).$
\end{proposition}

The following easy observation intuitively states that every planar graph $H$ is a minor of a big enough grid, where the relationship between the size of the grid and $|V(H)|$ is linear (see e.g.,~\cite{RobertsonST94quic}).
\begin{proposition}\label{label_verbindungen}
	There exists a function $\newfun{label_lebeziatnikov}:\mathbb{N}\to\mathbb{N}$ such that every planar graph on $n$ vertices is a minor of the
	$\funref{label_lebeziatnikov}(n)$-grid. Moreover, $\funref{label_lebeziatnikov}(n)=\mathcal{O}(n).$
\end{proposition}

We now prove that, if $G$ contains a complete $A$-apex grid as an $A$-fixed minor, then $A$ intersects any solution $S$, and that the partial modification of $A$ decreases the size of the graph.

\begin{lemma}\label{label_beschaffenheit}
There exists a function $\newfun{label_understanding}: \mathbb{N} \to \mathbb{N}$ such that the following holds.

Let
$\mathcal{F}$ be a finite collection of graphs, $\Lcal$ be a hereditary R-action, and $k\in\mathbb{N}$.
Let $G$ be a graph, $S'\subseteq V(G)$ be a set of size at most $k$, and $(H_2',\phi')\in\Lcal(G[S'])$.
Suppose that $G':=G^{S'}_{(H_2',\phi')}$ contains a complete $A$-apex $\funref{label_understanding}(k)$-grid $H$
as an $A$-fixed minor for some $A\subseteq V(G')$ with $|A|= a_{\mathcal{F}}.$

Then,
for every solution $(S,H_2,\phi)$ of \apb for $(G,S',H_2',\phi')$,
it holds that $A'\ne\emptyset$, where $A':=(S\setminus S')\cap A$, and that
 $|\phi^+(A')|<|A'|$.

Moreover $\funref{label_understanding}(k)=\mathcal{O}(\sqrt{(k+{a_{\mathcal{F}}}^2+1)}\cdot {s_{\mathcal{F}}}).$
\end{lemma}

\begin{proof}
Let $d=\funref{label_lebeziatnikov}(s_\Fcal-a_\Fcal)$ and $r=\big\lceil \sqrt{(k+a_\Fcal^{2}+1)}\cdot d\big\rceil.$
We set $\funref{label_understanding}(a_\Fcal,s_\Fcal,k)=r$
and we notice that, since $d=\mathcal{O}(s_\Fcal),$ it holds that $\funref{label_understanding}(k) =\mathcal{O}(\sqrt{(k+a_\Fcal^2+1)}\cdot s_\Fcal).$
\smallskip

Observe that since $r=\big\lceil \sqrt{(k+a_\Fcal^{2}+1)}\cdot d\big\rceil,$ $V(H\setminus A)$ can be partitioned
into $(k+a_\Fcal^{2} + 1)$ vertex sets $V_1, \ldots, V_{k+a_\Fcal^{2} + 1}$ such that, for every $i\in[k+a_\Fcal^{2} + 1],$ the graph $H[V_i]$ is a
$d$-grid.\smallskip

Let $\{S_v\mid v\in V(H)\}$ be a model of $H$ in $G'$.
Let $S^*:=S\setminus S'\cup\phi^+(S')$.
There exists a pair $(H_2^*,\phi^*)\in\Mcal(G'[S^*])$ such that $G^*:=G^S_{(H_2,\phi)}=G'^{S^*}_{(H_2^*,\phi^*)}$.
Note in particular that, given that $(H_2',\phi')$ is the restriction of $(H_2,\phi)$ to $S'$, it implies that $\phi^*|_{\phi^+(S')}={\sf id}_{\phi^+(S')}$ and that $\phi^*|_{V(G')\setminus\phi^+(S')}=\phi|_{V(G')\setminus\phi^+(S')}$.
Given that $|S^*|\le|S|\le k$, there is $I\subseteq [k+a_\Fcal^2+1]$ of size at least $a_\Fcal^2+1$ such that for each $i\in I$, for each $v\in V_i$, $S_v\cap S^*=\emptyset$.
Hence, $H':=H[A\cup\bigcup_{i\in I} V_i]$ is a minor of $G'$ such that $V(H')\cap S^*\subseteq A$.

\begin{figure}[h]
\center
\includegraphics[scale=0.78]{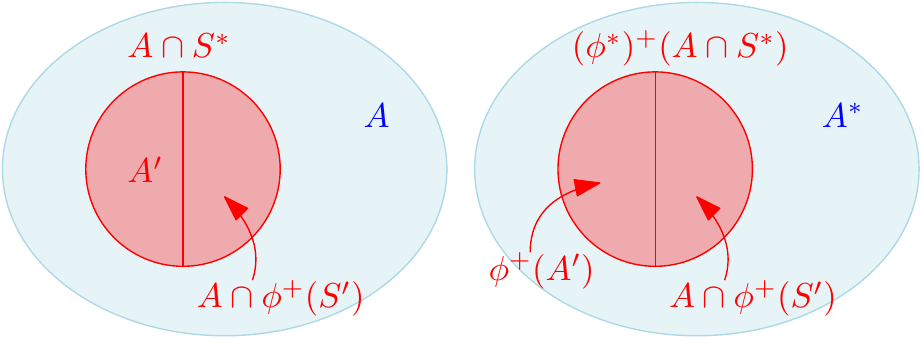}
\caption{The sets $A$ and $A^*$ in the proof of \autoref{label_beschaffenheit}.}
\label{fig_apex}
\end{figure}

\smallskip
Suppose toward a contradiction that $|\phi^+(A')|=|A'|$.
Let $A^*:=(A\setminus S^*)\cup(\phi^*)^+(S^*)$ be the set of $G^*$ obtained from $A$ after the modification $(H_2^*,\phi^*)$.
Given that $|A^*|-|A|=|\phi^+(A')|-|A'|$ (see \autoref{fig_apex} for an illustration),
this implies that $|A^*|=|A|$.
Intuitively, this means that the modification $(H_2^*,\phi^*)$ only deleted or added edges of $G'[A]$, but no vertices where deleted and no two vertices were identified together.
Therefore, the graph $H^*$ obtained from $H'$ by removing the edges between the vertices of $A$ is a minor of $G^*$.

\smallskip
Let $F$ be a graph in $\Fcal$ of apex number $a_\Fcal$.
We fix $i\in I$.
Given that $d=\funref{label_lebeziatnikov}(s_\Fcal-a_\Fcal),$ \autoref{label_verbindungen} implies that  every planar graph on $s_\Fcal-a_\Fcal$ vertices is a minor of $H[V_i]$.
Additionally, given that
$|I\setminus\{i\}|=a_\Fcal^{2}$
and that, for each $j\in I\setminus\{i\}$, $H''[V_j\cup A]$ is a complete $A$-apex $d$-grid,
for each pair of vertices in $A,$ we can find a path connecting them through some $H''[V_j\cup A],$
and thus
it implies that
every graph on $a_\Fcal$ vertices is a minor of $H''-V_i$.
Therefore, $F$ is a minor of $H''$, and hence of $G^*$, a contradiction.
\end{proof}

We finally prove the main result of the section.

\begin{proof}[Proof of \autoref{lem_obl}]
Let $r:=\funref{label_understanding}(k)$, $h=\funref{label_un}(k):=\funref{label_impercepbbly}(r,a_\Fcal)+2\cdot \funref{label_presupongamos}(r)+2$, $q=\funref{label_deux}(k):=\funref{label_einbegreifen}(r,a_\Fcal)$, and $p=\funref{label_trois}(k):=\funref{label_impercepbbly}(r,a_\Fcal)$.
Note that $r=\Ocal_{s_\Fcal}(k^{1/2})$, and thus that $h=\Ocal_{s_\Fcal}(k^2)$, $q=\Ocal_{s_\Fcal}(k^3)$,  and $p=\Ocal_{s_\Fcal}(k^2)$.
Let $G'_{\tilde{\Qcal}}$ be the graph obtained from $G'$ by contracting each bag $Q$ of $\tilde{\Qcal}$ to a single vertex $v_Q$.
Then $G'_{\tilde{\Qcal}}$ contains an $A$-apex $(h-2)$-grid $H$ as a subgraph such that each vertex of $A$ has at least $\funref{label_einbegreifen}(r,a_\Fcal)$ neighbors in the central $\funref{label_impercepbbly}(r,a_\Fcal)$-grid of $H-A$.
Therefore, $G'$ contains $H$ as an $A$-fixed minor.
By \autoref{label_emporteroient}, given that $h-2=\funref{label_impercepbbly}(r,a_\Fcal)+2\cdot \funref{label_presupongamos}(r)$, $q=\funref{label_einbegreifen}(r,a_\Fcal)$, and $p=\funref{label_impercepbbly}(r,a_\Fcal)$, it implies $G$ contains a complete $A$-apex $r$-grid as an $A$-fixed minor.
Therefore, we can conclude using \autoref{label_beschaffenheit}, given that $r=\funref{label_understanding}(k)$.
\end{proof}

\vspace{-1mm}

\section{The case of bounded treewidth}\label{sec_tw}

In this section, we present a dynamic programming algorithm in the case where the input graph has bounded treewidth.
To describe our dynamic programming algorithm, we need a particular type of tree decompositions, namely nice (rooted) tree decompositions.

\subparagraph{Nice tree decompositions.}
A \emph{rooted tree decomposition} is a triple $(T,\beta,r)$ where $(T,\beta)$ is a tree decomposition and $(T,r)$ is a rooted tree, that is, a tree in which a vertex $r$ is distinguished as its root.
A \emph{nice tree decomposition} of a graph $G$ is a rooted tree decomposition $(T,\beta,r)$ such that:
\begin{itemize}
	\item every node has either zero, one, or two children,
	\item if $x$ is a leaf of $T$, then $\beta(x)=\emptyset$ ($x$ is a \emph{leaf node}),
	\item if $x$ is a node of $T$ with a single child $y$, then $|\beta(x)\setminus\beta(y)|=1$ ($x$ is an \emph{introduce node}) or $|\beta(y)\setminus\beta(x)|=1$ ($x$ is a \emph{forget node}), and
	\item if $x$ is a node with two children $x_1$ and $x_2$, then $\beta(x)=\beta(x_1)=\beta(x_2)$ ($x$ is a \emph{join node}).
\end{itemize}

To compute a tree decomposition of a graph of bounded treewidth, we use the recent single-exponential $2$-approximation algorithm for treewidth of Korhonen \cite{Korhonen21asin}.

\begin{proposition}[\!\!\cite{Korhonen21asin}]\label{@naturalism}
There is an algorithm that, given an $n$-vertex graph $G$ and an integer $k$,
outputs either a report that $\tw(G)>k$, or a tree decomposition of $G$ of width at most $2k+1$ with $\Ocal(n)$ nodes.
Moreover, this algorithm runs in time $2^{\Ocal(k)} \cdot n$.
\end{proposition}

To find a nice tree decomposition from a given a tree decomposition, we use the following well-known result proved, for instance, in~\cite{AlthausZ21opti}.

\begin{proposition}[\!\!\cite{AlthausZ21opti}]\label{@estclusire}
Given an $n$-vertex graph $G$ and a tree decomposition $(T,\beta)$ of $G$ of width $w$, there is an algorithm that computes a nice tree decomposition of $G$ of width $w$ with at most $\Ocal(w\cdot n)$ nodes in time $\Ocal(w^2\cdot (n+|V(T)|))$.
\end{proposition}

Our dynamic programming algorithm essentially goes as follows.
Let $G$ be the input graph of treewidth $w$.
By \autoref{@naturalism} and \autoref{@estclusire}, we compute a nice tree decomposition $\Tcal=(T,\beta,r)$ of $G$ of width at most $2w+1$ with $\Ocal(w\cdot n)$ nodes in time $2^{\Ocal(k)}\cdot w^2\cdot n$.
Let $(S,H_2,\phi)$ be a solution of {\sc $\Lcal$-R-$\exc(\Fcal)$} for $(G,k)$.
Then, for each node of $T$, in a leaf-to-root manner, we guess the restriction $(H_2',\phi')$ of $(H_2,\phi)$ to the graph $G_t$ induced by the subtree of $T$ rooted at $t$.
That is, each time we introduce a vertex $v$, we guess whether $v$ belongs to $S$ or not, and if we guess that it does, we also guess how it is modified: it can either be deleted ($\phi'(v)=\emptyset$), or identified to a vertex ($\phi'(v)=\phi'(u)$ for some $u\in V(G_t)\setminus\{v\}$), or it can be a new vertex in $H_2'$ (when $\phi'^{-1}(\phi'(v))=\{v\}$).
In this latter case, we also need to guess the edges between $\phi'(v)$ and $u\in V(H_2')\setminus\{\phi'(v)\}$ to get $H_2'$.
Obviously, each time, we need to check that the guessed partial solution $(S',H_2',\phi')$ is such that $|S'|\le k$ and that $(G_t)_{(H_2',\phi')}^{S'}\in\exc(\Fcal)$, otherwise we reject this guess.
After the dynamic programming, we add a post-processing step to keep only the set $\Acal$ of guessed solutions $(S,H_2,\phi)$ for the root such that $(H_2,\phi)\in\Lcal(G[S])$.
Hence, we have a \yes-instance if and only if $\Acal\ne\emptyset$.
For ease of notation, we only formally write our dynamic programming algorithm for {\sc $\Lcal$-R-$\exc(\Fcal)$}.
It can be easily adapted to the annotated version \apb, by simply rejecting tuples that do not follow the annotation during the introduce and the join operations.

The number of possible partial solutions generated by the above description is too big to store them all, given that there is $\binom{n}{k}$ choices for the set $S$ of vertices involved in the modification.
Therefore, we instead store a ``signature'' of $G_t$ that keeps only the necessary information to ensure that each partial solution is represented by an element of the signature, that each element of the signature represents at least one partial solution, and that the signature at a node $t\in V(T)$ can be deduced from those of its children.
In order to bound the number of elements in a signature while still being able to check that a guessed partial solution $(S',H_2',\phi')$ at node $t$ is such that $(G_t)_{(H_2',\phi')}^{S'}\in\exc(\Fcal)$, we use the representative-based technique of \cite{BasteST23hittIV}.
Let ${\bf G}_t$ be a boundaried graph with underlying graph $G_t$ and  the bag $\beta(t)$ as its boundary.
This technique essentially guarantees that, if $G_t\in\exc(\Fcal)$, then we can replace ${\bf G}_t$ with a boundaried graph ${\bf R}$ with the same boundary but of smaller size, called representative of ${\bf G}_t$, such that, for any boundaried graph ${\bf H}$, ${\bf H}\oplus{\bf G_t}\in\exc(\Fcal)$ if and only if ${\bf H}\oplus{\bf R}\in\exc(\Fcal)$ (see below for the definition of $\oplus$).
In our case, $G_t$ does not necessarily belong to $\exc(\Fcal)$, but we know that $(G_t)_{(H_2',\phi')}^{S'}$ does.
Therefore, we remember a representative $\bf R$ of $G':=(G_t)_{(H_2',\phi')}^{S'}$ in the signature.
As said above, when introducing a vertex $v$, we may guess that $v\in S'$ and that $\phi'(v)=\phi'(u)$ for some $u\in V(G_t)\setminus \{v\}$.
Therefore, we need to remember each vertex of $G'$ that is a modified vertex (that is in $(\phi')^+(S')$).
Thus, the boundaried graph that we consider is a boundaried graph ${\bf G}'$ with underlying graph $G'$ and boundary $\beta(t)\cup(\phi')^+(S')$.
Moreover, to remember which vertex of the boundary is in $(\phi')^+(S')$, we reserve them the labels in $[k]$. So $H_2$ is the graph induced by the vertices with such labels in $R$.
Additionally, to be able to check that $|S'|\le k$ and that $(H_2',\phi')\in\Lcal(G[S'])$, we remember the graph $H_1':=G[S']$.
Finally, to be able to construct the extension of $H_1'$ when adding a new vertex $v$ in $S$, we remember the vertices of $S'$ that may be adjacent to $v$, that is $S_B:=S\cap \beta(t)$.
Therefore, the signature of ${\bf G}_t$ is  the set of all such $({\bf R},H_1',\phi',S_B)$.
See \autoref{subsec_sig} for a formal definition of the signature, and \autoref{fig_signature} for an illustration.
The way to construct a signature from its children for leaf (in this case, without children), forget, introduce, and join nodes is explained in \autoref{subsec_op}.

\subparagraph{Minors of boundaried graphs.} We say that a $t$-boundaried graph $G_1=(G_1,B_1,\rho_1)$ is a \emph{minor}
of a $t$-boundaried graph $G_2=(G_2,B_2,\rho_2)$, denoted by $G_1\preceq G_2$, if there is a sequence of removals
of non-boundary vertices, edge removals, and edge contractions in $G_2$, not allowing contractions
of edges with both endpoints in $B_2$, that transforms $G_2$ to a boundaried graph that is isomorphic
to $G_1$ (during edge contractions, boundary vertices prevail). Note that this extends the usual
definition of minors in graphs without boundary.

\subparagraph{Equivalent boundaried graphs and representatives.} We say that two boundaried graphs ${\bf G}_1=(G_1,B_1,\rho_1)$ and ${\bf G}_2=(G_2,B_2,\rho_2)$ are \emph{compatible} if $\rho_1(B_1)=\rho_2(B_2)$ and $\rho_2^{-1}\circ\rho_1$ is an isomorphism from $G_1[B_1]$ to $G_2[B_2]$.
Given two compatible boundaried graphs ${\bf G}_1=(G_1,B_1,\rho_1)$ and ${\bf G}_2=(G_2,B_2,\rho_2)$, we define ${\bf G}_1\oplus {\bf G}_2$ as the graph obtained if we take the disjoint union of $G_1$ and $G_2$ and, for every $i\in\rho_1(B_1)$, we identify vertices $\rho_1^{-1}(i)$ and $\rho_2^{-1}(i)$.
Given $h\in\bN$, we say that two boundaried graphs ${\bf G}_1$ and ${\bf G}_2$ are $h$-\emph{equivalent}, denoted by ${\bf G}_1\equiv_h {\bf G}_2$, if they are compatible and, for every graph $H$ with detail at most $h$ and every boundaried graph ${\bf F}$
compatible with ${\bf G}_1$ (hence, with ${\bf G}_2$ as well), it holds that
\[ H\preceq {\bf F}\oplus {\bf G}_1 \Longleftrightarrow H\preceq {\bf F}\oplus {\bf G}_2.\]
Note that $\equiv_h$ is an equivalence relation on $\cal B$. A minimum-sized (in terms of number of vertices) element of an equivalent class of $\equiv_h$ is called \emph{representative} of $\equiv_h$. For $t\in\bN$, a \emph{set of $t$-representatives} for $\equiv_h$, denoted by $\Rcal_h^t$, is a collection containing a minimum-sized representative for each equivalence class of $\equiv_h$ restricted to ${\cal B}^t$.\medskip

The following results were proved by Baste, Sau,  and Thilikos~\cite{BasteST23hittIV} and give a bound on the size of a representative and on  the
number of representatives for this equivalence relation, respectively.

\begin{proposition}[\!\!\cite{BasteST23hittIV}]\label{@iinelstaai}
For every $t\in\bN$, $q,h\in\bN_{\geq 1}$, and $\textbf{G}=(G,B,\rho)\in\Rcal_h^t$, if $G$ is does not contain $K_q$ as a minor, then $|V(G)|=\Ocal_{q,h}(t)$.
\end{proposition}

\begin{proposition}[\!\!\cite{BasteST23hittIV}]\label{@encounters}
For every $t\in\bN_{\geq 1}$, $|\Rcal_h^t|=2^{\Ocal_h(t\log t)}$.
\end{proposition}

Moreover, given a boundaried graph of bounded size, the following result gives an algorithm to find its representative.

\begin{proposition}[\!\!\cite{MorelleSST24fast}]\label{@objectives}
Given a finite collection of graphs $\Fcal$, $h,t,k\in\bN$, the set $\Rcal$ of representatives in $\Rcal_h^t$ whose underlying graphs are $\Fcal$-minor-free, there is an algorithm that, given a $t$-boundaried graph ${\bf G}$ with $k$ vertices whose underlying graph is $\Fcal$-minor-free, outputs in time $2^{\Ocal_{{\ell_\Fcal},h}(t\log t+\log(k+t))}$ the representative of ${\bf G}$ in $\Rcal$.
\end{proposition}

\subsection{Signature}\label{subsec_sig}

Let $\Fcal$ be a finite collection of graphs, $\Lcal$ be a R-action, and $k,w\in\bN$.
Let $\Rcal$ be the set of representatives in $\Rcal_{\ell_\Fcal}^{k+w}$ whose underlying graphs are $\Fcal$-minor-free.

\begin{figure}[h]
\center
\includegraphics[scale=0.8]{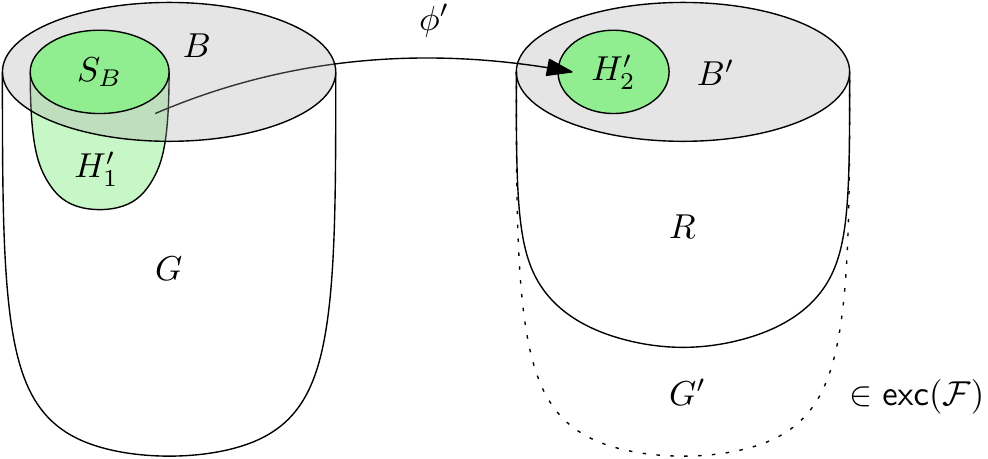}
\caption{An element $({\bf R}=(R,B',\rho'),H_1,\phi,S_B)$ in the signature of ${\bf G}=(G,B,\rho)$.}
\label{fig_signature}
\end{figure}

Let ${\bf G}=(G,B,\rho)$ be a $w$-boundaried graph with no label in $[k]$.
We call \emph{signature} of $\bf G$ the set of all tuples $({\bf R}=(R,B',\rho'),H_1',\phi',S_B)$ such that there exists a set $S'\subseteq V(G)$ of size at most $k$, there exists a graph $H_2'$ such that $(H_2',\phi')\in\Mcal(H_1')$ and $G':=G^{S'}_{(H_2',\phi')}\in\exc(\Fcal)$, and there exists an injection $\varphi:(\phi')^+(S')\mapsto[k]$ that is such that:
\begin{itemize}
\item $G[S']=H_1'$ and $R[(\phi')^+(S')]=H_2'$,
\item $S_B=S'\cap B$,
\item $B'=(B\setminus S_B)\cup (\phi')^+(S')$,
\item $\rho'$ is the function such that $\rho'|_{B\setminus S_B}=\rho|_{B\setminus S_B}$ and $\rho'|_{(\phi')^+(S')}=\varphi$, and
\item $\bf R$ is the representative in $\Rcal$ of $(G',B',\rho')$.
\end{itemize}
See \autoref{fig_signature} for an illustration.

\medskip
Let us give an upper bound on the number of tuples $({\bf R},H_1',\phi',S_B)$ in the signature of ${\bf G}$.
By \autoref{@encounters}, there are $2^{\Ocal_{\ell_\Fcal}((k+w)\log(k+w))}$ choices for $\bf R$.
Given that $H_1'$ has at most $k$ vertices, there are at most $2^{\binom{k}{2}}$ choices for $H_1'$ and at most $k^k$ choices for $\phi:V(H_1')\to V(H_2')\cup\{\emptyset\}$ (if $|V(H_1')|=|V(H_2')|$, then $\phi: V(H_1')\to V(H_2)'$ must be a bijection, and otherwise $|V(H_1')|\ge|V(H_2')\cup\{\emptyset\}|$).
Finally, there are $\binom{w}{\le k}$ choices for $S_B$.
Hence, the number of tuples is at most $2^{\binom{k}{2}+\Ocal_{\ell_\Fcal}((k+w)\log(k+w))}$.

\subsection{Dynamic programming}\label{subsec_op}

Let $G$ be a graph and let $\Tcal=(T,\beta,r)$ be a nice tree decomposition of $G$ of width $w$.
Let $\rho_0:V(G)\to[|k+1,k+V(G)|]$ be a bijection.
For $t\in V(T)$, we define by \emph{$G_t$} the graph induced by the subtree of $T$ rooted at $t$ and
by \emph{$\bf G_t$} the boundaried graph $(G_t,\beta(t),\rho_t)$, where $\rho_t:=\rho_0|_{\beta(t)}$.

Note that, for each element $({\bf R}=(R,B,\rho),H_1,\phi,S_B)$ of the signature of ${\bf G}_r$, given that $G_r=G$,
there is $S\subseteq V(G)$ of size at most $k$ such that $H_1=G[S]$ and $G^S_{(H_2,\phi)}\in\exc(\Fcal)$, where $H_2:=R[\rho^{-1}([k])]$.
Therefore, $(G,k)$ is a \yes-instance of {\sc $\Lcal$-R-$\Fcal$} if and only if there is an element $({\bf R},H_1,\phi,S_B)$ in the signature of ${\bf G}_r$ such that $(H_2,\phi)\in\Lcal(H_1)$.
This can be checked in time $2^{2\binom{k}{2}+\Ocal_{\ell_\Fcal}((k+w)\log(k+w))}$.

We want to build the signature of ${\bf G}_t$, $t\in V(T)$, in a bottom-up fashion.
Let $G_\emptyset$ be the graph with no vertices, and $\bf G_\emptyset$ be the corresponding boundaried graph.
Let us assume that $\Fcal$ does not contain $G_\emptyset$, since in that case the problem is trivial.
Let ${\sf Rep}$ be the algorithm of \autoref{@objectives}.

\subparagraph{Leaf nodes.} Let $t$ be a leaf of $T$.
Given that $\beta(t)=\emptyset$,
the signature of ${\bf G}_t$ is the singleton containing the tuple $({\bf G}_{\emptyset},G_\emptyset,\emptyset\to\emptyset,\emptyset)$.
Constructing the signature for a leaf node takes time $\Ocal(1)$.

\begin{figure}[h]
\center
\includegraphics[scale=0.8]{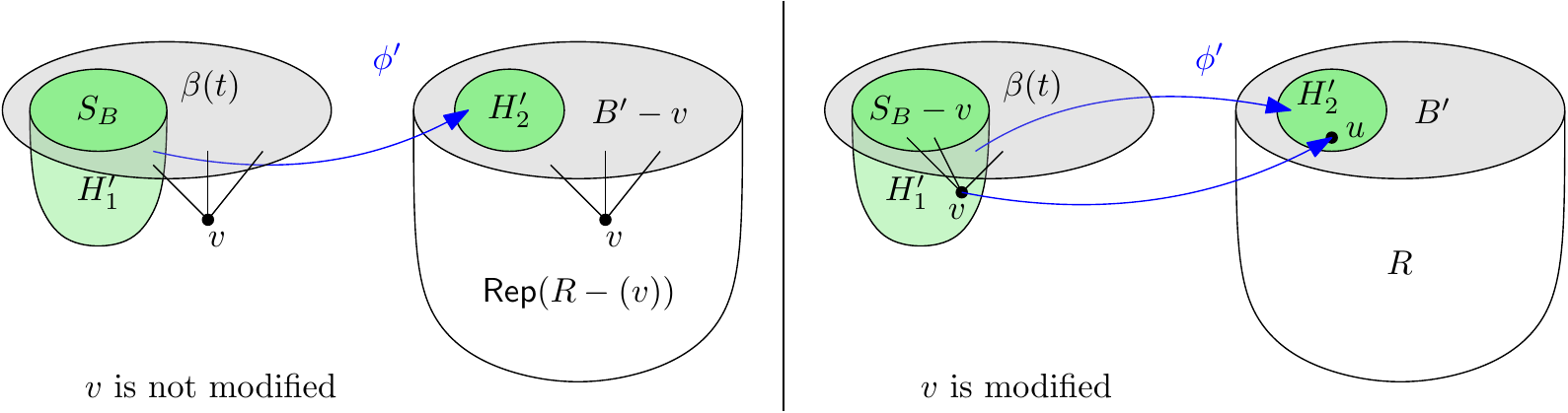}
\caption{Forgetting a vertex $v$.}
\label{fig_forget}
\end{figure}

\subparagraph{Forget nodes.}
When we forget a vertex $v$, we simply remove it from the boundary of ${\bf R}$ if it is not in the partial solution, i.e., in $S_B$.
Otherwise, if it is in the partial solution, we still need to remember it, so it remains in the boundary of $\bf R$. However, we remove it from $S_B$ as it does not belong to the current bag anymore.
See \autoref{fig_forget} for an illustration.
More formally, we do as follows.

Let $t\in V(T)$ be a forget node of $T$. Let $t'$ be the child of $t$ and $v\in\beta(t')\setminus\beta(t)$ be the forgotten vertex.
The signature of ${\bf G}_{t}$ is the set constructed by adding,
for each tuple $({\bf R}=(R,B',\rho'),H_1',\phi',S_B)$ of the signature of ${\bf G}_{t'}$,
the following tuple:
\begin{itemize}
\item {\em($v$ is not part of the modification)} if $v\in S_B$, then $({\sf Rep}({\bf R}-(v)),H_1',\phi',S_B)$, where ${\bf R}-(v):=(R,B'\setminus\{v\},\rho'|_{B'\setminus\{v\}})$,
\item {\em ($v$ is part of the modification)} otherwise, $({\bf R},H_1',\phi',S_B\setminus\{v\})$.
\end{itemize}
Given that ${\bf R}\in\Rcal\subseteq\Rcal_{\ell_\Fcal}^{k+w}$ and that $R$ does not contain $K_{s_\Fcal}$ as a minor, by \autoref{@iinelstaai}, $|V(R)|=\Ocal_{\ell_\Fcal}(k+w)$. Thus, by \autoref{@objectives}, ${\sf Rep}({\bf R}-(v))$ can be computed in time $2^{\Ocal_{\ell_\Fcal}((k+w)\log(k+w))}$.
Given that the signature of ${\bf G}_{t'}$ has at most $2^{k^2+\Ocal_{\ell_\Fcal}((k+w)\log(k+w))}$ elements, we conclude that constructing the signature of ${\bf G}_t$ takes time $2^{k^2+\Ocal_{\ell_\Fcal}((k+w)\log(k+w))}$.

\subparagraph{Introduce nodes.}
When we introduce a vertex $v$, we guess whether $v$ belongs to $S'$ or not, and if we guess that it does, we also guess how it is modified: it can either be deleted ($\phi'(v)=\emptyset$), or identified to a vertex ($\phi'(v)=\phi'(u)$ for some $u\in V(G_t)\setminus\{v\}$), or it can be a new vertex in $H_2'$ (when $\phi'^{-1}(\phi'(v))=\{v\}$).
In this latter case, we also need to guess the edges between $\phi'(v)$ and $u\in V(H_2')\setminus\{\phi'(v)\}$ to get $H_2'$.
If $v$ is not part of the modification, then we simply add $v$ to ${\bf R}$ (and to its boundary given that $v$ is in the current bag).
If $v$ is modified, then it is added to $H_1'$, and we need to check that the obtained graph $H_1'+v$ has at most $k$ vertices.
In the case when $v$ is either deleted or identified to another vertex, $H_2'$ does not change, but $v$ must be additionally mapped by $\phi'$ to $\emptyset$ in the first case, or to a modified vertex, that is, a vertex in $\rho^{-1}([k])$, in the second case. When $v$ is deleted, $\bf R$ does not change either, but when $v$ is identified, we need to add the edges between $\phi'(v)$ and the vertices adjacent to $v$ in $G_t$.
Otherwise, $v$ is mapped to a new vertex, in which case we add a new vertex to $H_2'$, and by extension, to ${\bf R}$, we guess its adjacencies to $H_2'$, and we guess its label in $[k]$.
In any case, we need to check
that the modified graph is indeed in $\exc(\Fcal)$.
See \autoref{fig_intro} for an illustration.
More formally, we do as follows.

\begin{figure}[h!tb]
\center
\includegraphics[scale=0.85]{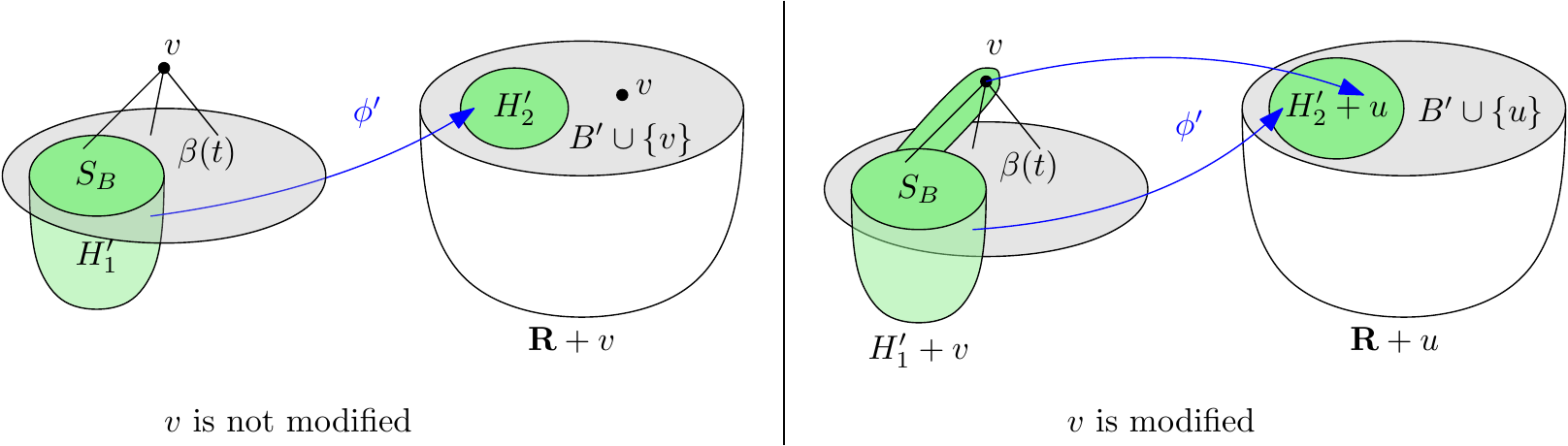}
\caption{Introducing a vertex $v$.}
\label{fig_intro}
\end{figure}

Let $t\in V(T)$ be an introduce node of $T$. Let $t'$ be the child of $t$ and $v\in\beta(t)\setminus\beta(t')$ be the introduced vertex.
Let $E_t\subseteq E(G[\beta(t)])$ be the set of edges $uv$ for $u\in\beta(t)$.
The signature of ${\bf G}_{t}$ is the set constructed by adding,
for each tuple $({\bf R}=(R,B',\rho'),H_1',\phi',S_B)$ of the signature of ${\bf G}_{t'}$,
the following tuples:
\begin{itemize}
\item {\em($v$ is not part of the modification)} if ${\bf R}+v\in\exc(\Fcal)$, then $({\sf Rep}({\bf R}+v),H_1',\phi',S_B)$, where
\begin{itemize}
\item ${\bf R}+v:=(R+v,B'\cup\{v\},\rho'\cup(v\mapsto \rho_t(v))$ and
\item $R+v$ is the graph with vertex set $V(R)\cup\{v\}$ and edge set the union of $E(R)$ and, for each edge $uv\in E_t$, the edge $uv$ if $u\notin S_B$ or the edge $\phi'(u)v$ if $u\in S_B$ and $\phi'(u)\ne 0$,
\end{itemize}
\item {\em ($v$ is deleted)} if $|V(H_1')|\le k-1$, then $({\bf R},H_1'+v,\phi'\cup(v\mapsto\emptyset),S_B\cup\{v\})$, where
\begin{itemize}
\item $H_1'+v$ is the graph with vertex set $V(H_1')\cup\{v\}$ and edge set the union of $E(H_1')$ and the edges $uv\in E_t$ for each $u\in S_B$,
\end{itemize}
\item {\em ($v$ is identified to a vertex $u$ that is in the partial solution)} if $|V(H_1')|\le k-1$, then, for each $u\in\rho^{-1}([k])$, $({\sf Rep}({\bf R}'),H_1'+v,\phi'\cup(v\mapsto u),S_B\cup\{v\})$, where
\begin{itemize}
\item ${\bf R}'$ is obtained from ${\bf R}$ by adding an edge $uw$ for each $w\in\beta(t)\setminus S_B$ such that $vw\in E(G)$ and
\item $H_1'+v$ is defined as above.
\end{itemize}
\item {\em ($v$ is part of the modification but not deleted nor identified to another vertex in the partial solution)}  for each $i\in [k]$ such that $\rho'^{-1}(i)=\emptyset$, for each $(H_2'+u_i,\phi'\cup(v\to u_i))\in\Mcal(H_1'+v)$ whose restriction to $V(H_1')$ is $(H_2',\phi')$, if $|V(H_1')|\le k-1$ and ${\bf R}+u_i\in\exc(\Fcal)$,
then $({\sf Rep}({\bf R}+u_i),H_1'+v,\phi'\cup(v\mapsto u_i),S_B\cup\{v\})$, where
\begin{itemize}
\item ${\bf R}+u_i:=(R+u_i,B'\cup\{u_i\},\rho'\cup(u_i\mapsto i))$,
\item $R+u_i$ is the graph with vertex set $V(R)\cup\{u_i\}$ and edge set the union of $E(R)$ and the edges $uu_i$ for each edge $uu_i\in E(H_2'+u_i)$,
\item $H_2':=R[(\phi')^+(V(H_1'))]$, and
\item $H_1'+v$ is defined as above.
\end{itemize}
\end{itemize}
As proved in the forget case, $|V(R)|=\Ocal_{\ell_\Fcal}(k+w)$.
Therefore, by \cite{KorhonenPS24mino}, checking whether ${\bf R}+v\in\exc(\Fcal)$ takes time $\Ocal_{\ell_\Fcal}((k+w)^{1+o(1)})$.
And again, by \autoref{@objectives}, ${\sf Rep}({\bf R}+v)$ can be computed in time $2^{\Ocal_{\ell_\Fcal}((k+w)\log(k+w))}$.
Given that $|V(H_2')|\le|V(H_1')|\le k-1$, there are at most $2^{k-1}$ choices for $H_2'+u_i$.
Hence, given that the signature of ${\bf G}_{t'}$ has at most $2^{k^2+\Ocal_{\ell_\Fcal}((k+w)\log(k+w))}$ elements, we conclude that constructing the signature of ${\bf G}_t$ takes time $2^{k^2+\Ocal_{\ell_\Fcal}((k+w)\log(k+w))}$.

\begin{figure}[h]
\center
\includegraphics[scale=0.78]{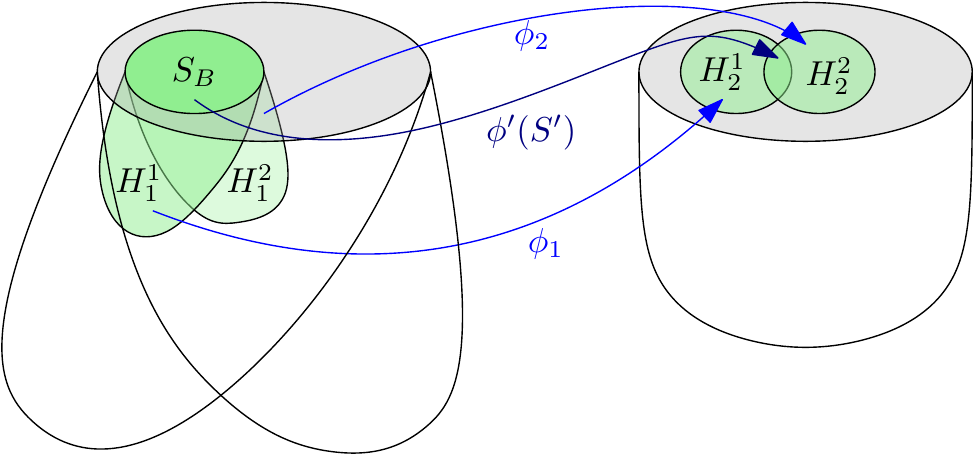}
\caption{Joining bags.}
\label{fig_join}
\end{figure}

\subparagraph{Join nodes.}
When we join two bags, we join any partial solutions of both sides that are compatible together.
They are compatible if the restriction to the current bag is the same on both sides.
Vertices on the same label in $[k]$ are vertices of the modification that are identified.
Additionally, we need to guess the edges that may be added between vertices of $H_2^1-V(H_2^2)$ and vertices of $H_2^2-V(H_2^1)$.
See \autoref{fig_join} for an illustration.
More formally, we do as follows.

Let $t\in V(T)$ be a join node of $T$. Let $t_1$ and $t_2$ be the children of $t$.
The signature of ${\bf G}_{t}$ is the set constructed by adding,
for each tuple $({\bf R_i}=(R_i,B_i,\rho_i),H_1^i,\phi_i,S_B)$ of the signature of ${\bf G}_{t_i}$, for $i\in[2]$,
each tuple $({\sf Rep}({\bf R}'), H_1,\phi,S_B)$
such that:
\begin{itemize}
\item $|V(H_1)|\le k$ and $(H_2',\phi)\in\Mcal(H_1)$, where:
\begin{itemize}
\item $H_1:={\bf H}_1^1\oplus{\bf H}_1^2$, where ${\bf H}_1^i:=(H_1^i,S_B,\rho_t|_{S_B})$ for $i\in[2]$,
\item $H_2:={\bf H}_2^1\oplus{\bf H}_2^2$, where $H_2^i:=R_i[\rho_i^{-1}([k])]$ and ${\bf H}_2^i:=(H_2^i,\phi_i^+(S_B),\rho_i|_{\phi_i(S_B)})$ for $i\in[2]$,
\item $H_2'$ is obtained from $H_2$ by adding a set $E'$ of edges between vertices of $H_2^1-\phi_1(S_B)$ and vertices of $H_2^2-\phi_2(S_B)$, and
\item $\phi_1|_{S'}=\phi_2|_{S'}$, allowing us to define that $\phi:V(H_1)\to V(H_2)$ such that $\phi|_{V(H_1^i)}=\phi_i$ for $i\in[2]$,
\end{itemize}
\item ${\bf R}'\in\exc(\Fcal)$, where
\begin{itemize}
\item ${\bf R}'$ is obtained from $({\bf R}_1'\oplus{\bf R}_2', B,\rho)\in\exc(\Fcal)$ by adding the edge set $E'$ to the underlying graph,
\item ${\bf R}_1'$ and ${\bf R}_2'$ are compatible, with $B_i':=B_i\setminus (\rho_i^{-1}([k])\setminus \phi_i(S_B))$ and ${\bf R}_i':=(R_i,B_i',\rho_i|_{B_i'})$ (informally, ${\bf R}_i'$ is obtained from ${\bf R}_i$ by removing from the boundary the vertices that are part of the modification but not in the bag $\beta(t)$),
\item $B:=B_1\cup B_2$, which is well-defined given that $B_1\cap B_2=B_1'=B_2'$ by compatibility, and
\item $\rho$ is a function such that $\rho|_{B_i}=\rho_i$ for $i\in[2]$, which is well-defined by compatibility.
\end{itemize}
\end{itemize}

The signature of each ${\bf G}_{t_i}$ has at most $2^{k^2+\Ocal_{\ell_\Fcal}((k+w)\log(k+w))}$ elements, and there are at most $2^{k^2}$ choices for $E'$, so constructing the signature of ${\bf G}_t$ takes time $2^{\Ocal_{\ell_\Fcal}(k^2+(k+w)\log(k+w))}$.

\subparagraph{Running time.}
Given that each step takes time $2^{\Ocal_{\ell_\Fcal}(k^2+(k+w)\log(k+w))}$ and the the tree decomposition has $\Ocal(w+n)$ nodes by \autoref{@estclusire}, we conclude that the dynamic programming algorithm takes time $2^{\Ocal_{\ell_\Fcal}(k^2+(k+w)\log(k+w))}\cdot n.$

Remark that, while the dynamic programming algorithm solve here the decision problem, it suffices to apply standard backtracking to obtain a solution in case of a \yes-instance.

\vspace{-2mm}

\section{Conclusion}
\label{sec_conclusions}
\vspace{-1mm}

For a large family of graph modification problems involving a bounded number of vertices,
if the target class $\Hcal$ is minor-closed, we provided an algorithm solving the problem in time $2^{\poly(k)}\cdot n^2$.
This is actually the same running time as the best known running time for {\sc Vertex Deletion to $\Hcal$}~\cite{MorelleSST24fast}.
For the other graph modification problems encompassed by our result,
such as {\sc Edge Deletion to $\Hcal$}, {\sc Edge Contraction to $\Hcal$}, {\sc Vertex Identification to $\Hcal$}, or {\sc Independent Set Deletion to $\Hcal$}, the only minor-closed $\Hcal$ for which an algorithm with an explicit parametric dependence in the solution size was known, to the authors' knowledge, were the class of forests and the class of union of paths.
Other problems, such as {\sc Matching Deletion to $\Hcal$},  {\sc Matching Contraction to $\Hcal$}, {\sc Induced Star Deletion to $\Hcal$}, or {\sc Subgraph Complementation to $\Hcal$}, were not even considered yet from the parameterized complexity viewpoint, other than in the meta-theorem of~\cite{SauST25para}.

The degree of $\poly(k)$ in the running time comes from the irrelevant vertex technique and is quite huge.
In the bounded genus case, we reduce the running time to $2^{\Ocal(k^{9})}\cdot n^2$ thanks to some improvement on the irrelevant vertex technique.
This does not match the parametric dependence in the running time of $2^{\Ocal(k^2\log k)}\cdot n^{\Ocal(1)}$ for {\sc Vertex Deletion to $\Hcal$}~\cite{KociumakaP19dele} for $\Hcal$ of bounded genus, though we possibly have a better dependence on $n$.
To the authors' knowledge, this is the first bounded genus result with an explicit parametric dependence in the solution size for the other graph modification problems encompassed by our result.

Improving more the parametric dependence in the general case would certainly require coming up with new techniques.
On the other hand, given the recent results of~\cite{KorhonenPS24mino} for minor containment, it is worth studying whether the quadratic dependence on $n$ could be improved to an almost-linear dependence while maintaining a good dependence on $k$.
Note that the approach of~\cite{KorhonenPS24mino} heavily uses Courcelle's theorem~\cite{Courcelle90them}, which would require to be translated to a plausibly very involved  dynamic programming algorithm to keep a good parametric dependence on $k$.

Given that we require the replacement action $\Lcal$ to be hereditary for our irrelevant vertex technique to work, we unfortunately restrict the graph modification problems that we solve.
For instance, {\sc Planar Subgraph Isomorphism} can be expressed as an {\sc $\Lcal$-R-Planar} problem for a specific $\Lcal$, which is not hereditary.
Hence, we do not encompass this problem in our general algorithm, while such an algorithm is provided in~\cite{FominGT19modi}, where the constraint about $\Lcal$ being hereditary is not required.
While most of the ``reasonable'' modification problems correspond to a hereditary replacement action, it is worth investigating whether our result can be extended to non-hereditary replacement actions.

Here, we only consider modifications that affect a bounded number of vertices of the input graph.
This is necessary as we want $k$ to decrease by one each time we find an obligatory vertex
(or, more precisely, as we want the size of the increasingly guessed partial solution to be bounded by $k$),
so that the depth of the branching tree is bounded.
Some relevant graph modification problems,  however, such as {\sc Elimination Distance to $\Hcal$}~\cite{MorelleSST24fast} or {\sc $\Hcal$-Treewidth}~\cite{EibenGHK21meas} (where we want to delete a vertex set $X$ whose ``torso'' has bounded treedepth or treewidth, respectively, such that $G-X\in\Hcal$),
consider a modification that affects a set of vertices that may have {\sl unbounded} size.
In this case, the branching method does not seem applicable.
However, the irrelevant vertex technique still works, and provided that we have a dynamic programming for graphs of bounded treewidth,
an algorithm can still be designed in some cases, but with a worse parametric dependence on $k$.
This is what is done, for instance, in~\cite{MorelleSST24fast} for {\sc Elimination Distance to $\Hcal$}.
Therefore, we could consider extending the results of this paper to (some kinds of) modifications involving sets of vertices or edges of unbounded size.

\bibliographystyle{plainurl}

\end{document}